\documentclass[english]{article}
\usepackage{geometry}
\usepackage{float}
\usepackage{mathrsfs}
\usepackage{amsmath}
\usepackage{amssymb}
\usepackage{graphicx}
\usepackage{amsfonts}
\usepackage{bm}
\usepackage{subfigure}
\usepackage{epsfig}
\makeatletter

 \usepackage{amsthm}
 \usepackage{amsfonts}
 \usepackage{bm}
 \usepackage{enumerate}

\@ifundefined{definecolor}{\@ifundefined{definecolor}
 {\@ifundefined{definecolor}
 {\usepackage{color}}{}
}{}
}{}

 \newtheorem{theorem}{Theorem}[section]
 \newtheorem{lem}{Lemma}[section]
 \newtheorem{rem}{Remark}[section]
 \newtheorem{prop}{Proposition}[section]
 
\newcounter{hypA}
 \newenvironment{hypA}{\refstepcounter{hypA}\begin{itemize}
   \item[({\bf A\arabic{hypA}})]}{\end{itemize}}

\begin{document}

\begin{center}

{\Large \textbf{Static Parameter Estimation for ABC Approximations of Hidden Markov Models}}

\vspace{0.5cm}

BY ELENA EHRLICH$^{1}$, AJAY JASRA$^{2}$ \& NIKOLAS KANTAS$^{3}$ 

{\footnotesize $^{1}$Department of Mathematics,
Imperial College London, London, SW7 2AZ, UK.}\\
{\footnotesize E-Mail:\,}\texttt{\emph{\footnotesize elena.ehrlich05@ic.ac.uk}}\\
{\footnotesize $^{2}$Department of Statistics \& Applied Probability,
National University of Singapore, Singapore, 117546, SG.}\\
{\footnotesize E-Mail:\,}\texttt{\emph{\footnotesize staja@nus.edu.sg}}\\
{\footnotesize $^{3}$Department of Statistical Science,
University College London, London, WC1E 6BT, UK.}\\
{\footnotesize E-Mail:\,}\texttt{\emph{\footnotesize n.kantas@ucl.ac.uk}}
\end{center}

\begin{abstract}
In this article we focus on Maximum Likelihood estimation (MLE) for the
static parameters of hidden Markov models (HMMs).
We will consider
the case where one cannot or does not want to compute the conditional
likelihood density of the observation given the hidden state because of increased computational complexity or analytical intractability. Instead we
will assume that one may obtain samples from this conditional likelihood
and hence use approximate Bayesian computation (ABC) approximations
of the original HMM. ABC approximations are biased, but the bias
  can be controlled to arbitrary precision via a parameter
  $\epsilon>0$; the bias typically goes to zero as $\epsilon \searrow
  0$.  We first establish that the bias in the log-likelihood and
  gradient of the log-likelihood of the ABC approximation, for a fixed
  batch of data, is no worse than $\mathcal{O}(n\epsilon)$, $n$ being the number of data; hence, for computational reasons, one might expect reasonable parameter estimates using such an ABC approximation. Turning to the computational problem of estimating $\theta$, we propose, using the ABC-sequential Monte Carlo (SMC) algorithm in \cite{jasra}, an approach based upon simultaneous perturbation stochastic approximation (SPSA). Our method is investigated on two numerical examples.
\end{abstract}

\section{Introduction}
Hidden Markov models provide a flexible description of a wide variety
of real-life phenomena; see \cite{cappe} for an overview. An HMM\ is a
pair of discrete-time stochastic processes, $\left\{
  X_{n}\right\}_{n\mathbb{\geq}0}$ and $\left\{ Y_{n}\right\}
_{n\geq1}$, where $X_{n}\in\mathsf{X}\subseteq\mathbb{R}^{d_{x}}$ is
an unobserved process and
$y_{n}\in\mathsf{Y}\subseteq\mathbb{R}^{d_{y}}$ is observed. The
hidden process $\left\{ X_{n}\right\} _{n\mathbb{\geq}0}$ is a\ Markov
chain with initial density $\mu_{\theta}(x_0)$ at time $0$ and
transition density $f_{\theta}\left(x_{n}|x_{n-1}\right) $, with
$\theta\in\Theta\subseteq\mathbb{R}^{d_{\theta}}$ i.e.
\begin{equation}
\mathbb{P}_{\theta}(X_{0}\in A)=\int_{A}\mu_{\theta}(x_0)dx_0\quad\text{ and }%
\quad\mathbb{P}_{\theta}(X_{n}\in A|X_{n-1}=x_{n-1})=\int_{A}f_{\theta}(x_{n}|x_{n-1}%
)dx_{n}\quad n\geq1 \label{eq:evol}%
\end{equation}
where $\mathbb{P}_{\theta}$ denotes probability,
$A\in\mathsf{B}(\mathsf{X})$ and $dx_{n}$ is Lebesgue measure.  In
addition, the observations $\left\{ Y_{n}\right\} _{n\geq1}$\
conditioned upon $\left\{ X_{n}\right\} _{n\mathbb{\geq}0}$ are
statistically independent and have marginal density
$g_{\theta}\left(y_{n}|x_{n}\right) $, i.e.%
\begin{equation}
\mathbb{P}_{\theta}(Y_{n}\in B|\{X_{k}\}_{k\geq 0}=\{x_{k}\}_{k\geq 0})=\int_{B}%
g_{\theta}(y_{n}|x_{n})dy_{n}\quad n\geq1 \label{eq:obs}%
\end{equation}
with $B\in\mathsf{B}(\mathsf{Y})$. The HMM is given by equations
(\ref{eq:evol})-(\ref{eq:obs}) and is often referred to in the
literature as a state-space model.  Here $\theta$ is a static
parameter, which is to be estimated in using MLE and online
as the data arrive; this problem has a large range of real
applications such as financial modelling or weather prediction.

Statistical inference from the class of HMMs described above is typically non-trivial. In most scenarios of practical interest one cannot calculate the likelihood:
$$
p_{\theta}(y_{1:n}) = \int g_{\theta}(y_n|x_n) \pi_{\theta}(x_n|y_{1:n-1}) dx_n
$$
where $y_{1:n}:=(y_1,\dots,y_n)$ and $\pi_{\theta}(x_n|y_{1:n-1})$ is the predictor; see e.g.~\cite{cappe} for the standard filtering recursions. Hence as the likelihood is not analytically tractable, one must resort to numerical
methods, not only to compute it, but to maximize $p_{\theta}(y_{1:n})$ w.r.t.~$\theta$. When $\theta$ is known, a popular collection of techniques for both estimating the likelihood as well as performing filtering
and smoothing are sequential Monte Carlo methods e.g.~\cite{Doucet_2001}. SMC techniques simulate a collection of $N$ samples in parallel, sequentially in time and combine importance sampling and resampling to approximate
a sequence of probability distributions of increasing state-space known pointwise up-to a multiplicative constant. These techniques provide a natural estimate of the likelihood.  The estimate
is quite well understood and is known to be unbiased \cite{delmoral} and, in addition, the relative variance is known to increase linearly with $n$ \cite{cerou1,whiteley}, $n$ being the number of data. When $\theta$ is unknown, as is the case here, estimation of $\theta$ is complicated by the path-degeneracy problem of SMC methods; e.g.~\cite{poy}. However, there are still many specialized SMC techniques which can successfully be used for online parameter estimation of HMMs in a wide variety of contexts, 
such as \cite{dds,poy}. Most of these techniques require the evaluation of $g_{\theta}(y|x)$ and potentially the gradient vectors as well. 

In this article, we consider the scenario where $g_{\theta}(y_n|x_n)$ is either intractable, in the sense that one cannot calculate it or an unbiased estimator of it, or one does not want to calculate the density,
potentially due to the high-dimensionality of $X_n$. It is assumed that one can sample from $g_{\theta}(y_n|x_n)dy_n$. In this case, one cannot use standard or the more advanced SMC methods that are mentioned above
(or indeed many other techniques) and hence exact online parameter estimation is difficult to achieve. One approach which is designed to deal with this problem are ABC techniques; see e.g.~\cite{marin}. Whilst there are a number of other
competitors \cite{gauchi}, we focus upon ABC ideas; see \cite{gauchi,jasra} for some discussion of the relative merits of ABC against competing methods. 
In the context of HMMs there has been some work on the construction of ABC approximations of HMMs \cite{jasra,mckinley}, computational techniques for filtering and smoothing \cite{jasra,martin,calvet} and their statistical consistency for parameter estimation \cite{dean,dean1}. 
ABC approximations of HMMs are biased, but the bias can be controlled to arbitrary precision via a parameter $\epsilon>0$; the bias
typically goes to zero as $\epsilon \searrow 0$.
At present there is not a methodology which can achieve our objective of online parameter estimation. In this article we do the following:
\begin{enumerate}
\item{Investigate the bias in the log-likelihood and the gradient of the log-likelihood that is induced by the ABC approximation for a fixed data set.}
\item{Develop an SMC approach with cost $\mathcal{O}(N)$ that allows one to estimate the static parameters in an online fashion.}
\end{enumerate}
In order to estimate the parameters one must obtain numerical estimates of the log-likelihood and gradient of this quantity. It is then
important to understand what happens to the bias of the ABC approximation of these latter quantities, as the time parameter (number of data, $n$) grows.
We establish, under some assumptions, that this ABC bias, for both quantites is no worse than $\mathcal{O}(n\epsilon)$; this result is
associated to the theoretical work in \cite{dean,dean1}.
These former results indicate that 
the ABC approximation is amenable to numerical implementation:
parameter estimation 
will not necessarily be dominated by the bias; we discuss why this is the case in Remarks \ref{rem:abc_approx_error} and \ref{rem:abc_approx_error1}. For 2.~we introduce an SMC approach based upon SPSA \cite{spall} 
to estimate the parameters in an online manner (see also \cite{poy1} in the context of HMMs). This methodology can be expected to `work well' when:
\begin{itemize}
\item{$d_x$ is large and $d_{\theta}$, $d_y$ are small to moderate.}
\end{itemize}
Whilst these statements are somewhat delicate (e.g.~what is large), in the scenario of high-dimensional states, it has been established in \cite{beskos1} that the \emph{simulation} error does not explode in the dimension.
As a result, the ideas here can be seen as principled competitors (and related to - see \cite{nott}) to ensemble kalman filter-based algorithms such as in \cite{frei}.

This paper is structured as follows. In Section \ref{sec:model} we discuss the model and ABC approximation. Our bias result is also given.
In Section \ref{sec:comp} our computational strategy is outlined. In Section \ref{sec:numer} the method is investigated from a computational perspective.
In Section \ref{sec:summ} the article is concluded with some discussion of future work. The proofs of our results can be found in the appendix.

\section{Model and Approximation}\label{sec:model}

\subsection{Model and Estimation}\label{sec:model1}

Consider first the joint filtering or smoothing density of the HMM
given by 
$$
\pi_{\theta}(x_{0:n}|y_{1:n})=\frac{\mu_{\theta}(x_{0})\prod_{k=1}^{n}g_{\theta}(y_{k}|x_{k})f_{\theta}(x_{k}|x_{k-1})}{\int_{\mathsf{X}^{n+1}}\mu_{\theta}(x_{0})\prod_{k=1}^{n}g_{\theta}(y_{k}|x_{k})f_{\theta}(x_{k}|x_{k-1})dx_{0:n}}
$$
where $\theta\in\Theta\subseteq\mathbb{R}^{d_{\theta}}$
is the static parameter, $x_{n}\in\mathsf{X}$ are the hidden states
and $y_{n}\in\mathsf{Y}$ the observations. This quantity can be computed
recursively using 
\begin{eqnarray}
\pi_{\theta}(x_{0:k}|y_{1:k-1}) & = & \int_{\mathsf{X}}\pi_{\theta}(x_{0:k-1}|y_{1:k-1})f_{\theta}(x_{k}|x_{k-1})dx_{k}\label{eq:filter1}\\
\pi_{\theta}(x_{0:k}|y_{1:k}) & = & \frac{g_{\theta}(y_{k}|x_{k})\pi_{\theta}(x_{0:k}|y_{1:k-1})}{p_{\theta}(y_{k}|y_{1:k-1})}\label{eq:filter2}
\end{eqnarray}
with the \emph{recursive likelihood} being 
\begin{equation}
p_{\theta}(y_{k}|y_{1:k-1})=\int_{\mathsf{X}}g_{\theta}(y_{k}|x_{k})\pi_{\theta}(x_{0:k}|y_{1:k-1})dx_{k}\label{eq:filter3}
\end{equation}

Furthermore we write the log- (marginal) likelihood at time $n$:
$$
\log(p_{\theta}(y_{1:n}))=\sum_{k=1}^{n}\log(p_{\theta}(y_{k}|y_{1:k-1})).
$$
In the context of MLE one is usually interested computing
\begin{equation*}
\hat{\theta}=\arg\max_{\theta\in\Theta}\log(p_{\theta}(y_{1:n}))\label{eq:MLE_offline}
\end{equation*}
Note that this is a batch or off-line method, which means that one
needs to wait first to collect the complete dataset and then compute
the ML estimate. For a long observation sequence the computation of
the gradient at each iteration of the algorithm can be prohibitive.
Therefore, one uses on-line
methods whereby the estimate of the parameter is updated sequentially
as the data arrives. A practical alternative would be to consider
the following update scheme at time $k$, for some sequence $\{a_k\}_{k\geq 1}$
\begin{equation*}
\theta_{k+1}=\theta_{k}+a_{k+1}\left.\nabla\log\left(p_{\theta}(y_{k}|y_{1:k-1})\right)\right\vert _{\theta=\theta_{k}}.
\end{equation*}
Upon receiving $y_{k}$, the parameter estimate is updated in the
direction of ascent of the conditional density of this new observation.
The algorithm in the present form is not suitable for on-line implementation
due to the need to evaluate the gradient of $\log p_{\theta}(y_{k}|y_{0:k-1})$
at the current parameter estimate which would require computing the
filter from time $0$ to time $k$ using the current parameter value
$\theta_{k}$.

A recursive ML (RML) algorithm bypassing this problem
has been proposed in the literature when $\mathsf{X}$ is finite in \cite{legland1997}. It relies on the following update
scheme 
\begin{equation*}
\theta_{k+1}=\theta_{k}+a_{k+1}\nabla\log\left(p_{\theta_{0:k}}(y_{k}|y_{1:k-1})\right)\label{eq:RML}
\end{equation*}
where the positive non-increasing step-size sequence $\left\{ a_{k}\right\} _{k\geq1}$
satisfies $\sum_{k}a_{k}=\infty$ and $\sum_{k}a_{k}^{2}<\infty$ 
\cite{legland1997}; e.g. $a_{k}=k^{-\alpha}$
for $0.5<\alpha\leq1$. The quantity $\nabla\log p_{\theta_{0:k}}(y_{k}|y_{1:k-1})$
is defined as 
\begin{equation*}
\nabla\log\left(p_{\theta_{0:k}}(y_{k}|y_{1:k-1})\right)=\nabla\log\left(p_{\theta_{0:k}}(y_{1:k})\right)-\nabla\log\left(p_{\theta_{0:k-1}}(y_{1:k-1})\right)\label{eq:timevaryingscore}
\end{equation*}
where the notation $\nabla\log\left(p_{\theta_{0:k}}(y_{1:k})\right)$
indicates that at each time $k$ the quantities in (\ref{eq:filter1})-(\ref{eq:filter3})
are computed using the parameter estimate $\theta_{k}$. The asymptotic
properties of this algorithm\ (i.e. the behavior of $\theta_{k}$\ in
the limit as $k$\ goes to infinity) have been studied in \cite{legland1997}
for a finite state-space HMM. It is shown that under regularity conditions
this algorithm converges towards a local maximum of the average log-likelihood;
this average log-likelihood being maximized at the `true' parameter
value. 

In this article, we would like to implement approximate versions of
these on-line and off-line ML schemes when both the following cases
hold:
\begin{itemize}
\item Case 1: {We can sample from the conditional distribution of $Y|x$,
for any fixed $\theta$ and $x$.} 
\item Case 2: {We cannot or do not want to evaluate the conditional density
of $Y|x$, $g_{\theta}(y|x)$ and do not have access to an unbiased
estimate of it.} 
\end{itemize}
Apart from using likelihoods which do not admit computable densities
such as some stable distributions, this context might appear relevant
to the context when one is interested to use SMC methods and evaluate
$g_{\theta}(y|x)$ when $d_{x}$ is large. SMC methods for \emph{filtering} do not always scale well with the dimension of the hidden state $d_x$, often
requiring a computational cost $\mathcal{O}(\kappa^{d_x})$, with $\kappa>1$; see e.g.~\cite{beskos1,bickel}.
 A more detailed discussion on the difficulties of using SMC methods
in high dimensions is far beyond the scope of this article, but we
remark the ideas in this paper can be relevant in this context.

\subsection{ABC Approximation and Noisy ABC}

To facilitate statistical inference, we consider an ABC approximation
of the joint smoothing density (e.g.~\cite{jasra,mckinley}):
$$
\pi_{\theta,\epsilon}(u_{1:n},x_{0:n}|y_{1:n}) =
\frac{\mu_{\theta}(x_0)\prod_{k=1}^n
  K_{\theta,\epsilon}(y_k|u_k)g_{\theta}(u_k|x_k)f_{\theta}(x_k|x_{k-1})}{\int_{\mathsf{X}^{n+1}\times\mathsf{Y}^{n}}
  \mu_{\theta}(x_0)\prod_{k=1}^n K_{\theta,\epsilon}(y_k|u_k)
  g_{\theta}(u_k|x_k)f_{\theta}(x_k|x_{k-1}) du_{1:n}x_{0:n}}
$$
where $u_n\in\mathsf{Y}$ are pseudo observations, $\epsilon>0$ 
$K_{\theta}:\mathsf{Y}\times\mathbb{R}_+\times\Theta\rightarrow\mathbb{R}_+\cup\{0\}$ is some kernel function that has bandwidth that depends upon a precision parameter $\epsilon>0$. Examples include:
\begin{eqnarray*}
  K_{\theta,\epsilon}(y_k|u_k) & = & \mathbb{I}_{\{u:|y_k-u|<\epsilon\}}(u_k)\\
  K_{\theta,\epsilon}(y_k|u_k) & = & \phi_{d_y}(y_k;u_k,\epsilon I_{d_x})
\end{eqnarray*}
where $\mathbb{I}$ is the indicator function, $|\cdot|$ is the $\mathbb{L}_1-$norm, $\phi_{d}(y;\xi,\Sigma)$ is normal density on $d-$dimensions with mean $\xi$ and covariance $\Sigma$ and $I_d$ is the $d-$dimensional identity matrix.

Consider the quantity, to be used below:
\begin{equation}
g_{\theta,\epsilon}(y_k|x_k) = \frac{\int_{\mathsf{Y}}  K_{\theta,\epsilon}(y_k|u_k)g_{\theta}(u_k|x_k) du_k}{\int_{\mathsf{Y}^{2}}  K_{\theta,\epsilon}(y_k|u_k)g_{\theta}(u_k|x_k) du_kdy_k}.
\label{eq:g_epsilon}
\end{equation}
Throughout the article we \emph{critically} choose
$K_{\theta,\epsilon}(y_k|u_k)$ such that the denominator of
\eqref{eq:g_epsilon} does not depend upon $x_k$ or $\theta$.
As noted in \cite{dean}, after integrating out the $u_{1:n}$, this representation leads to a new (or perturbed) HMM
with transitions $f_{\theta}$ and likelihoods $g_{\theta,\epsilon}$.
Parameter estimation associated to the smoother $\pi_{\theta,\epsilon}$ just considers the function:
$$
\log(p_{\theta,\epsilon}(y_{1:n})) = \sum_{k=1}^n \log(p_{\theta,\epsilon}(y_k|y_{1:k-1}))
$$
where
$$
p_{\theta,\epsilon}(y_k|y_{1:k-1}) = \int_{\mathsf{X}} 
g_{\theta,\epsilon}(y_k|x_{k})\pi_{\theta,\epsilon}(x_{k}|y_{1:k-1})dx_k.
$$
We term the maximizer of
$p_{\theta,\epsilon}(y_{1:n})$ as the ABC-MLE.
One can then define a RML procedure for the ABC-HMM as in Section \ref{sec:model1}:
$$
\theta_{k+1}=\theta_{k}+a_{k+1}\nabla \log\{p_{\theta,\epsilon}(y_k|y_{1:k-1})\}.
$$
In practice, one can consider an estimation of $p_{\theta,\epsilon}(y_{1:n})$ including factors independent of $\theta,\epsilon$; this is discussed in Section \ref{sec:comp}.

Results on associated to the asymptotics of the
ABC-MLE (i.e.~as $n$ grows) can be found in \cite{dean,dean1}; there is an asymptotic bias. 
In addition, in the case
of noisy ABC, where the data become corrupted, there is no asymptotic bias and one can recover the true parameter. We remark that the methodology that is considered in this article can \emph{easily} incorporate noisy ABC. However, there may be several reasons why one
may not want to use noisy ABC: (1) the consistency results (currently) depend upon the data originating from the original HMM; (2) the current simulation-based methodology may not be able to push $\epsilon$ towards zero.
For (1),  if the data do not originate from the HMM of interest, it has not been studied what happens with regards to the asymptotics of noisy ABC for HMMs. It may be that some investigators might be uncomfortable with assuming that the data originate
from the exactly the HMM being fitted. For (2) the asymptotic bias (which is under assumptions either $\mathcal{O}(\epsilon)$ or $\mathcal{O}(\epsilon^2)$ \cite{dean,dean1}) could be less than the asymptotic variance (under assumptions $\mathcal{O}(\epsilon^2)$ \cite{dean,dean1}) as
$\epsilon$ could be much bigger than 1 when using current simulation methodology. We do not use noisy ABC in this article, but acknowledge its fundamental importance
with regards to parameter estimation associated to ABC for HMMs; our approach is pragmatic, taking into account points (1)-(2). 


 \subsection{Result}\label{sec:result}

We now prove an upper-bound on the bias induced by the ABC approximation on the log-likelihood and gradient of the log-likelihood. The latter is more relevant for parameter estimation, but the mathematical
arguments are considerably more involved for this quantity, in comparison to the ABC bias of the log-likelihood. Hence the log-likelihood is considered as a simple preliminary result. These results are to be taken in the context of
ABC (not noisy ABC) and help to provide some guarantees associated to the numerics.

We consider the scenario
$$
K_{\theta,\epsilon}(y_k|u_k)  = \mathbb{I}_{A_{\epsilon,y_k}}(u_k)
$$
where the set $A_{\epsilon,y_k}$ is specified below.
Throughout $|\cdot|$ is understood to be an $\mathbb{L}_{1}-$norm. The hidden-state is assumed to lie on a \emph{compact} set, i.e.~$\mathsf{X}$ is compact.
We use the notation $\mathcal{P}(\mathsf{X})$ to denote the class of probability measures on $\mathsf{X}$ and $\mathcal{M}(\mathsf{X})$ the collection of finite and signed measures on
$\mathsf{X}$. $\|\cdot\|$ denotes the total variation distance. The initial distribution of the hidden Markov chain is written as $\mu_{\theta}\in\mathcal{P}(\mathsf{X})$.
In addition, we condition on the observed data and do not mention them in any mathematical statement of results (due to the assumptions below). We do not consider the instance of whether the data originate, or not, from a HMM. For the control of the bias of the gradient of the log-likelihood (Theorem \ref{theo:grad_ll_bias}), we assume that $d_{\theta}=1$. This is not restrictive as one can use the arguments
to prove analgous results when $d_{\theta}>1$, by considering componentwise arguments for the gradient. In addition, for the gradient result, the derivative of $\mu_{\theta}$ is
written $\widetilde{\mu_{\theta}}\in\mathcal{M}(\mathsf{X})$.
We make the following assumptions, which are extremely strong. They are made to keep the proofs as short as possible.

\begin{hypA}\label{hyp:like_cont}
\emph{Lipschitz Continuity of the Likelihood}. There exist
$L<+\infty$ such that for any $x\in\mathsf{X}$,
$y,y'\in\mathsf{Y}$, $\theta\in\Theta$
\[
|g_{\theta}(y|x) - g_{\theta}(y'|x)| \leq L |y-y'|.
\]
\end{hypA}

\begin{hypA}\label{hyp:stat}
\emph{Statistic and Metric}. The set $A_{\epsilon, y}$ is:
\[
A_{\epsilon, y} = \{u:|y-u|<\epsilon\}.
\]

\end{hypA}

\begin{hypA}\label{hyp:like_bound}
\emph{Boundedness of Likelihood and Transition}.
There exist $0<\underline{C}<\overline{C}<+\infty$ such that for all $x,x^{\prime}\in
\mathsf{X}$, $y\in\mathsf{Y}$, $\theta\in\Theta$
\begin{align*}
\underline{C}  &  \leq f_{\theta}(x^{\prime}|x
)\leq \overline{C},\\
\underline{C} &  \leq g_{\theta}(y|x)\leq \overline{C}.
\end{align*}
\end{hypA}

\begin{hypA}\label{hyp:like_grad_cont}
\emph{Lipschitz Continuity of the Gradient of the Likelihood}. 
$f_{\theta}(x'|x)$, $g_{\theta}(y|x')$
are differentiable in $\theta$ for each $x,x^{\prime}\in
\mathsf{X}$, $y\in\mathsf{Y}$. In addition,
there exist $L<+\infty$ such that for any $x\in\mathsf{X}$,
$y,y'\in\mathsf{Y}$, $\theta\in\Theta$
\[
|\nabla\{g_{\theta}(y|x)\} - \nabla\{g_{\theta}(y'|x)\}| \leq L |y-y'|.
\]
\end{hypA}

\begin{hypA}\label{hyp:like_grad_bound}
\emph{Boundedness of Gradients of the Likelihood and Transition}.  
There exist $0<\underline{C}<\overline{C}<+\infty$ such that for all $x,x^{\prime}\in
\mathsf{X}$, $y\in\mathsf{Y}$, $\theta\in\Theta$
\begin{align*}
\underline{C}  &  \leq \nabla\{f_{\theta}(x^{\prime}|x
)\}\leq \overline{C},\\
\underline{C} &  \leq \nabla\{ g_{\theta}(y|x)\}\leq \overline{C}.
\end{align*}
\end{hypA}

We first have the result on the ABC bias of the log-likelihood. The proof is in appendix \ref{app:log_like}.

\begin{prop}\label{prop:ll_bias}
Assume (A1-3). Then 
there exist a $C<+\infty$ such that for any
$n\geq 1$, $\mu_{\theta}\in\mathcal{P}(\mathsf{X})$, $\epsilon>0$, $\theta\in\Theta$ we have:
$$
|\log(p_{\theta}(y_{1:n}))-\log(p_{\theta,\epsilon}(y_{1:n}))| \leq Cn\epsilon.
$$
\end{prop} 

\begin{rem}\label{rem:abc_approx_error}
The above proposition gives some simple guarantees on the bias of the ABC log-likelihood. 
When using SMC algorithms to approximate $\log(p_{\theta}(y_{1:n}))$, the overall error will be decomposed into the deterministic bias that is present from the ABC approximation (that in Proposition \ref{prop:ll_bias}) and the numerical error of approximating the log-likelihood.
Under some assumptions, the $\mathbb{L}_2-$error of the SMC estimate of the log-likelihood should not deteriorate any faster than linearly in time; this is due to the results cited previously.
Thus, as the time parameter increases, the ABC bias of the log-likelihood will not necessarily dominate the simulation-based error that would be present even if $g_{\theta}$ is evaluated. 
\end{rem}

Proposition \ref{prop:ll_bias} is reasonably straight-forward to prove, but, is of less interest in the context of parameter estimation, as one is interested in the gradient of the log-likelihood.
We now have the result on the ABC bias of the gradient of the log-likelihood. The proof in appendix \ref{app:log_like_grad}.

\begin{theorem}\label{theo:grad_ll_bias}
Assume (A1-5). Then 
there exist a $C<+\infty$ such that for any
$n\geq 1$, $\mu_{\theta}\in\mathcal{P}(\mathsf{X})$, $\widetilde{\mu_{\theta}}\in\mathcal{M}(\mathsf{X})$, $\epsilon>0$, $\theta\in\Theta$ we have:
$$
|\nabla\{\log(p_{\theta}(y_{1:n}))\}-\nabla\{\log(p_{\theta,\epsilon}(y_{1:n}))\}| \leq Cn\epsilon (2+\|\widetilde{\mu_{\theta}}\|).
$$
\end{theorem} 

\begin{rem}\label{rem:abc_approx_error1}
The above Theorem again provides some explicit guarantees when using an ABC approximation along with SMC-based numerical methods. For example, if one can consider approximating gradients in an ABC
context (see \cite{yidrlim}), then from the results of \cite{singh}, one expects that the variance of the SMC estimates to increase only linearly in time. Again, as time increases
the ABC bias does not necessarily dominate the variance that would be present even if $g_{\theta}$ is evaluated (i.e.~one uses SMC on the true model).
\end{rem}

\begin{rem}
  The result in Theorem \ref{theo:grad_ll_bias} can be found in
  eq.~(72) of \cite{dean} and direct limit (as $\epsilon\searrow 0$)
  in \cite{dean1}. However, we adopt a new (and fundamentally different) proof technique, with a substantially clearer proof and an additional result of
  independent interest is proved.  We derive the stability w.r.t.~time
  of the bias of the ABC approximation of the filter derivative; see
  Theorem \ref{theo:filt_deriv_bias} in appendix \ref{app:filt_deriv}.
\end{rem}

 \section{Computational Strategy}\label{sec:comp}

 \subsection{SMC}

In order to perform online parameter estimation, we will need to use a SMC algorithm to approximate $p_{\theta,\epsilon}(y_k|y_{1:k-1})$ for $\theta$ fixed; this is a critical quantity that we will use below.
An algorithm which can do this is the SMC approach in \cite{jasra} which is detailed in Figure \ref{fig:abc_smc}, with proposals $\{q_{k,\theta}\}_{1\leq k \leq n}$ with density w.r.t.~Lebesgue
measure.

On the basis of Figure \ref{fig:abc_smc}, one can approximate $p_{\theta,\epsilon}(y_{1:n})$, \emph{up-to a constant that is independent of} $\theta$, as follows. 
In an abuse of notation, we denote this SMC estimate (which does not include factors that do not depend on $\theta$) as $p_{\theta,\epsilon}^N(y_{1:n})$.
The SMC estimate is
$$
p_{\theta,\epsilon}^N(y_{1:n}) = \prod_{k=1}^n \frac{1}{N} \sum_{i=1}^N \widetilde{W}_k^{(i)}
$$
with 
$$
p_{\theta,\epsilon}^N(y_k|y_{1:k-1}) = \frac{1}{N} \sum_{i=1}^N \widetilde{W}_k^{(i)}.
$$
These estimates are unbiased for any $N\geq 1$ (see \cite{delmoral}). In practice, we are interested in the log-likelihoods; taking logarithms of the above estimates generally leads to a biased
approximation of $\log\{p_{\theta,\epsilon}(y_{1:n})\}$ and $\log\{p_{\theta,\epsilon}(y_k|y_{1:k-1})\}$. One can implement a form of bias correction, using the Taylor series expansion ideas in \cite{pitt}.
Throughout, we use the bias-corrected estimates:
\begin{eqnarray}
\widehat{\log\{p_{\theta,\epsilon}(y_{1:n})\}} & = & \log\{p_{\theta,\epsilon}^N(y_{1:n})\} + \frac{1}{2N} p_{\theta,\epsilon}^N(y_{1:n})^{-2} \nonumber\\
\widehat{\log\{p_{\theta,\epsilon}^N(y_k|y_{1:k-1})\}} & = & \log\{p_{\theta,\epsilon}^N(y_k|y_{1:k-1})\} + \frac{1}{2N} p_{\theta,\epsilon}^N(y_k|y_{1:k-1})^{-2}\label{eq:bias_correct}.
\end{eqnarray}


The parameter $\epsilon$ can be computed adaptively; see \cite{jasra}. It is remarked that a drawback of this algorithm is that when $d_y$ grows with $\epsilon,N$ fixed, one cannot expect the algorithm to
work well for every $\epsilon$; typically one must increase $\epsilon$ to yield reasonable algorithmic results and this is at the cost of increasing the bias. To maintain $\epsilon$ at a reasonable level, one
must consider more advanced strategies which are not investigated here.

One final point, which is often useful in practice. One can modify the ABC approximation to:
$$
\pi_{\theta,\epsilon}(u_{1:n}^1,\dots,u_{1:n}^M,x_{0:n}|y_{1:n}) \propto
\mu_{\theta}(x_0)\prod_{k=1}^n \bigg[
\Big(\frac{1}{M}\sum_{j=1}^M K_{\theta,\epsilon}(y_k|u_k^j) \Big)\prod_{j=1}^M g_{\theta}(u_k^j|x_k)\bigg]f_{\theta}(x_k|x_{k-1})
$$
which yields the same bias as the original ABC approximation (on integrating the $u$ variables) but can yield substantial computational
improvements. This is because as $M$ grows one approximates a marginal SMC that does not sample the auxiliary $u$ variables.

\begin{figure}[h]
\begin{itemize}
\item \textsf{Step} \textsf{0.}{\ }
\textsf{For} $i=1,\dots,N$ \textsf{sample} $X_0^{(i)}$ i.i.d.~from $\mu_{\theta}(x_0)dx_0$.
\textsf{Set } $W_0^{(i)}=1/N$ \textsf{for each} $i\in\{1,\dots,N\}$.
\textsf{Set} $k=0$\textsf{.}
\par
\item \textsf{Step} \textsf{1. }
\textsf{\ Resample }%
$N$\textsf{\ particles from }%
\[
\widehat{\pi}_{k}\left(  \cdot\right)  =\sum_{i=1}^{N}W_{k}^{\left(  i\right)
}\delta_{x_{k}^{(i)}}\left(  \cdot\right)  ,
\]
\textsf{which are also denoted }$\{ x_{k}^{(i)}\} $\textsf{, and set }%
$W_{k}^{\left(  i\right)  }=\frac{1}{N}$\textsf{. Set }$k=k+1$\textsf{ and if
}$k=n+1$\textsf{, stop.}
\par
\item \textsf{Step} \textsf{2. For }$i=1,\dots,N,$\textsf{\ sample }%
$X_{k}^{(i)}$ \textsf{from }$q_{k,\theta}(x_k|x_{k-1}^{(i)})dx_k$\textsf{ and
}$U_{k}^{(i)}$\textsf{\ from the likelihood }$g_{\theta}(u_k|x_{k}^{(i)})du_k$.
\textsf{ Compute}
\[
W_k^{(i)}\propto W_{k-1}^{(i)} \widetilde{W}_k^{(i)} \quad\quad\quad
\widetilde{W}_k^{(i)} = \frac{K_{\theta,\epsilon}(y_k|u_k^{(i)})f_{\theta}(x_k^{(i)}|x_{k-1}^{(i)})}{q_{k,\theta}(x_k|x_{k-1}^{(i)})}
,\label{eq:incrementalweight}%
\]
\textsf{renormalize the weights and return to Step 1.}
\end{itemize}
\caption{SMC Algorithm for ABC target.}%
\label{fig:abc_smc}
\end{figure}

\begin{rem}\label{rem:collapsed_abc}
We note that, suppressing $\theta$, if the HMM can be written in the form:
\begin{eqnarray*}
Y_n & = & \xi_n(X_n,W_n)\quad n\geq 1\\
X_n &= & \varphi_n(X_{n-1},V_n)\quad n\geq 1
\end{eqnarray*}
where $X_0=x_0\in\mathsf{X}$ is known, $Y_n\in\mathsf{Y}$, 
$V_n\in\mathsf{X}$ with $\{V_n\}_{n\geq 1}$ i.i.d.~$W_n\in\mathsf{Y}$ with $\{W_n\}_{n\geq 1}$ i.i.d.~and independent of $\{V_n\}_{n\geq 1}$ and
$\xi_n:\mathsf{X}\times\mathsf{Y}\rightarrow\mathsf{Y}$, $\varphi_n:\mathsf{X}\times\mathsf{X}\rightarrow\mathsf{X}$.
Suppose that:
\begin{itemize}
\item{One can evaluate the densities of $W_n$ and $V_n$ and sample from the associated distributions.}
\item{One can evaluate $\xi_n$ (resp.~$\varphi_n$) pointwise, for each $n\geq 1$ and $X_n,W_n$ (resp.~$X_{n-1},V_n$).}
\end{itemize}
One can construct a `collapsed' (see \cite{murray}) ABC approximation (assuming $K_{\theta,\epsilon}(y|u) = \mathbb{I}_{A_{y,\epsilon}}(u)$, $
A_{y,\epsilon} = \{u\in\mathsf{Y}:d(u,y)<\epsilon\}$, with $d$ a distance metric on $\mathsf{Y}$) 
$$
\pi_{\epsilon}(w_{1:n},v_{1:n}|y_{1:n}) \propto \prod_{k=1}^n \mathbb{I}_{A_{y_k,\epsilon}}(\xi_k(\varphi^{(k)}(x_0,v_{1:k}),w_k))p(w_k)p(v_k).
$$
Hence a version of the SMC algorithm in Figure \ref{fig:abc_smc} can be derived which does not need to sample from the dynamics of the data. In additon one does not need access to the transition density of the hidden Markov chain.
This representation, however, does not always apply.
\end{rem}

 \subsection{SPSA}\label{sec:spsa}

Recall the RML procedure in Section \ref{sec:model1}, where $g_{\theta}(y|x)$ is not intractable:
\begin{equation}
\theta_{k+1} = \theta_k + a_{k+1}\nabla \log(p_{\theta_{0:k}}(y_k|y_{1:k-1}))
\label{eq:sa_update}
\end{equation}
for $\{a_n\}$ a sequence of step-sizes. In practice, one does not know the gradient and must resort to (e.g.) SMC techniques to approximate it; see for example \cite{poy}.
In our ABC context one can run the algorithm in Figure \ref{fig:abc_smc} to approximate the ABC filter. To recursively update $\theta$, at least using the ideas in
\cite{poy}, one has to evaluate 
\begin{equation}
\log(g_{\theta}(y|x)) \quad \textrm{and} \quad \nabla \log(g_{\theta}(y|x))
\label{eq:quantitites}
\end{equation}
which we will not have access to.

We propose the following computational scheme; the idea is to use SPSA, which does not require the quantities in \eqref{eq:quantitites}.
Introduce a decreasing sequence of positive numbers $\{c_k\}$. Suppose, with $\{a_k\}$ as in the update, \eqref{eq:sa_update}, we have
$$
\forall k, a_k>0 \quad a_k,c_k\rightarrow 0 \quad \sum_k a_k = \infty \quad \sum_k \frac{a_k^2}{c_k^2} < \infty.
$$

Start with some initial guess $\theta_0$ and perform the standard SMC update (i.e.~as in Figure \ref{fig:abc_smc})  for two sets of particles. One with parameter:
$$
\theta_0  + c_0 \Delta_0
$$
and the other with parameter:
$$
\theta_0  - c_0 \Delta_0
$$ 
where $\Delta_0$ is a $d_{\theta}-$dimensional vector with each entry $\pm 1$ Bernoulli distributed (see \cite{spall}).
For both algorithms compute $\widehat{\log(p_{\theta_0  + c_0 \Delta_0,\epsilon}(y_1))}$ and
$
\widehat{\log(p_{\theta_0  - c_0 \Delta_0,\epsilon}(y_1))} 
$ 
respectively, where the estimates are the bias-corrected versions as in equation \eqref{eq:bias_correct}.
To obtain the next parameter estimate, in the $i^{th}-$dimension, take
$$
\theta_{1,i} = \theta_{0,i} + a_1 \frac{\widehat{\log(p_{\theta_0  + c_0 \Delta_0,\epsilon}(y_1))}-\widehat{\log(p_{\theta_0  - c_0 \Delta_0,\epsilon}(y_1))}}{2c_0 \Delta_{0,i}}.
$$

At any subsequent time-point, with $\theta_k$ and perform the standard SMC update for two sets of particles. One with parameter:
$$
\theta_k  + c_k \Delta_k
$$
and the other with parameter:
$$
\theta_k  - c_k \Delta_k
$$ 
For both algorithms compute
$
\widehat{\log(p_{\theta_k  + c_k \Delta_k,\epsilon}(y_k|y_{1:k-1}))}
$
and
$
\widehat{\log(p_{\theta_k  - c_k \Delta_k,\epsilon}(y_k|y_{1:k-1}))}
$
To obtain the next parameter estimate, in the $i^{th}-$dimension, take
$$
\theta_{k+1,i} = \theta_{k,i} + a_{k+1} \frac{\widehat{\log(p_{\theta_k  + c_k \Delta_k,\epsilon}(y_k|y_{1:k-1}))}-\widehat{\log(p_{\theta_k  - c_k \Delta_k,\epsilon}(y_k|y_{1:k-1}))}}{2c_k \Delta_{k,i}}.
$$
This algorithm does not require one to evaluate $g_\theta$ or its gradient. We refer the reader to \cite{spall} and \cite{poy1} for a theoretical justification of this procedure.

\section{Numerical Simulations}\label{sec:numer}

We consider two numerical examples that are designed to investigate the accuracy and behaviour of our numerical algorithms. In order to do this, we do not consider scenarios where $g_{\theta}$ is intractable.

\subsection{Linear Gaussian Model}\label{sec:lg}

We consider the following linear Gaussian HMM, with $\mathsf{Y}=\mathsf{X}=\mathbb{R}$:
\begin{align} 
  Y_n &= X_n + \sigma_w W_n\nonumber\\
  X_n &= \phi X_{n-1} + \sigma_v V_n,\nonumber
\end{align}
with $W_n,V_n$ independent and
$W_n\stackrel{\textrm{i.i.d.}}{\sim}\mathcal{N}(0,1)$, $V_n\stackrel{\textrm{i.i.d.}}{\sim}\mathcal{N}(0,1)$. In the subsequent examples, we will use a simulated dataset obtained with $\theta=(\sigma_v,\phi,\sigma_w)=(0.2, 0.9, 0.3)$. 

\subsubsection{Offline MLE}\label{sec:lg_off}
We begin by considering a small data set, of $n=1000$ data points. The offline scenario is the one for which we can expect the best possible performance of the ABC-SMC; if we cannot obtain reasonable parameter estimates in this scenario we would not expect ABC to be useful in practice.
We are concerned with obtaining offline ABC-SMC estimates
\begin{align}
  \theta_{j+1}(i) = \theta_{j}(i) + a_{j+1} \frac{\widehat{\log(p_{\theta_j  + c_j \Delta_j}(y_{1:1000}))}-\widehat{\log(\widehat{p}_{\theta_j  - c_j \Delta_j}(y_{1:n}))}}{2c_j \Delta_{j}(i)},\nonumber
\end{align}
where $j$ is the iteration, $\theta_{j}(i)$ is the parameter estimate in the $i^{th}$-dimension, and $\Delta_{j}(i)$ is the $i^{th}$-entry of the Bernoulli distributed vector. 
For the SPSA stepsizes, we chose $c_j=j^{-0.1}$, $a_j=1$ for $j<10000$, and $a_j=(j-10000)^{-0.8}$ for $j\geq 10000$.
The iteration consists of running the ABC-SMC algorithm for 1000 data-points, with the current value of $\theta$.

In Figure \ref{fig:lg_off}, we compare offline estimates of the following cases:
\begin{enumerate}[(a)]
  \item Kalman Filter (KF) with SPSA
  \item SMC on the true model using $N=1000$, with SPSA
  \item ABC-SMC using $N=200$, $M=10$, $\epsilon=0.1$, with SPSA
  \item Maximum Likelihood estimates (MLE) from an offline grid search optimization.
\end{enumerate}
In this particular test case, we can observe good relative performance of the ABC-SMC procedure, with regards to estimating parameters. This strong performance allows us to investigate a slightly more challenging scenario.

\begin{figure}[h]
  \centering
  \includegraphics[width=0.5\textwidth]{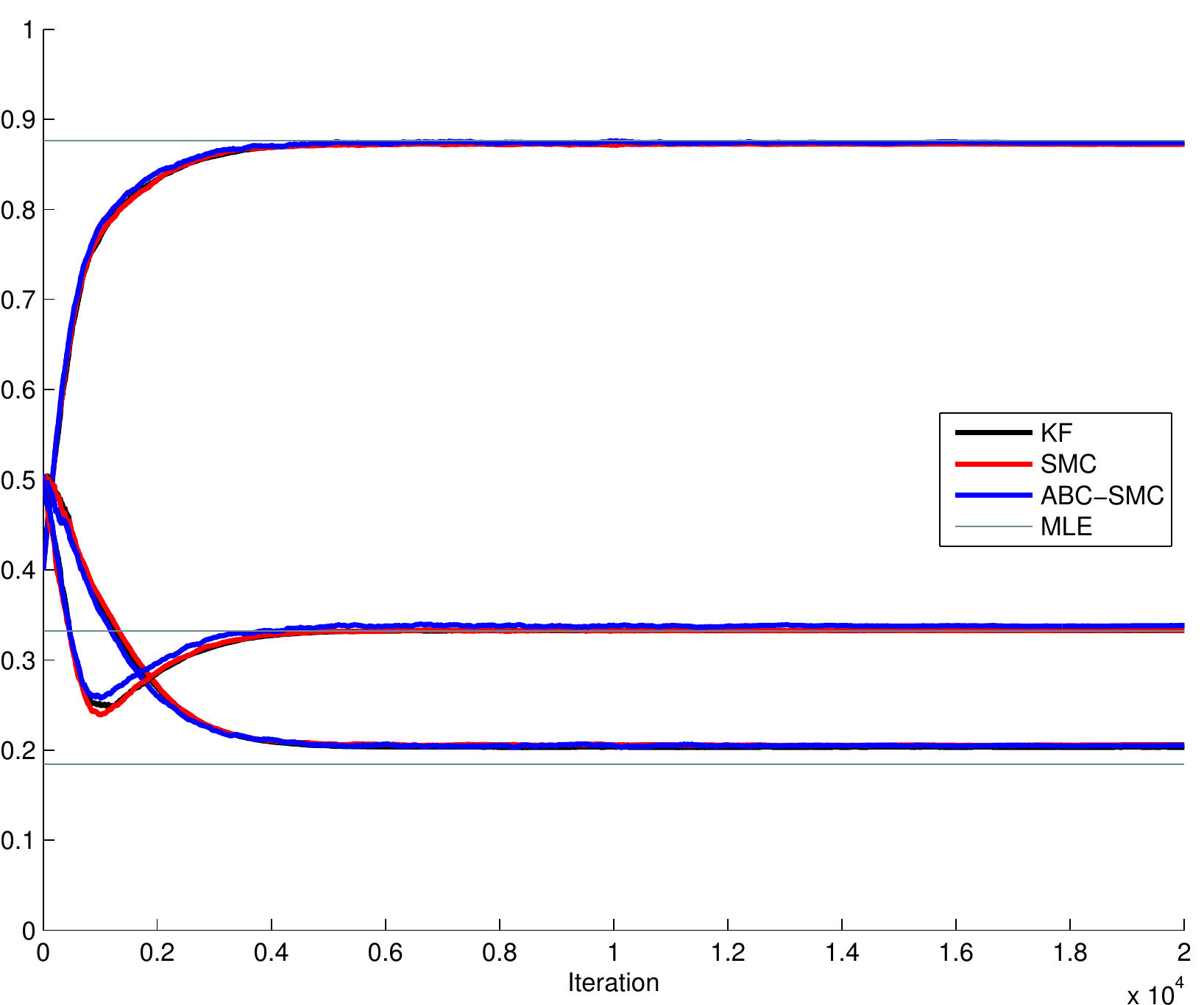}
  \caption{A typical run of the offline parameter estimates obtained by the KF, SMC, and ABC-SMC for the linear Gaussian HMM
, along with ther parameters' offline MLEs. }  
  \label{fig:lg_off}
\end{figure}

\subsubsection{Online MLE}\label{sec:lg_on}

We now consider a larger data set with $n=50,000$ data points, simulated with the previously indicated parameter values. We use the online SPSA method described in Section \ref{sec:spsa}. The SMC (i.e.~on the true model) and ABC-SMC algorithms were employed with the same $N$ (and $M$, $\epsilon$ for ABC-SMC) as in the offline case, and the SPSA sequences are similar to their offline forms, in Section \ref{sec:lg_off}.

We ran fifty independent runs of the each algorithm considered in the previous Section. In Figure \ref{fig:lg_on}, we plot the medians and credible intervals for the 25-75\% and 5-95\% percentiles of the parameter estimates (across the independent runs). The $\widehat{\theta}_k$ converge after $k=20000$ time steps, with the KF and SMC yielding similarly valued estimates. We observe increased variance from left to right in Figure \ref{fig:lg_on}, which we attribute to the randomness of SMC and ABC-SMC respectively. In particular, the expected reduced accuracy of ABC-SMC
against SMC is apparent, but, the bias does not appear to be substantial (for ABC-SMC) in this particular example.

\begin{figure}[h]
  \centering
  \subfigure[Kalman]{\includegraphics[width=0.32\textwidth]{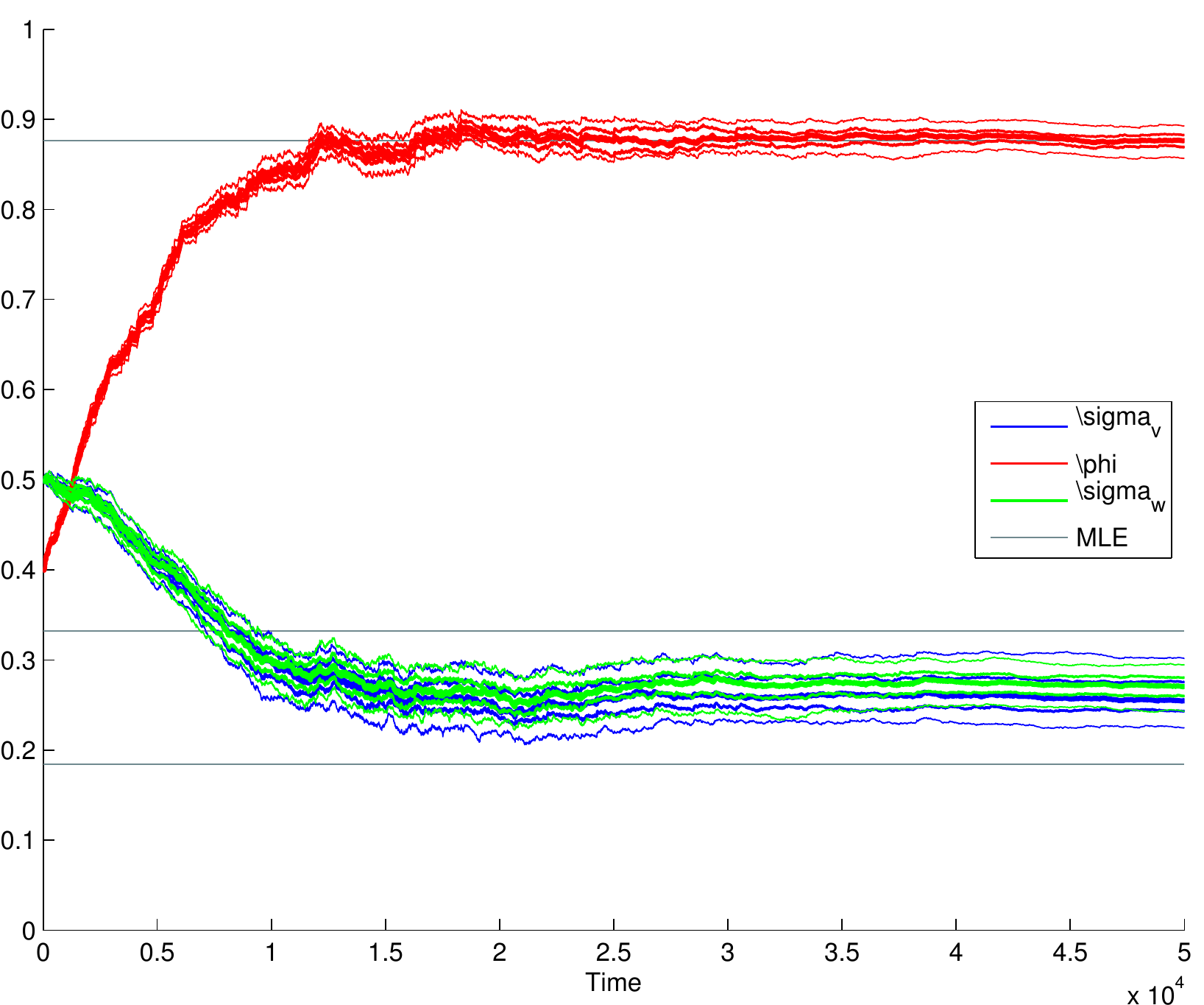}\label{fig:lg_on_kf}}
  \subfigure[Sequential Monte Carlo]{\includegraphics[width=0.32\textwidth]{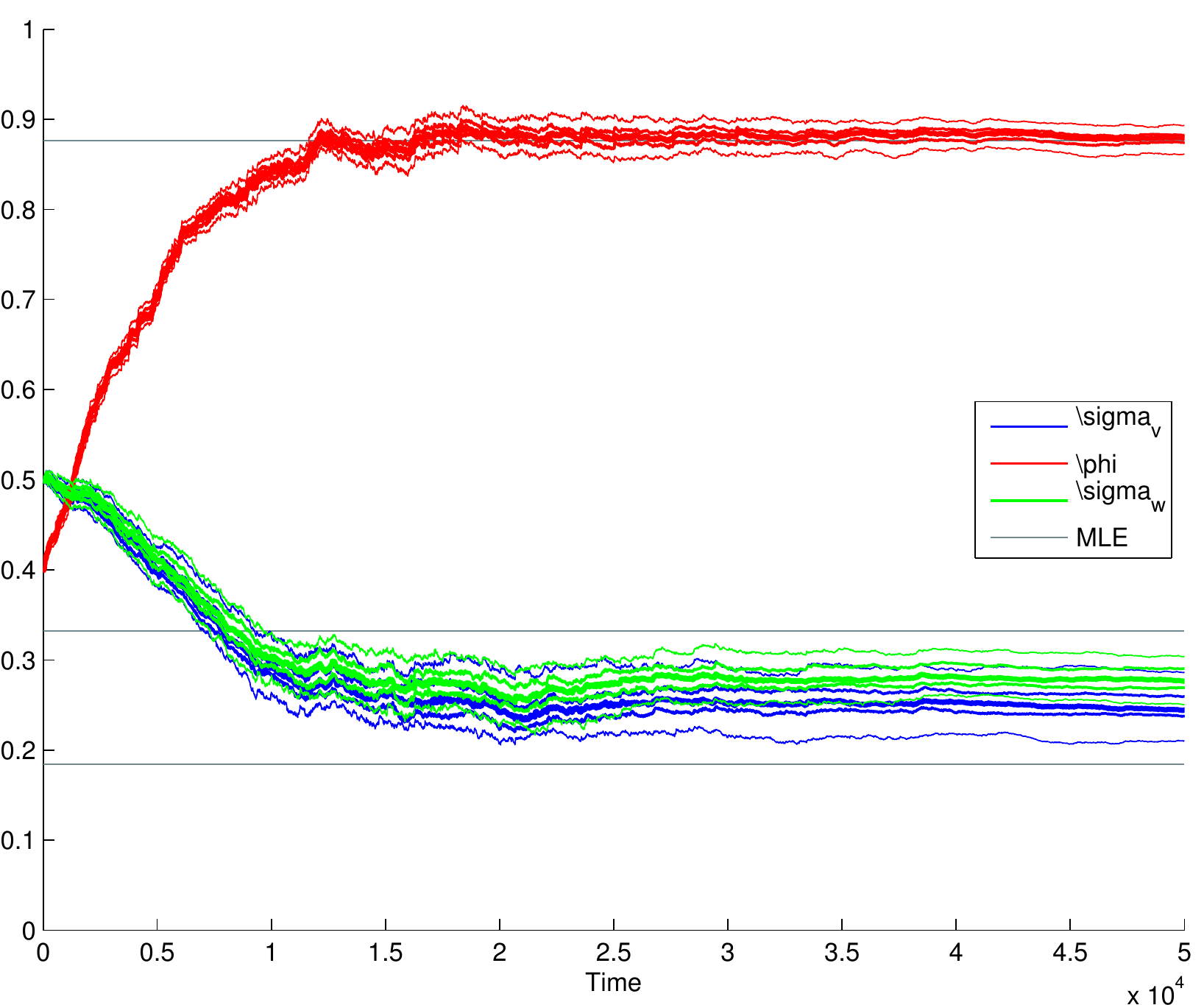}\label{fig:lg_on_pf}}
  \subfigure[SMC-ABC]{\includegraphics[width=0.32\textwidth]{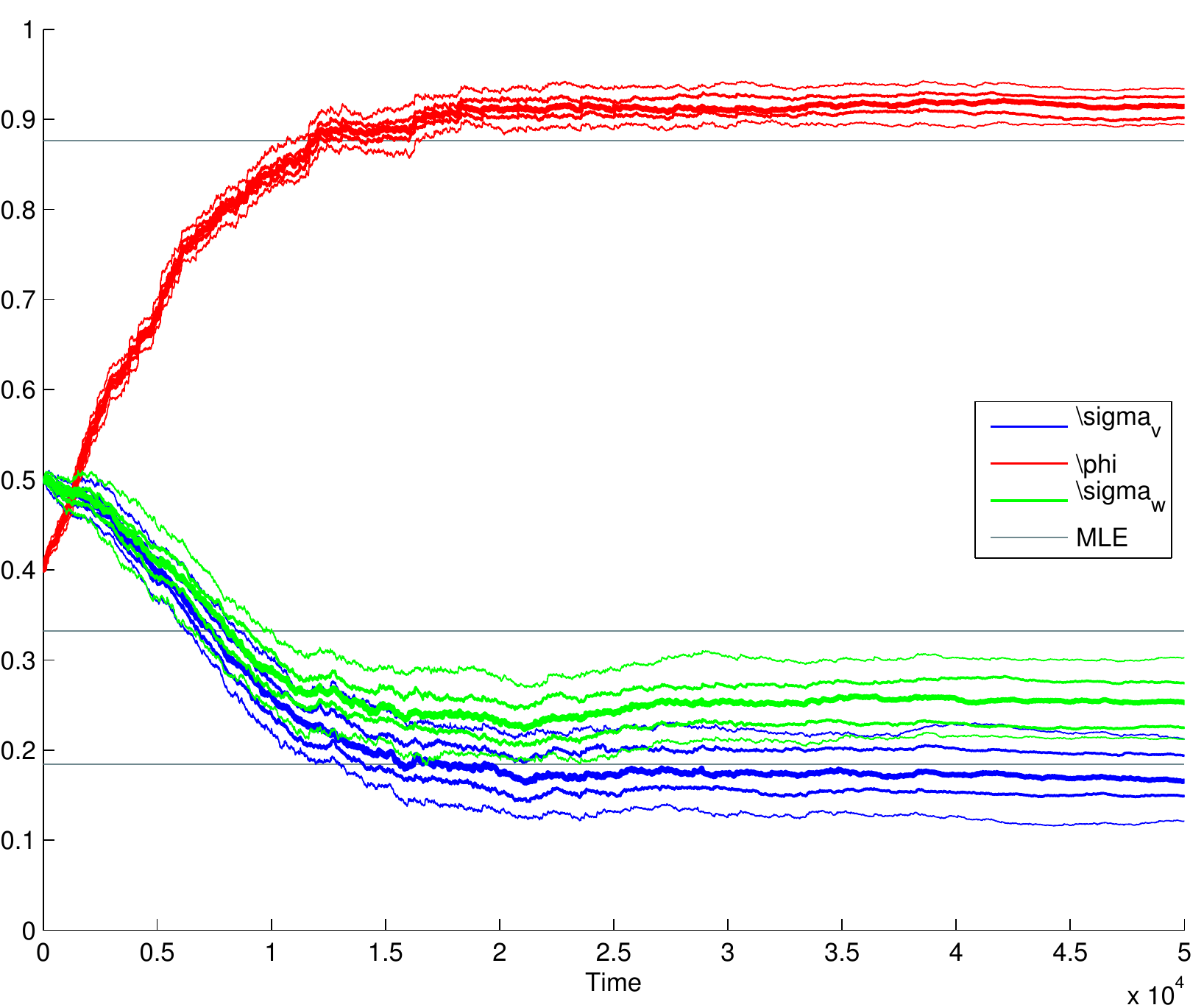}\label{fig:lg_on_abc}}
  \caption{Credible intervals for the 5-95\% and 25-75\% percentiles, and the medians for multiple runs of online parameter estimates streamed by the KF, SMC, and ABC-SMC for the linear Gaussian HMM
. }  
  \label{fig:lg_on}
\end{figure}

\subsection{Lorenz '63 Model} \label{sec:l63}

\subsubsection{Model and Data}

We now consider the following non-linear state-space model with $\mathsf{X}=\mathsf{Y}=\mathbb{R}^3$. The original model is such that hidden process evolves deterministically according to the
Lorenz '63 system of ordinary differential equations,
\begin{align}
    \frac{dX_t(1)}{dt} &= \sigma_{63}\big(X_t(2)-X_t(1)\big)\nonumber\\
    \frac{dX_t(2)}{dt} &= \rho X_t(1) - X_t(2) - X_t(1) X_t(3)\nonumber\\
    \frac{dX_t(3)}{dt} &= X_t(1) X_t(2) - \beta X_t(3).\nonumber
\end{align}
where we recall that the arguments $X_t(j)$ are the $j^{th}-$dimension at time $t$; where $t$ is continuous here. We modify the model to one such that the hidden process is a discrete-time Markov chain with stochastic dynamics:
$$
X_n=f_n(X_{n-1}) +V_n, \quad n\geq 1
$$
where $f_n$ is the $4^{th}$-order approximation Runge Kutta solution to the Lorenz '63 system, $V_{n}\stackrel{\textrm{i.i.d.}}{\sim}\mathcal{N}(0,\tau I_{d_x})$ and $X_0$ is taken as known. Here $\tau$ is used to represent the 
time-discretization.


For the observations:
\begin{align}
  Y_n &= H X_n + Q W_n,\nonumber\quad n\geq 1
\end{align}
where $W_n\stackrel{\textrm{i.i.d.}}{\sim}\mathcal{N}(0,I_{d_y})$, $W_n$ is independent of $V_n$ and $Q$ is the Cholesky root of a Toeplitz matrix defined by the parameters $\kappa$ and $\sigma$ as follows:
\begin{align}
  Q_{ij} &= \sigma S\left( \kappa^{-1} \min(|i-j|,d_y - |i-j|) \right), \quad i,j\in\{1,\ldots,d_y\}\nonumber\\
  S(z) &= \left\{
    \begin{array}{l l}
      1 - \frac{3}{2}z + \frac{1}{2}z^3, & 0\leq z\leq 1\\
      0, & z>1
    \end{array}\right. ,\nonumber
\end{align}
and
\begin{align}
  H_{ij} &= \left\{
    \begin{array}{l l}
      \frac{1}{2}, & i=j\\
      \frac{1}{2}, & i=j-1\\
      0, & i\neq j
    \end{array}\right. .\nonumber
\end{align}
When  $\theta=(\kappa,\sigma,\sigma_{63},\rho,\beta)=(2.5,2,10,28,\frac{8}{3})$, $n=5000$ and $\tau=0.05$, 
a visualisation of the Lorenz '63 (hidden) dynamics is shown in Figure \ref{fig:state} and the associated simulated dataset in \ref{fig:data}.
\begin{figure}[ht]
  \centering
  \subfigure[Hidden Markov $x_{1:5000}$]{\includegraphics[width=0.33\textwidth]{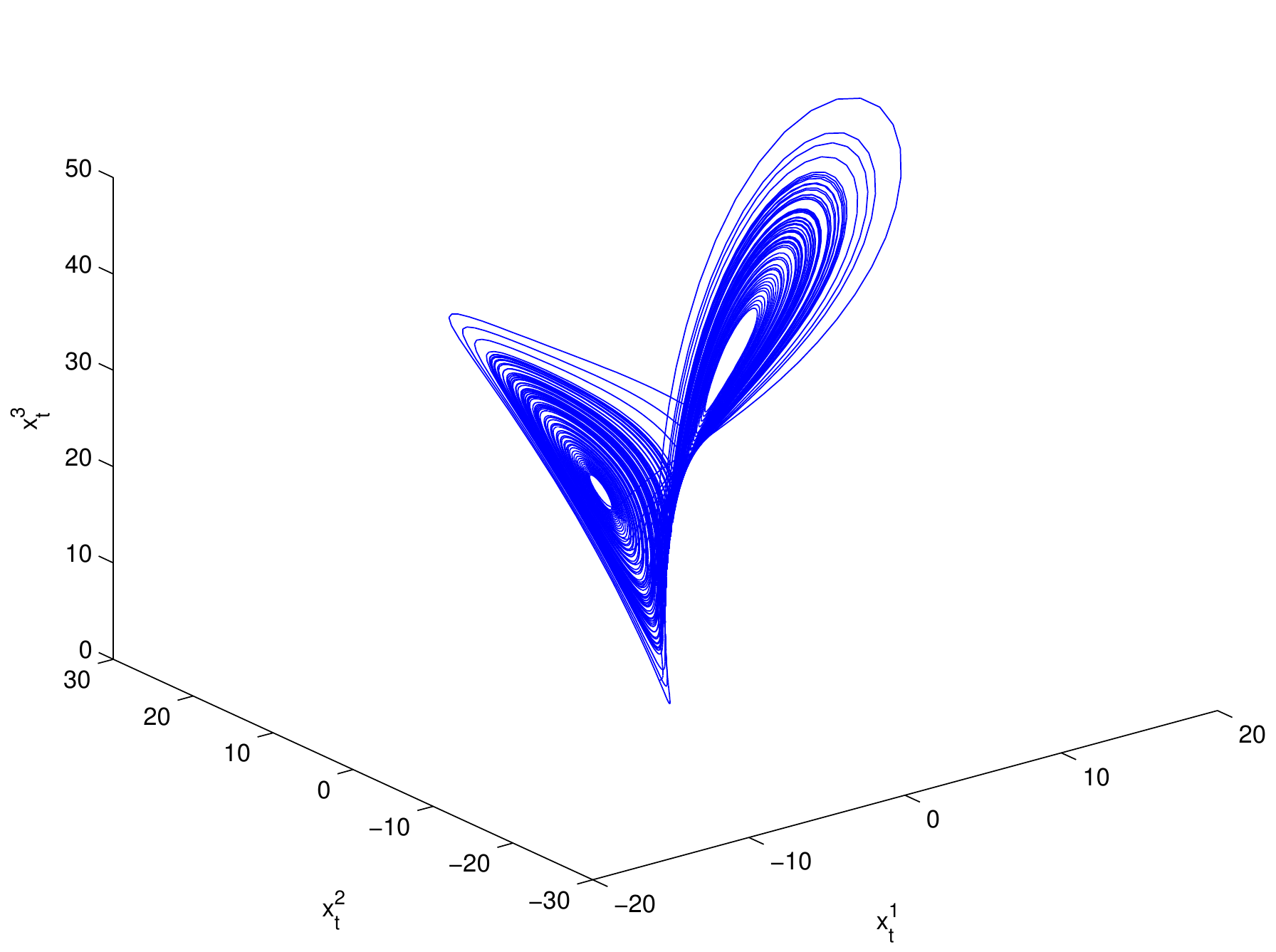}\label{fig:state}}
  \subfigure[Observed data $y_{1:5000}$]{\includegraphics[width=0.33\textwidth]{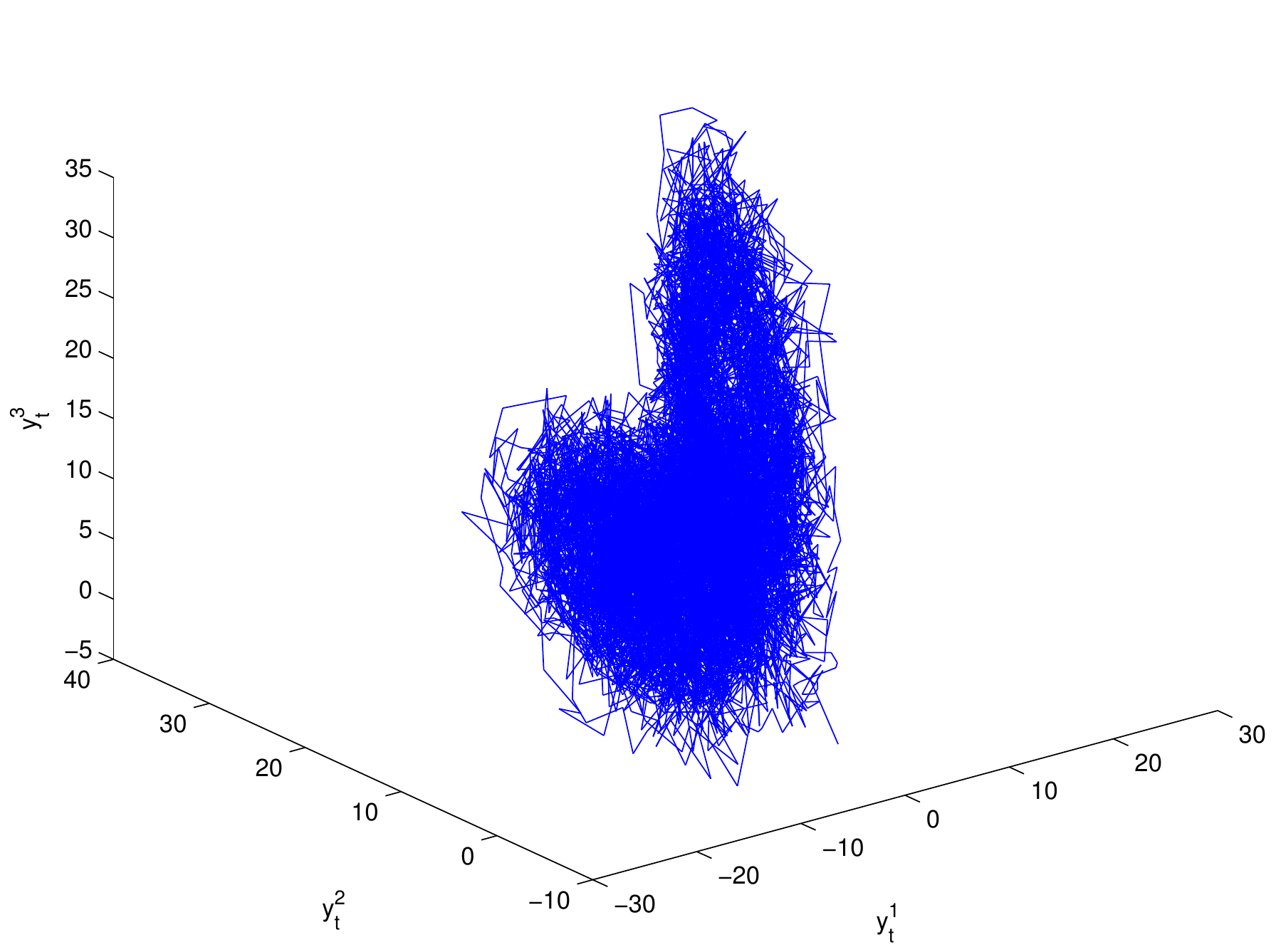}\label{fig:data}}
  \caption{Evolution of the 3-dimensional Lorenz '63 HMM in Section \ref{sec:l63} }
  \label{fig:lorenz}
\end{figure}

For the simulated dataset in Figure \ref{fig:data}, we use ABC-SMC to obtain online parameter estimates for $\theta$ and we study the performance of these estimates under different settings. We will use $\widehat{\theta}^{N,M}_{\epsilon,n}$ to denote the estimate of $\theta$
at time $n$,
 that was estimated using $N$ particles, $M$ pseudo-observations and a Gaussian kernel with covariance $\epsilon I_{d_y}$. We will compare the behaviour of the algorithm as each of $N, M, n, \epsilon$ varies. 

\subsubsection{Numerical Results}

We now examine the performance of the algorithm with $N\in\{100, 1000, 10000\}$. For each value of $N$, we ran fifty independent runs of ABC-SMC, using $M=10$ and $\epsilon=1$. In Figures \ref{fig:l63_boxplot_kappa_N}-\ref{fig:l63_boxplot_rho_N} we plot boxplots of the terminal parameter estimates, $\widehat{\theta}^{N,10}_{1,5000}$, against their true values marked by dotted green lines. In Figures \ref{fig:l63_biasvar_kappa_N}-\ref{fig:l63_biasvar_rho_N} we plot the absolute value of the Monte Carlo (MC) bias (that is, the absolute difference between the estimate and true value), in red, and the MC standard deviation, in blue. The MC bias and standard deviation points are fitted with least-squares curves proportional to $\frac{1}{\sqrt{N}}$, the standard MC rates with which the accuracy of the estimates is expected to improve. 
With regards to the variability of the estimates one sees the expected reduction in variability as $N$ increases. The bias is harder to quantify; it will not necessarily be the case that as $N$ grows the bias falls. This is because there is a Monte Carlo bias (from the SMC), an optimization bias
(from the SPSA), an approximation bias (from the ABC) and the fact that the data have been generated from the model (so the true static parameters might not be exact). Increasing $N$ can only deal with the SMC bias (which for estimates with parameters fixed is $\mathcal{O}(N^{-1})$), but the addition of parameter estimation again does not make it easy to understand what happens here. The main point is simply as expected; one obtains significantly more reproducible/consistent results as $N$ grows.

\begin{figure}[h]
  \centering
  \subfigure[$\hat{\kappa}^{N,10}_{1,5000}$]{\includegraphics[width=0.24\textwidth,,height=5.5cm]{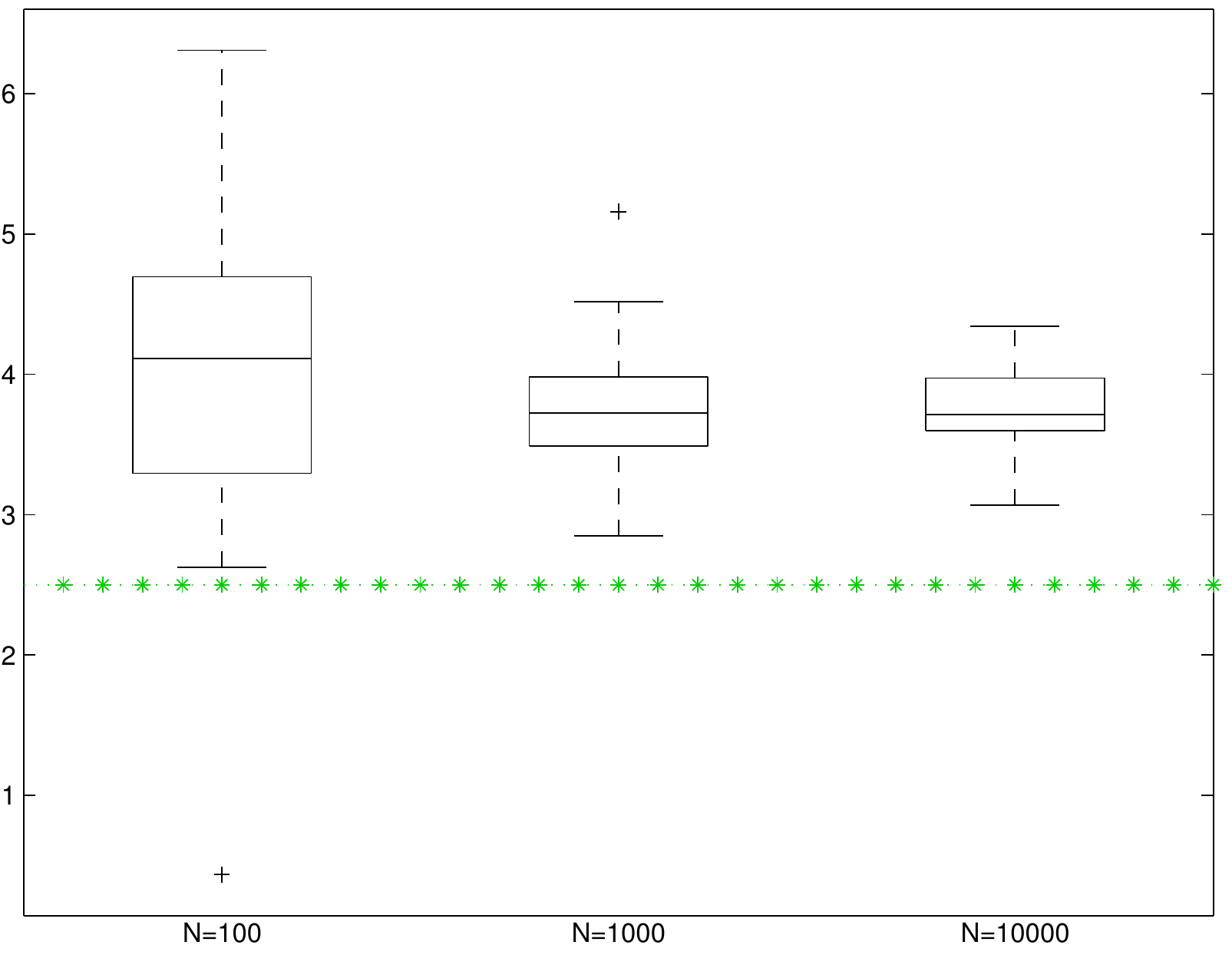}\label{fig:l63_boxplot_kappa_N}}
  \subfigure[$\hat{\sigma}^{N,10}_{1,5000}$]{\includegraphics[width=0.24\textwidth,height=5.5cm]{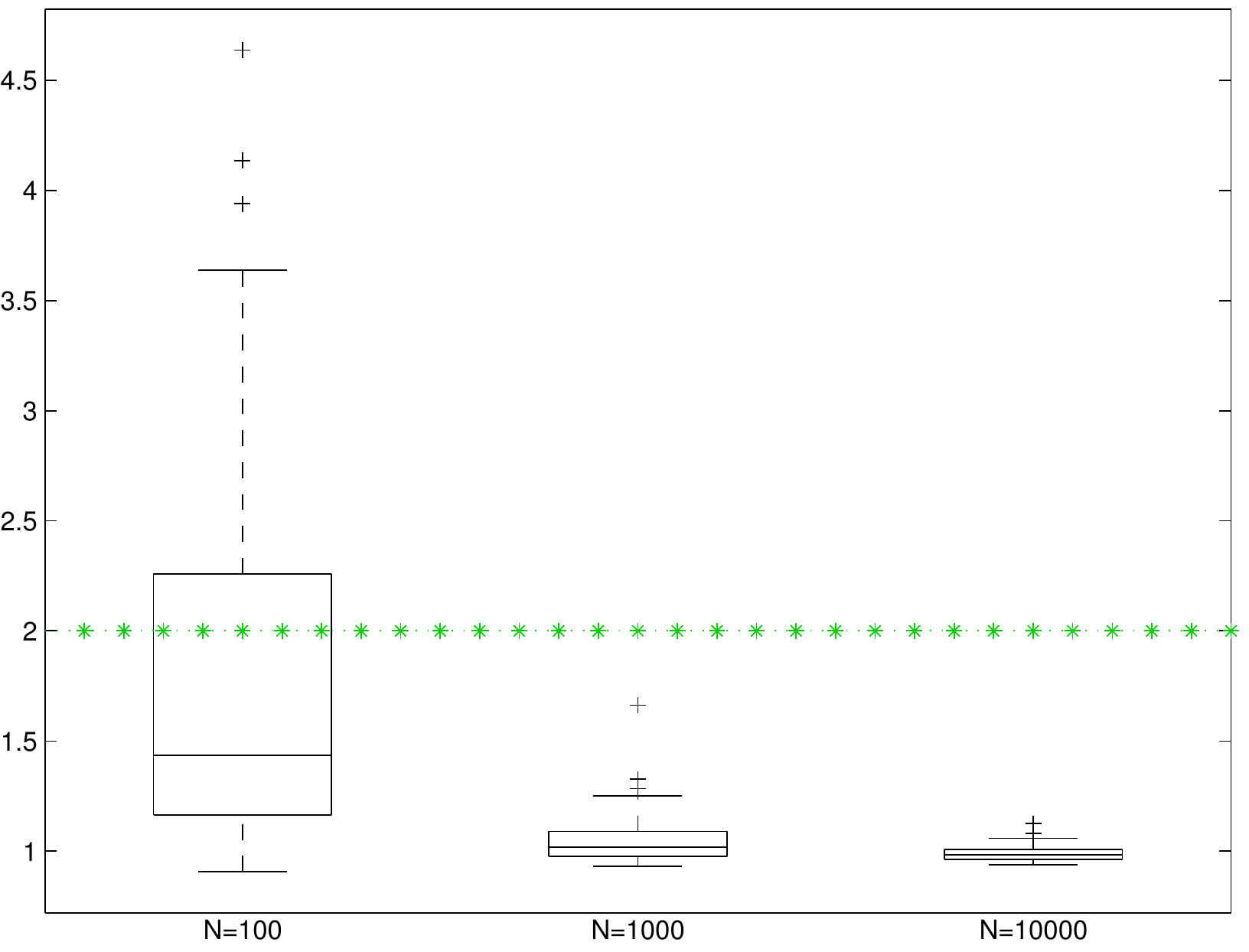}\label{fig:l63_boxplot_sigma_N}}
  \subfigure[$\hat{\sigma}^{N,10}_{63_{1,5000}}$]{\includegraphics[width=0.24\textwidth,height=5.5cm]{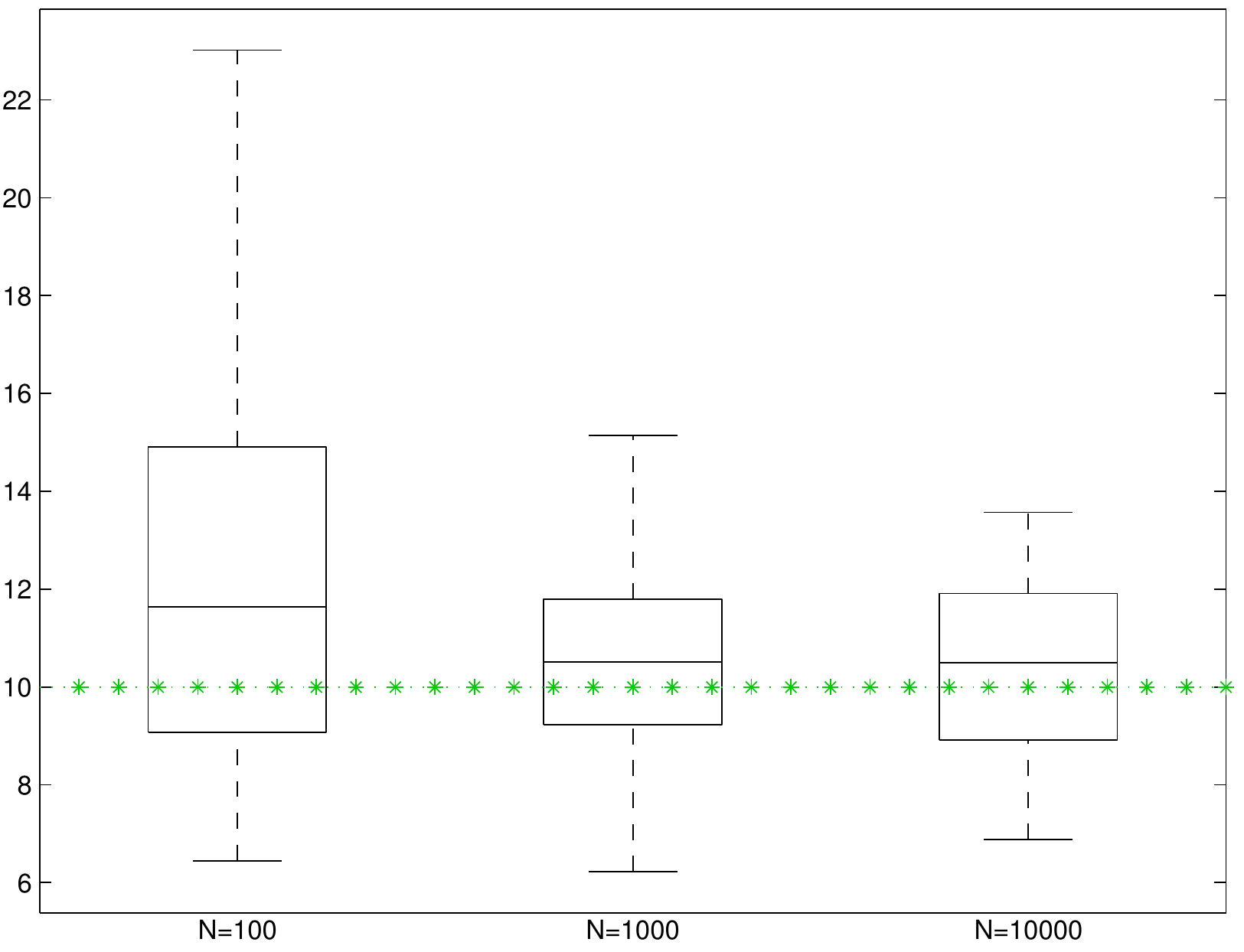}\label{fig:l63_boxplot_sigma_63_N}}
  \subfigure[$\hat{\rho}^{N,10}_{1,5000}$]{\includegraphics[width=0.24\textwidth,height=5.5cm]{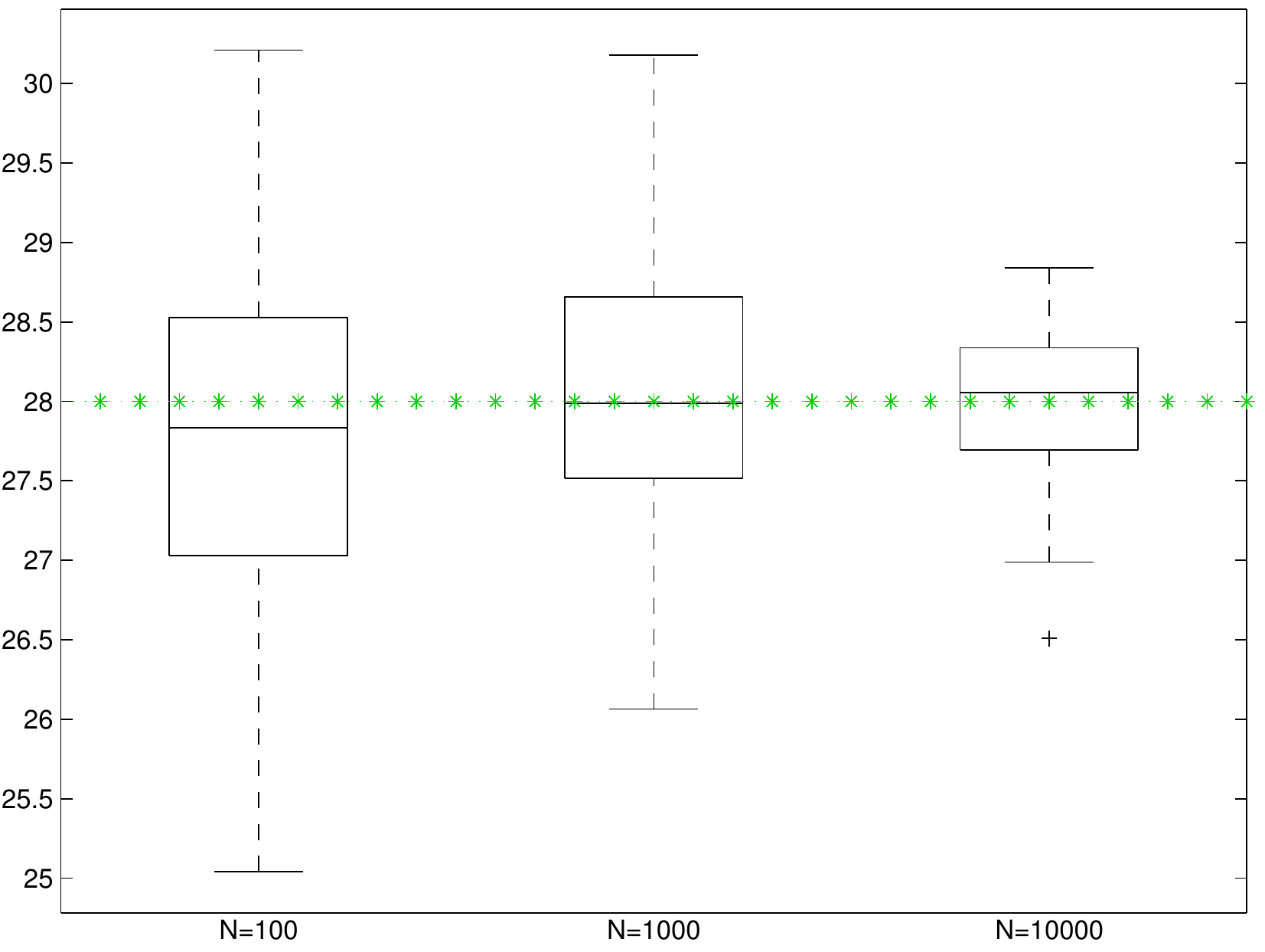}\label{fig:l63_boxplot_rho_N}}\\
  \subfigure[$\hat{\kappa}^{N,10}_{1,5000}$]{\includegraphics[width=0.24\textwidth,height=5.5cm]{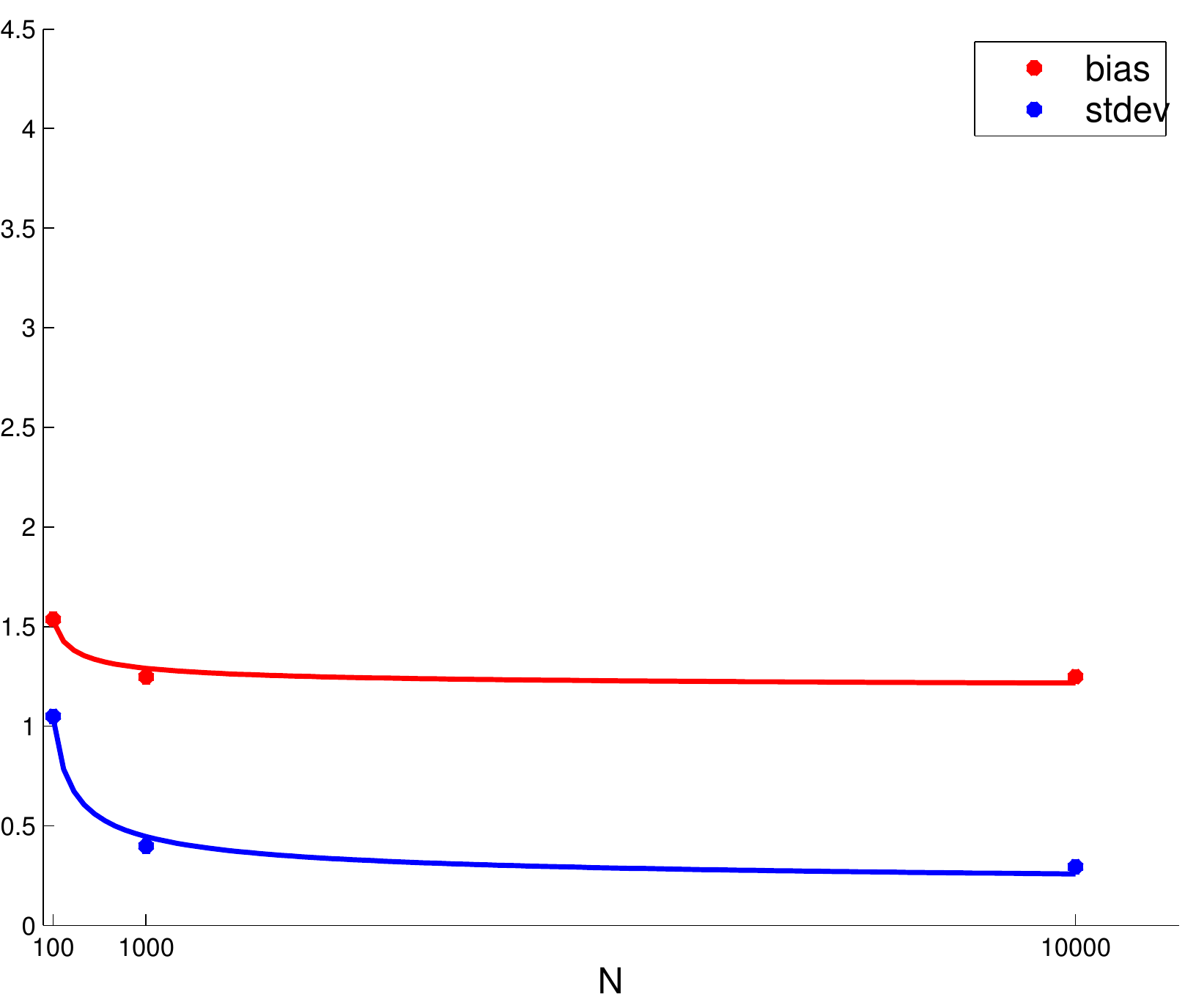}\label{fig:l63_biasvar_kappa_N}}
  \subfigure[$\hat{\sigma}^{N,10}_{1,5000}$]{\includegraphics[width=0.24\textwidth,height=5.5cm]{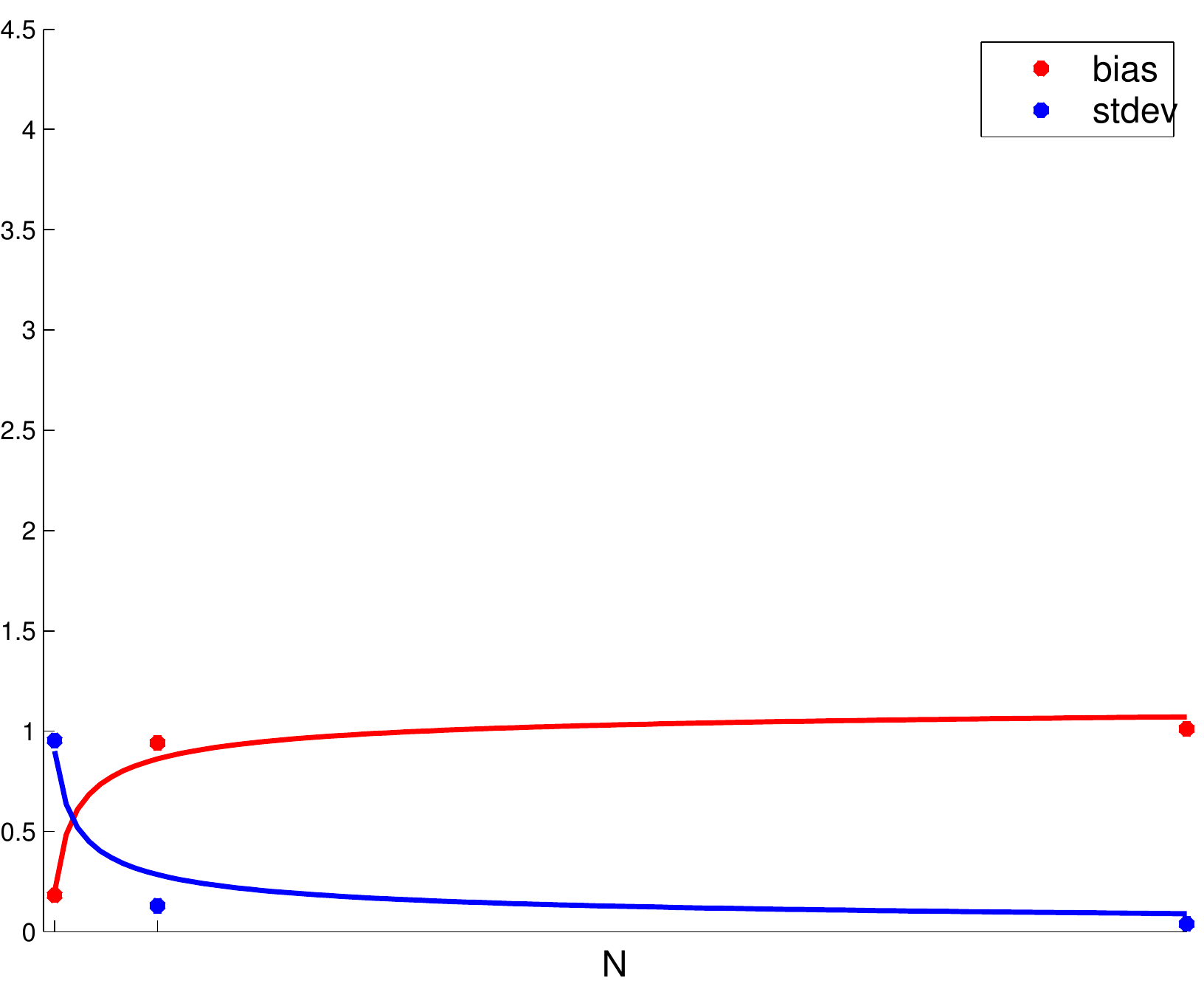}\label{fig:l63_biasvar_sigma_N}}
  \subfigure[$\hat{\sigma}^{N,10}_{63_{1,5000}}$]{\includegraphics[width=0.24\textwidth,height=5.5cm]{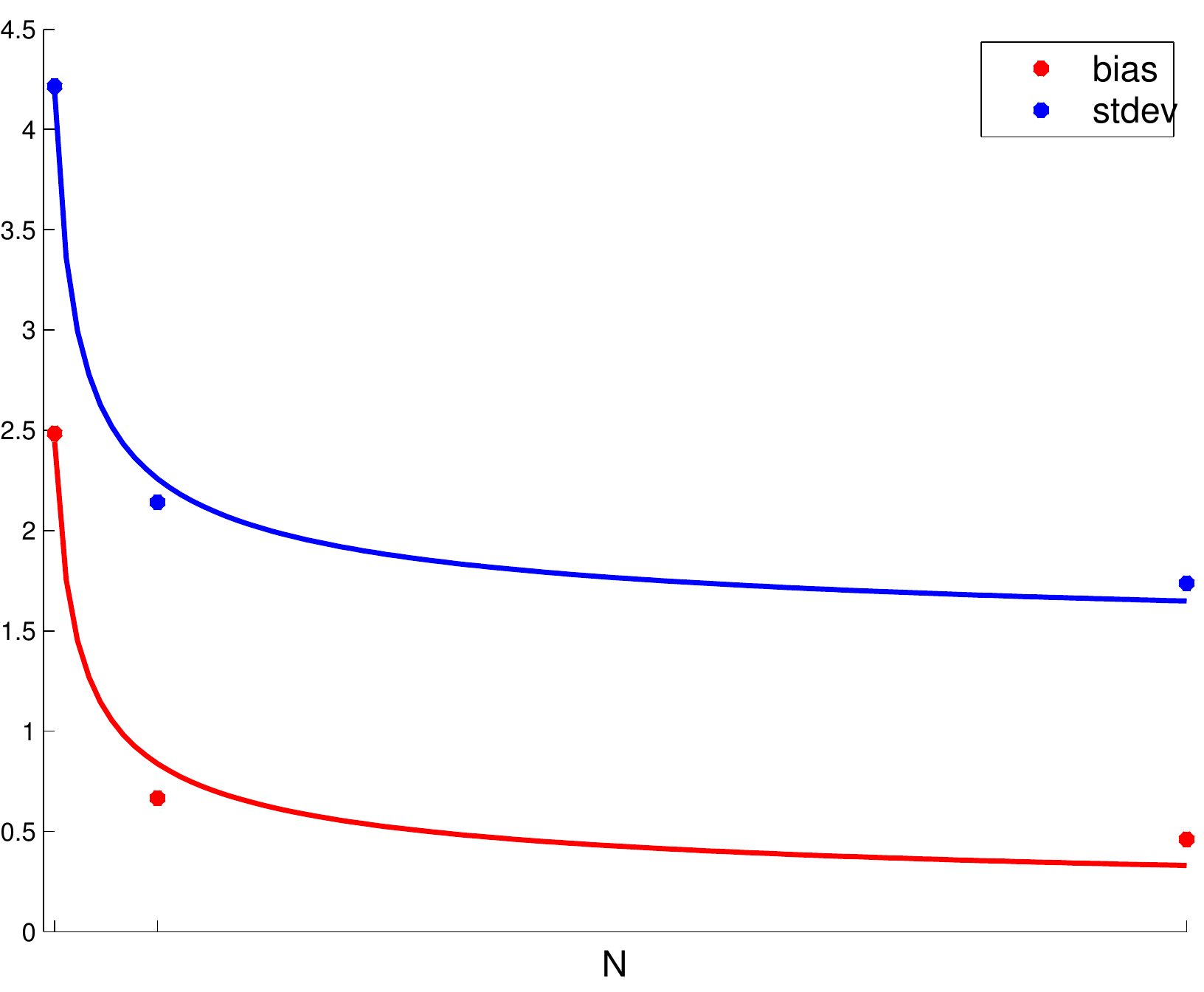}\label{fig:l63_biasvar_sigma_63_N}}
  \subfigure[$\hat{\rho}^{N,10}_{1,5000}$]{\includegraphics[width=0.24\textwidth,height=5.5cm]{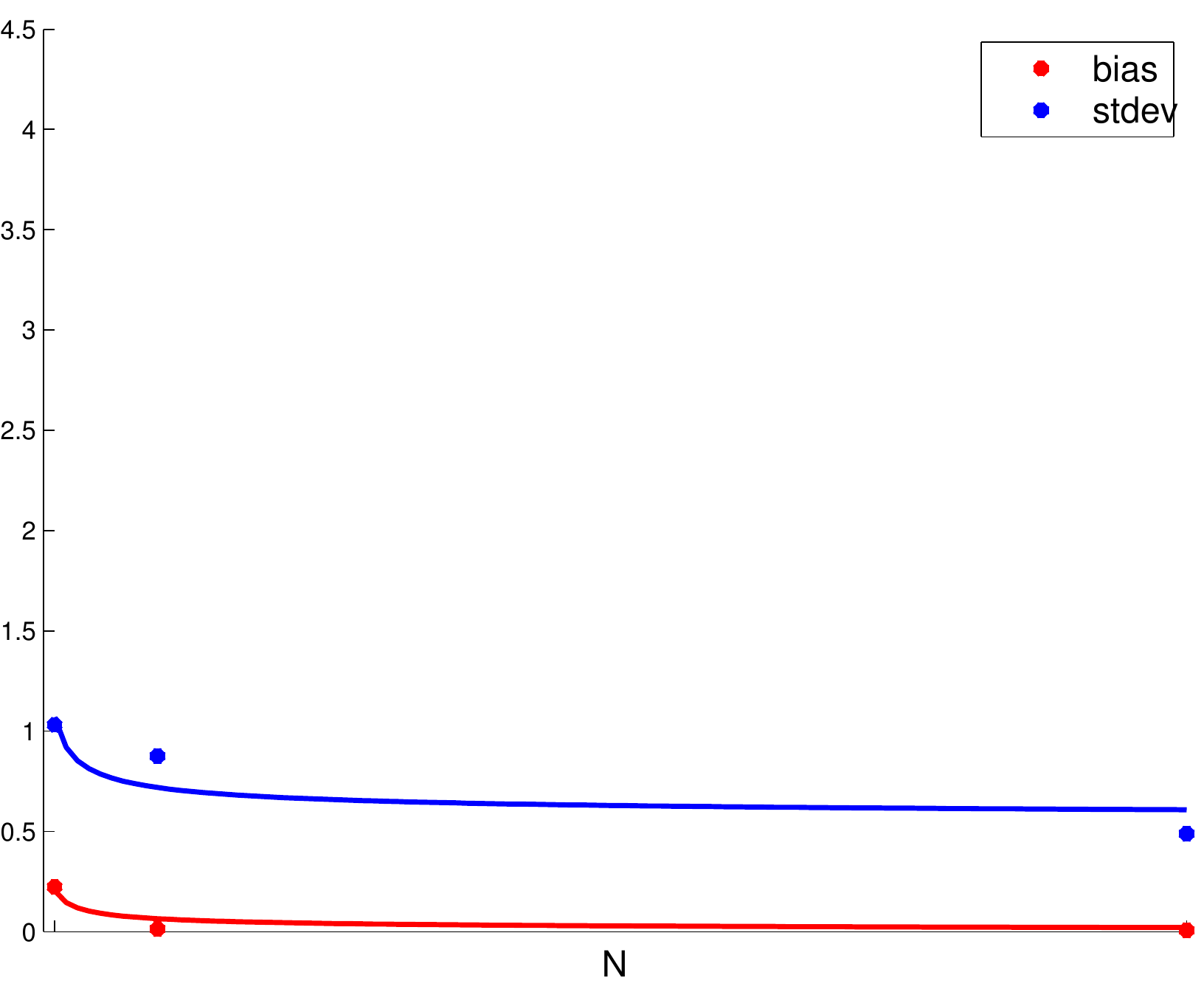}\label{fig:l63_biasvar_rho_N}}
  \caption{$\widehat{\theta}^{N,10}_{1,5000}$ when estimating $\theta=(\kappa,\sigma,\sigma_{63},\rho)$ of the Lorenz '63 HMM
, using ABC-SMC with values of $N\in\{100, 1000, 10000\}$. Figures \ref{fig:l63_boxplot_kappa_N}-\ref{fig:l63_boxplot_rho_N} show the $\widehat{\theta}^{N,10}_{1,5000}$ in boxplots and their true values in dotted green lines. Figures \ref{fig:l63_biasvar_kappa_N}-\ref{fig:l63_biasvar_rho_N} show the MC bias and MC standard deviation of the $\widehat{\theta}^{N,10}_{1,5000}$, in red and blue, with curves of least squared-error $\propto\frac{1}{\sqrt{N}}$. }
  \label{fig:l63_N}
\end{figure}

Next we look at the influence of the pseudo-observations. For $M\in\{1, 3, 5, 10, 25, 50\}$, we show in Figures \ref{fig:l63_boxplot_kappa_M}-\ref{fig:l63_boxplot_rho_M} the boxplots of the terminal estimates $\widehat{\theta}^{5000,M}_{1,5000}$ from fifty independent runs of ABC-SMC, using $N=5000$ and $\epsilon=1$. The dotted green lines marks the true $\theta$ values which generate the data. In Figures \ref{fig:l63_biasvar_kappa_M}-\ref{fig:l63_biasvar_rho_M}, the MC biases and the MC standard deviations of the $\widehat{\theta}^{5000,M}_{1,5000}$ are plotted point-wise, in red and blue, with lines of least squared-error fit to them. 
As $M$ increases, we see reductions in the MC variance. This reduction in variance can be attributed to the fact that the ABC-SMC algorithm approximates an algorithm that does not simulate the pseudo data; hence by a Rao-Blackwellization argument, one expects a reduction in variance. 
These results are consistent with~\cite{del2011adaptive}.
For this example, after $M\geq 5$, there seems to be little impact on the accuracy of the estimates; it is not clear whether such performance occurs for other examples.

\begin{figure}[hb]
  \centering
  \subfigure[$\hat{\kappa}^{5000,M}_{1,5000}$]{\includegraphics[width=0.24\textwidth,height=5.5cm]{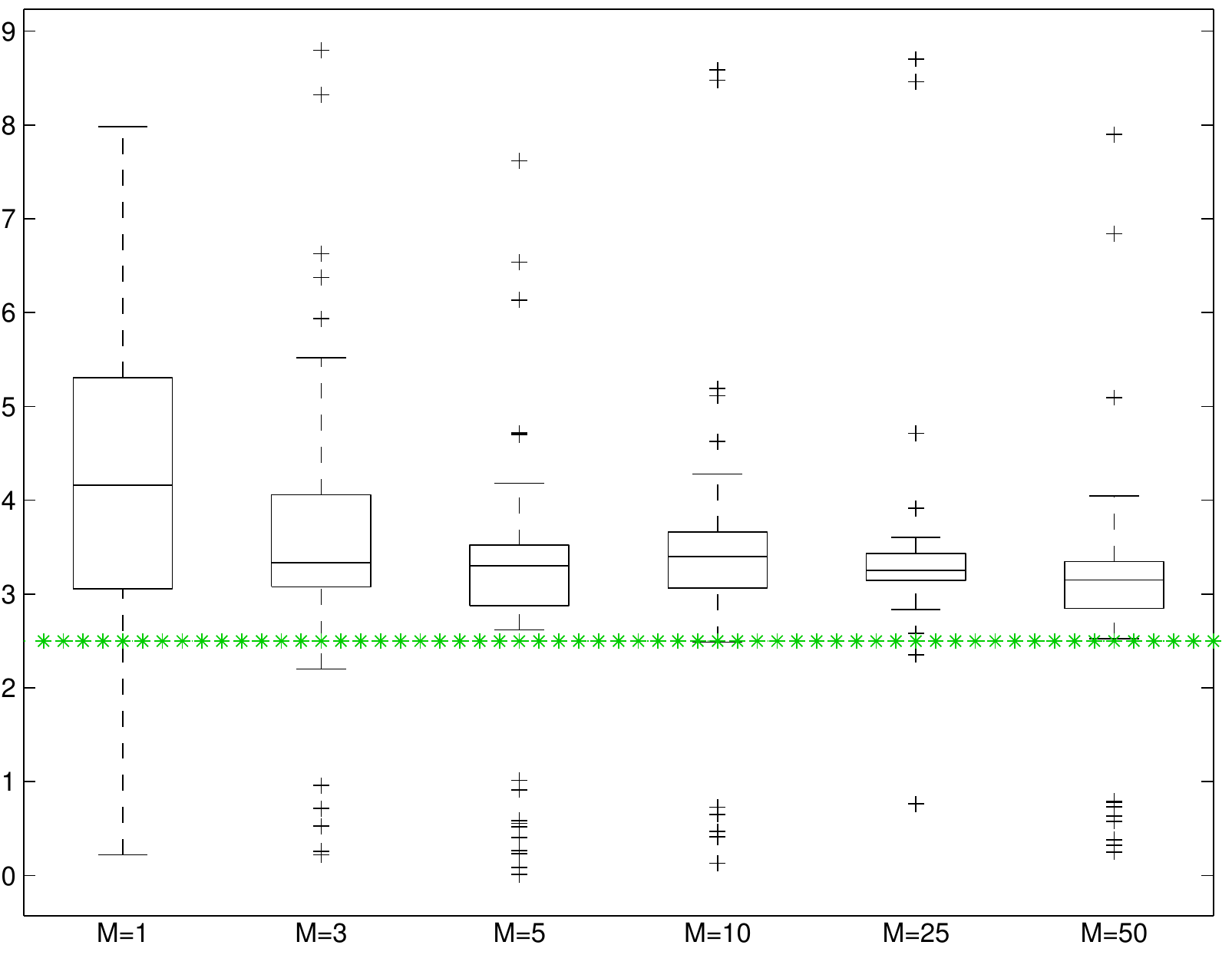}\label{fig:l63_boxplot_kappa_M}}
  \subfigure[$\hat{\sigma}^{5000,M}_{1,5000}$]{\includegraphics[width=0.24\textwidth,height=5.5cm]{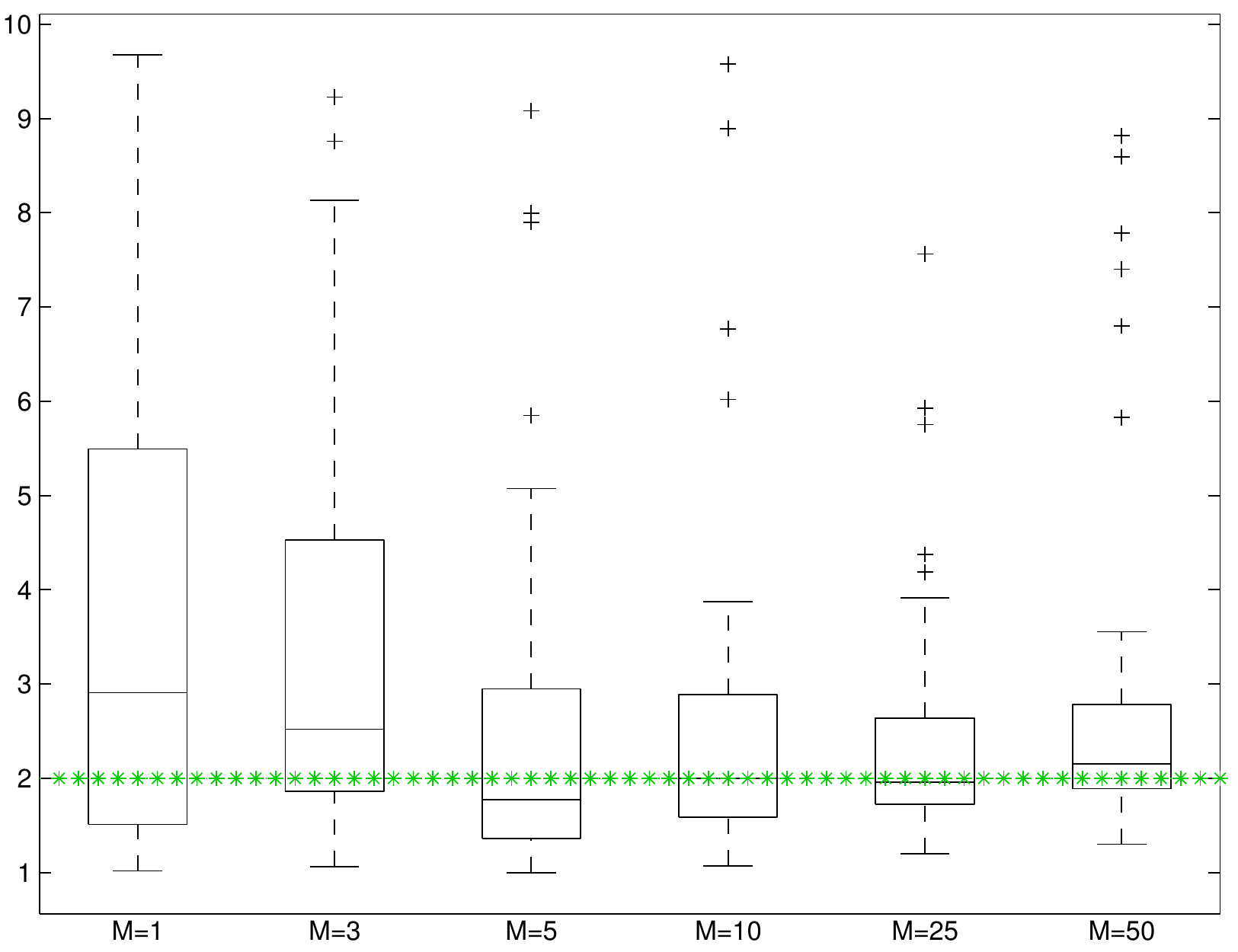}\label{fig:l63_boxplot_sigma_M}}
  \subfigure[$\hat{\sigma}_{63_{1,5000}}^{5000,M}$]{\includegraphics[width=0.24\textwidth,height=5.5cm]{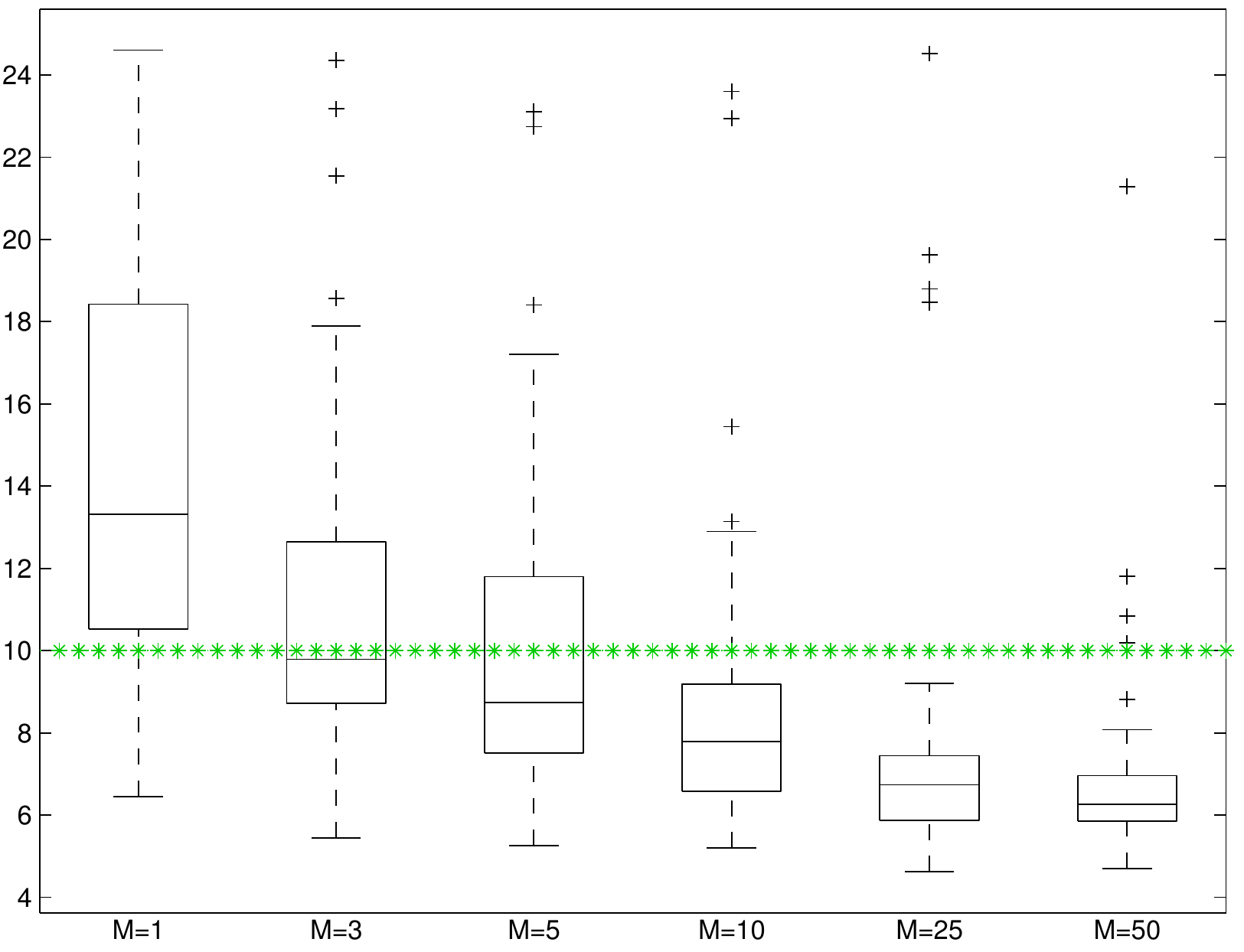}\label{fig:l63_boxplot_sigma_63_M}}
  \subfigure[$\hat{\rho}^{5000,M}_{1,5000}$]{\includegraphics[width=0.24\textwidth,height=5.5cm]{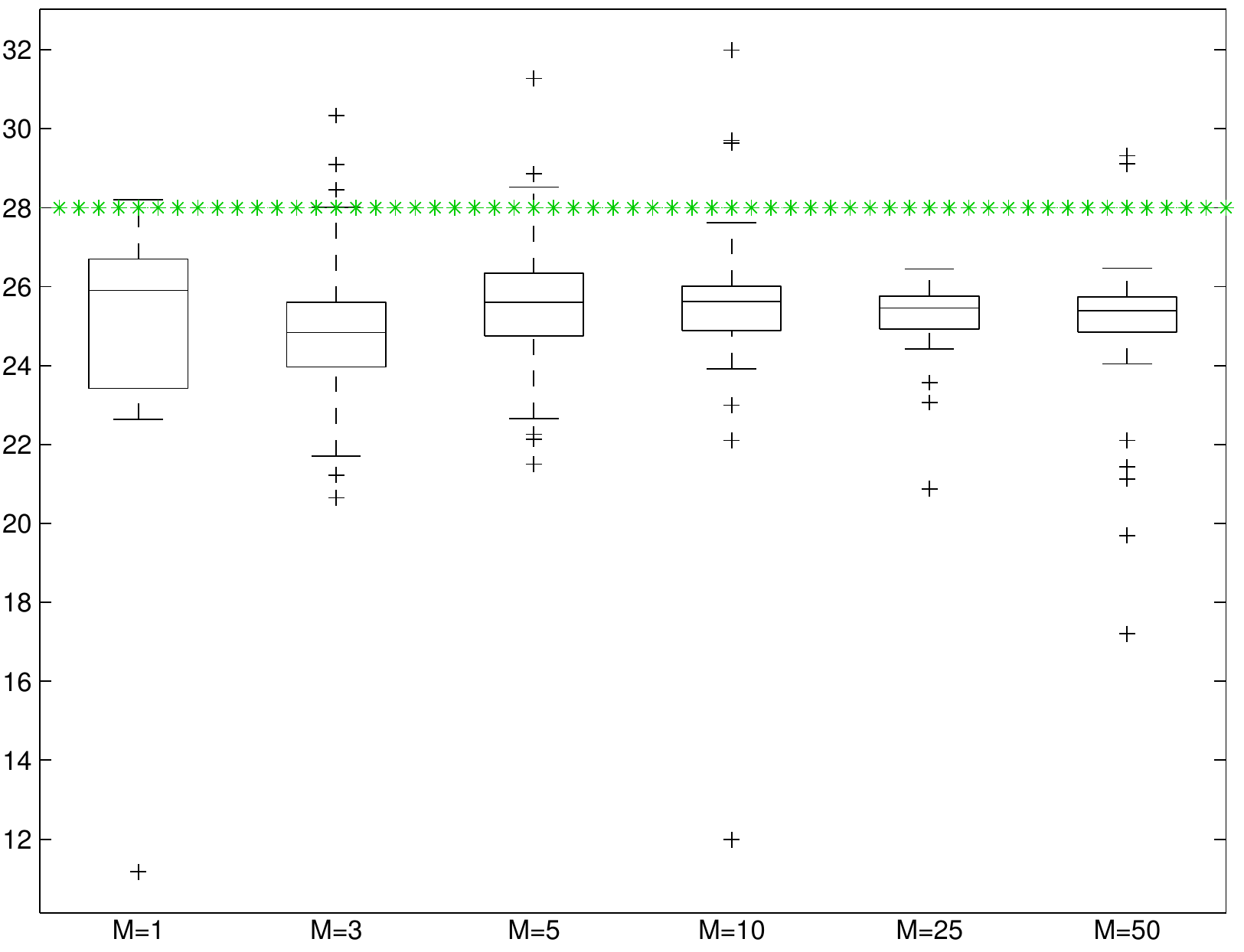}\label{fig:l63_boxplot_rho_M}}\\
  \subfigure[$\hat{\kappa}^{5000,M}_{1,5000}$]{\includegraphics[width=0.24\textwidth,height=5.5cm]{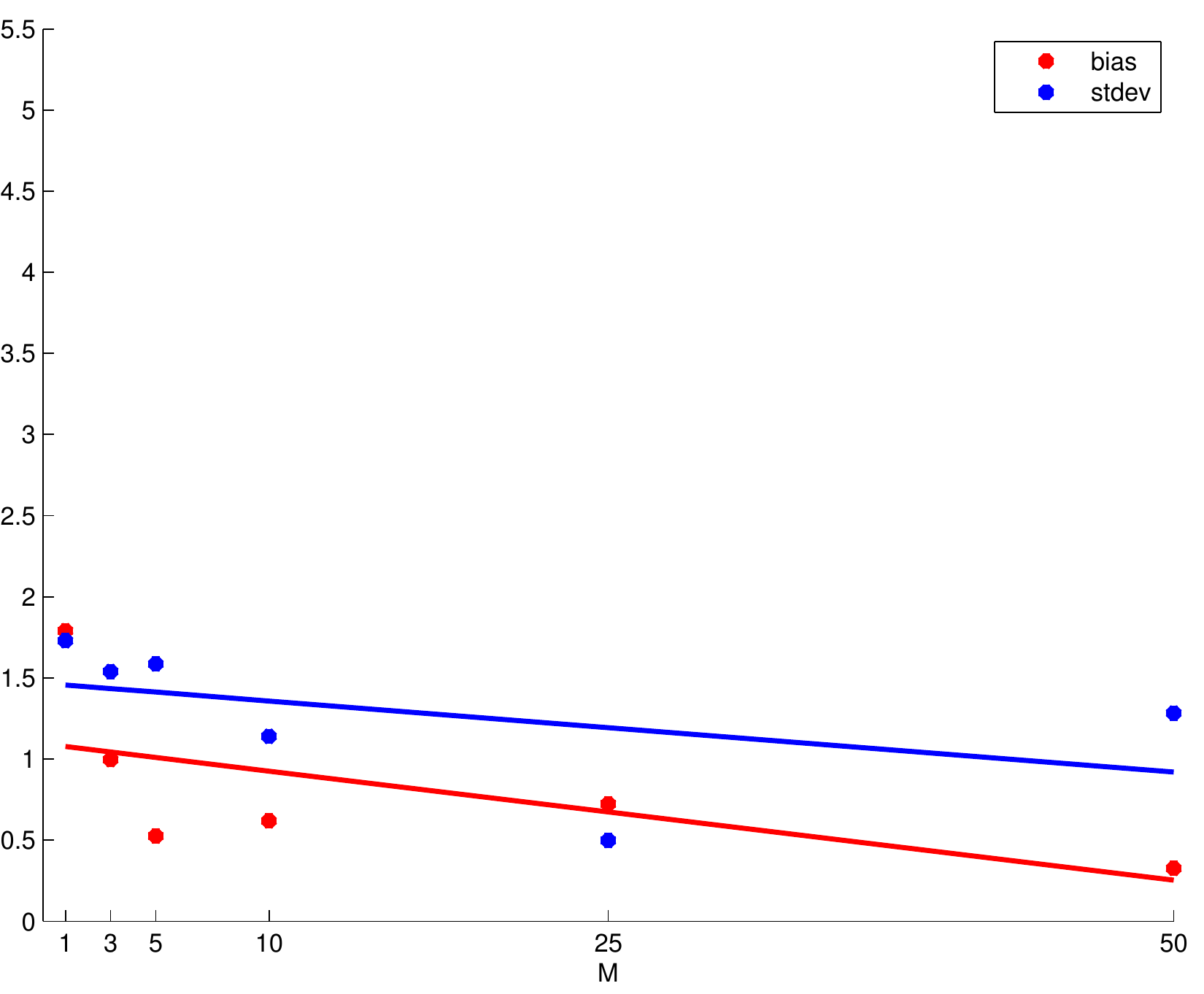}\label{fig:l63_biasvar_kappa_M}}
  \subfigure[$\hat{\sigma}^{5000,M}_{1,5000}$]{\includegraphics[width=0.24\textwidth,height=5.5cm]{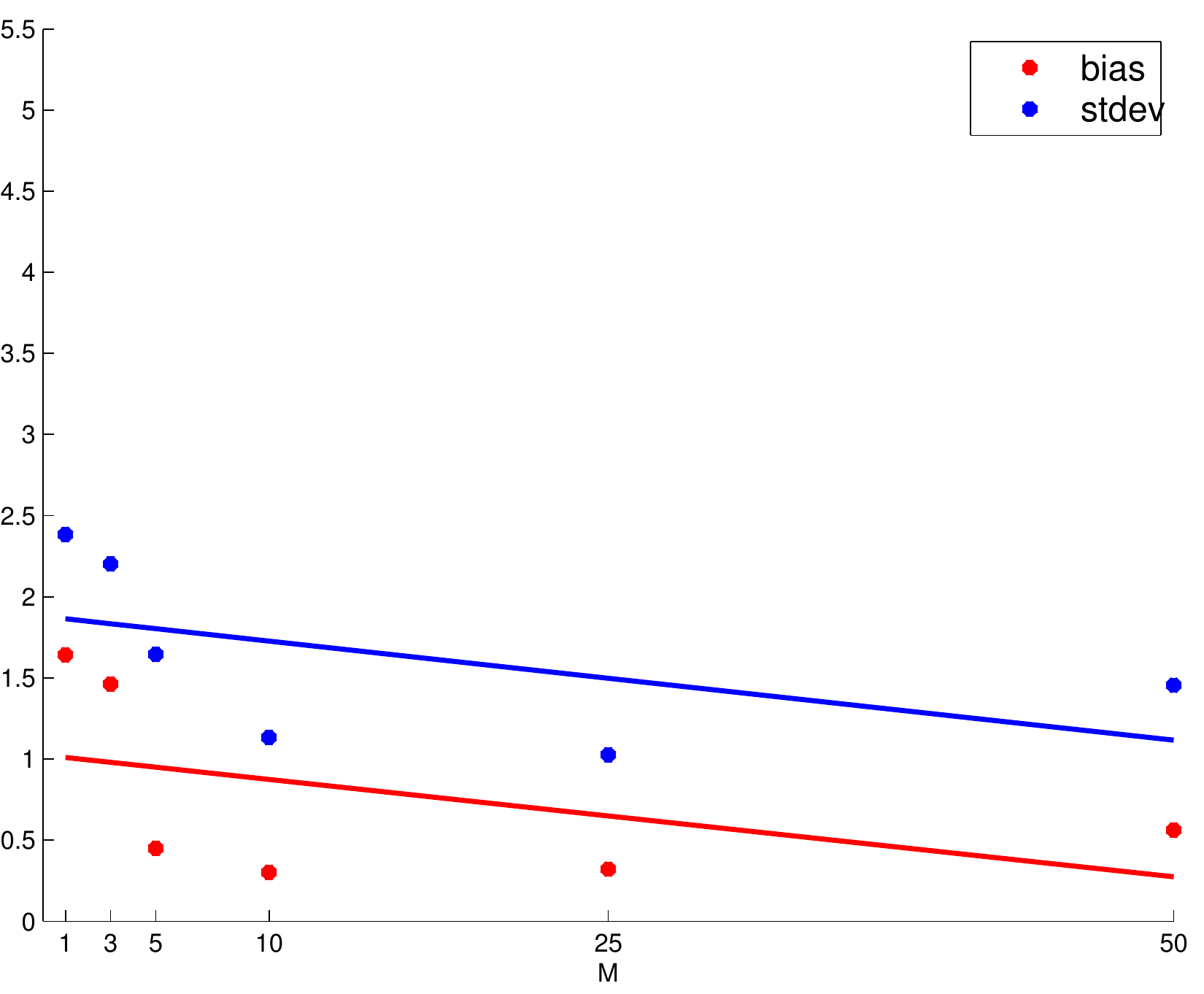}\label{fig:l63_biasvar_sigma_M}}
  \subfigure[$\hat{\sigma}_{63_{1,5000}}^{5000,M}$]{\includegraphics[width=0.24\textwidth,height=5.5cm]{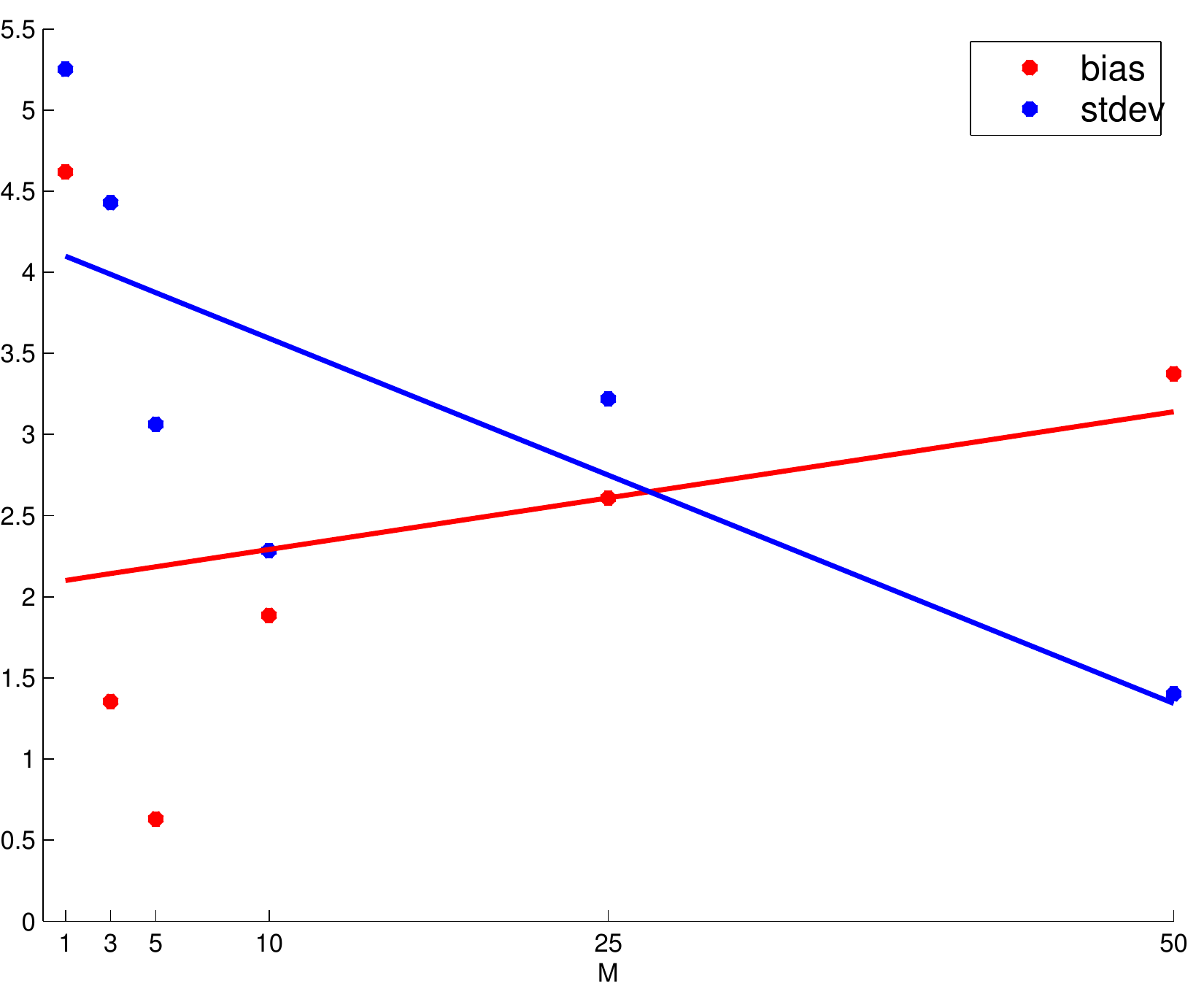}\label{fig:l63_biasvar_sigma_63_M}}
  \subfigure[$\hat{\rho}^{5000,M}_{1,5000}$]{\includegraphics[width=0.24\textwidth,height=5.5cm]{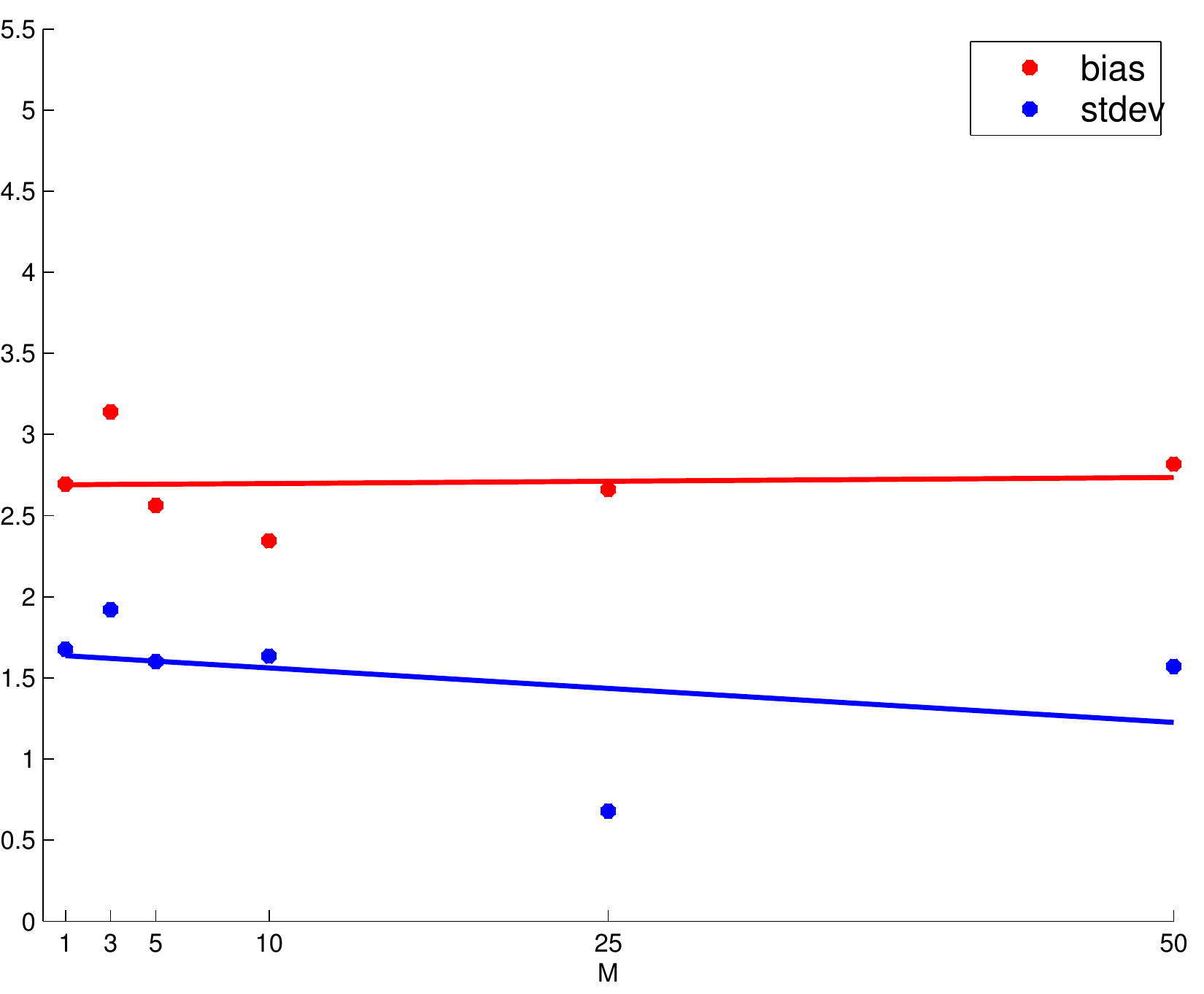}\label{fig:l63_biasvar_rho_M}}
  \caption{ $\widehat{\theta}^{5000,M}_{1,5000}$ when estimating $\theta=(\kappa,\sigma,\sigma_{63},\rho)$ of the Lorenz '63 HMM
, using ABC-SMC with values of $M\in\{1, 3, 5, 10, 25, 50\}$. Figures \ref{fig:l63_boxplot_kappa_M}-\ref{fig:l63_boxplot_rho_M} show the $\widehat{\theta}^{5000,M}_{1,5000}$ in boxplots and their true values in dotted green lines. Figures \ref{fig:l63_biasvar_kappa_M}-\ref{fig:l63_biasvar_rho_M} show the MC bias and MC standard deviation of the $\widehat{\theta}^{5000,M}_{1,5000}$, in red and blue, with lines of least squared-error. }
  \label{fig:l63_M}
\end{figure}

We now vary $n$; for $n\in\{5000, 10,000, 15,000\}$. We ran fifty independent runs of ABC-SMC using $N=200$, $M=10$, and $\epsilon=1$, and plotted boxplots of the terminal estimates $\widehat{\theta}^{200,10}_{1,n}$, in Figures \ref{fig:l63_boxplot_kappa_n}-\ref{fig:l63_boxplot_rho_n}, against the true values of $\theta$ marked in dotted green lines. Recall that recursive maximum likelihood estimation tries to maximise $\frac{1}{n}\log(p_{\theta,\epsilon}(y_{1:n}))$, so we expect $n$ not to have a great effect on the bias nor the variance (also due to the bias results in Section \ref{sec:result} and the subsequent consistency results in \cite{dean,dean1}). This is confirmed in Figures \ref{fig:l63_biasvar_kappa_n}-\ref{fig:l63_biasvar_rho_n}, where the absolute value of the MC biases and the MC standard deviations have been plotted in red and blue, and fitted with linear lines of least squared-error.

\begin{figure}[h]
  \centering
  \subfigure[$\hat{\kappa}^{200,10}_{1,n}$]{\includegraphics[width=0.24\textwidth,height=5.5cm]{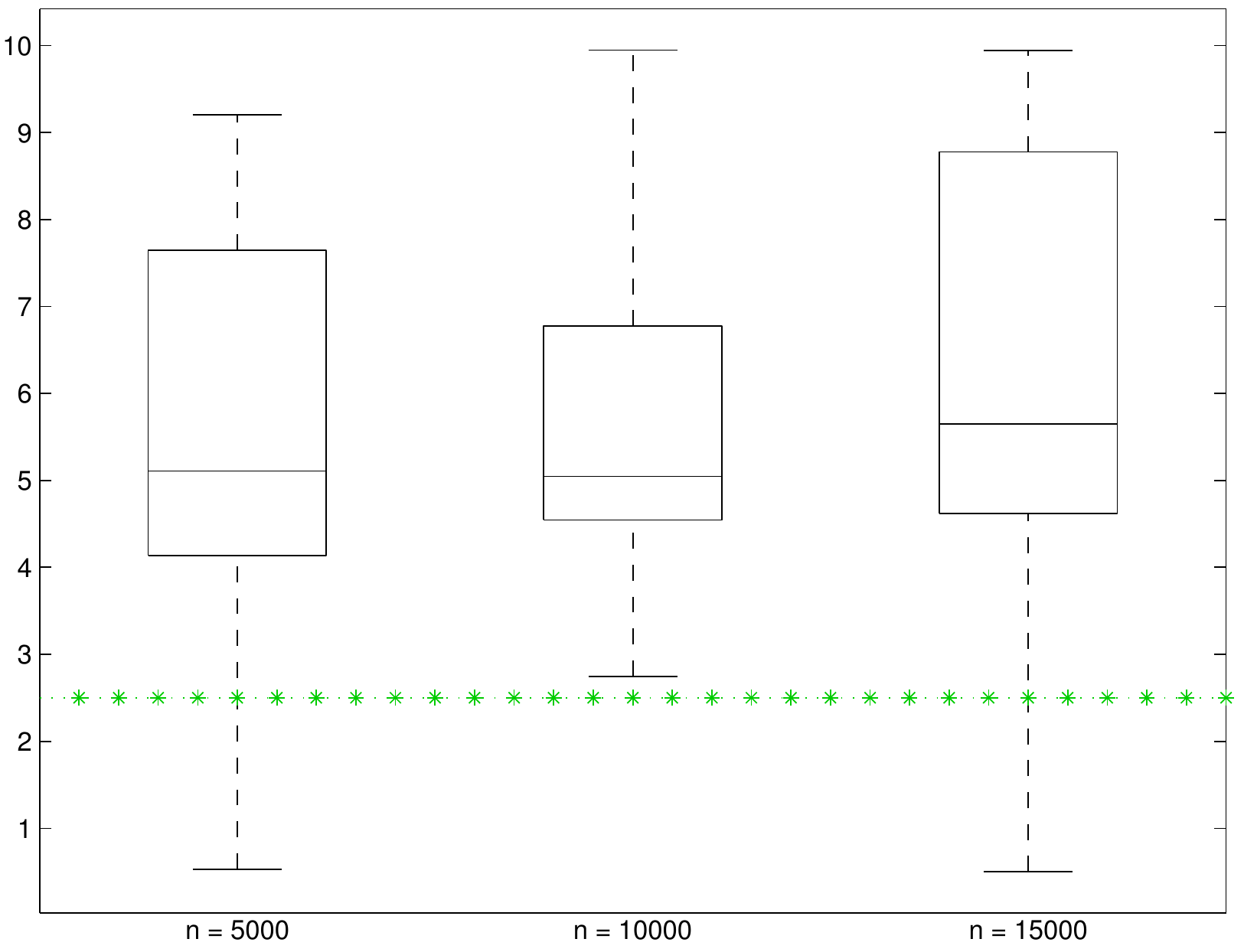}\label{fig:l63_boxplot_kappa_n}}
  \subfigure[$\hat{\sigma}^{200,10}_{1,n}$]{\includegraphics[width=0.24\textwidth,height=5.5cm]{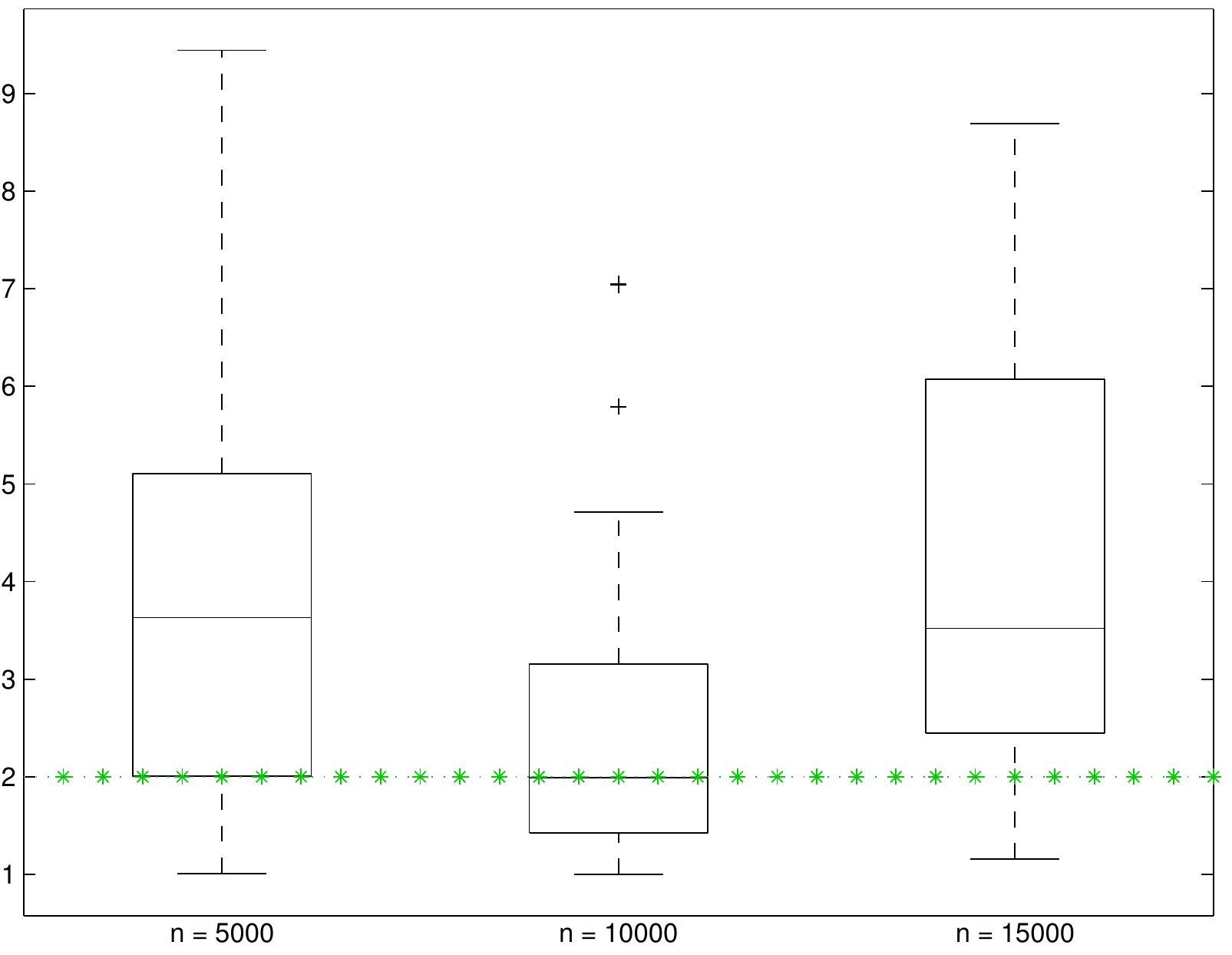}\label{fig:l63_boxplot_sigma_n}}
  \subfigure[$\hat{\sigma}^{200,10}_{63_{1,n}}$]{\includegraphics[width=0.24\textwidth,height=5.5cm]{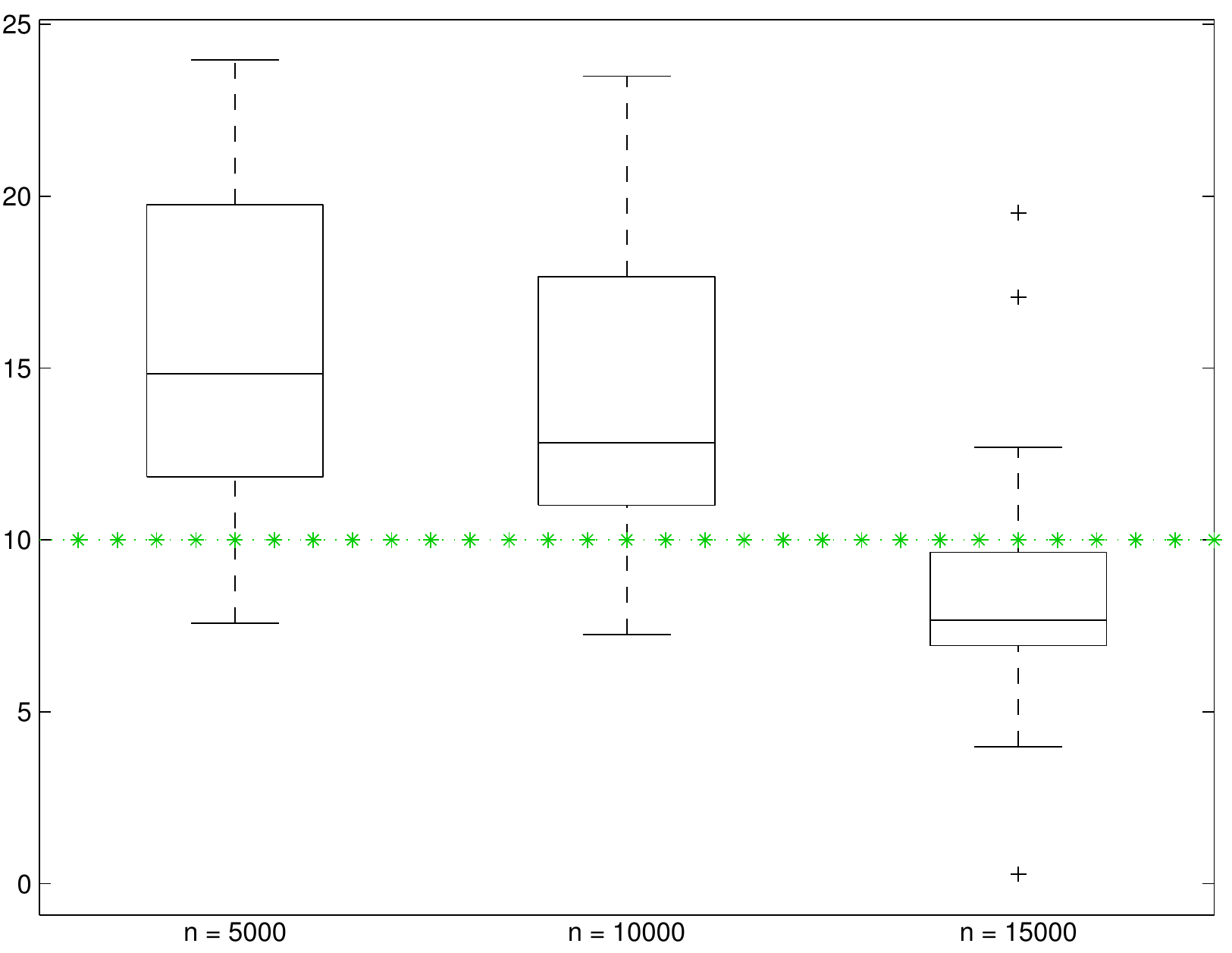}\label{fig:l63_boxplot_sigma_63_n}}
  \subfigure[$\hat{\rho}^{200,10}_{1,n}$]{\includegraphics[width=0.24\textwidth,height=5.5cm]{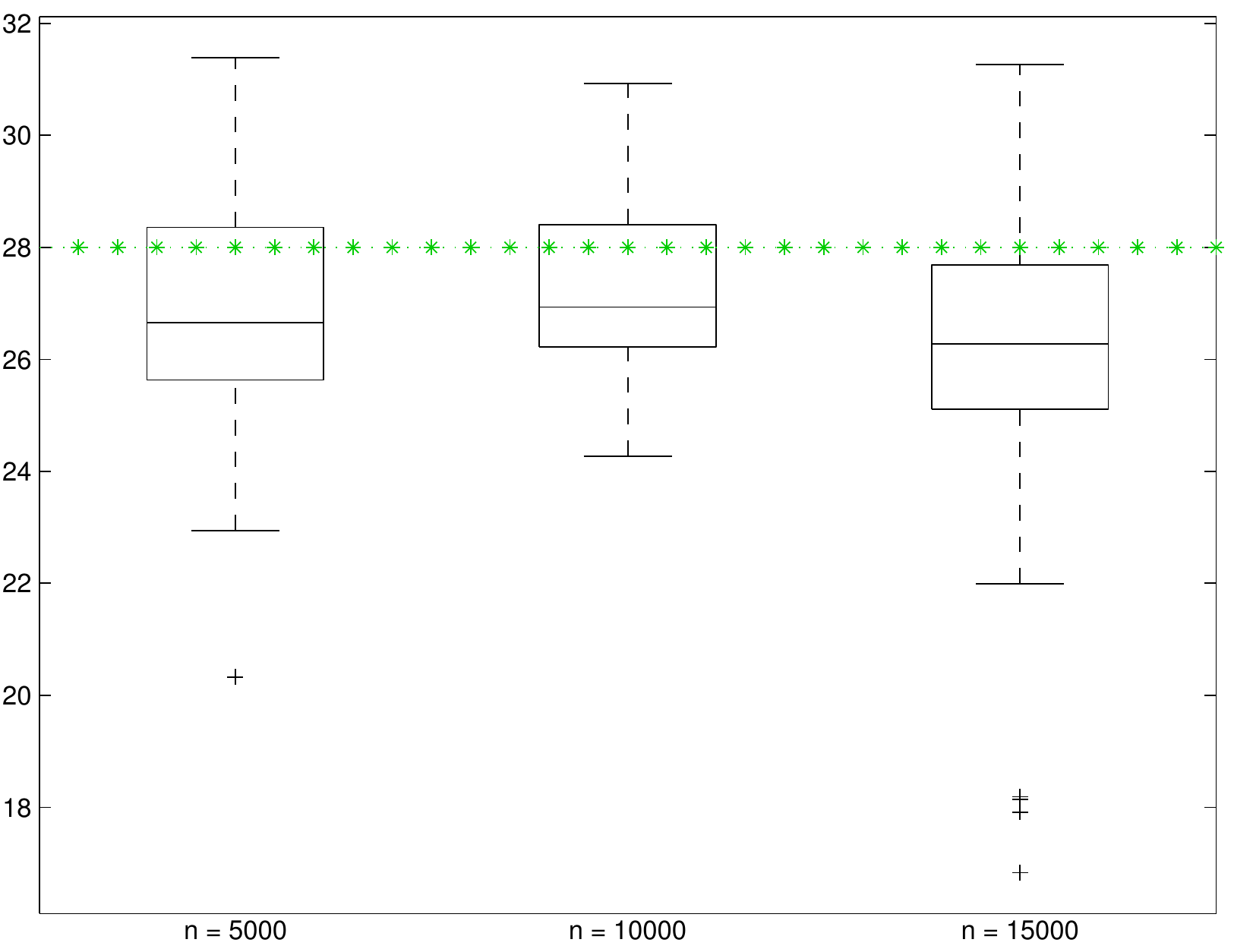}\label{fig:l63_boxplot_rho_n}}\\
  \subfigure[$\hat{\kappa}^{200,10}_{1,n}$]{\includegraphics[width=0.24\textwidth,height=5.5cm]{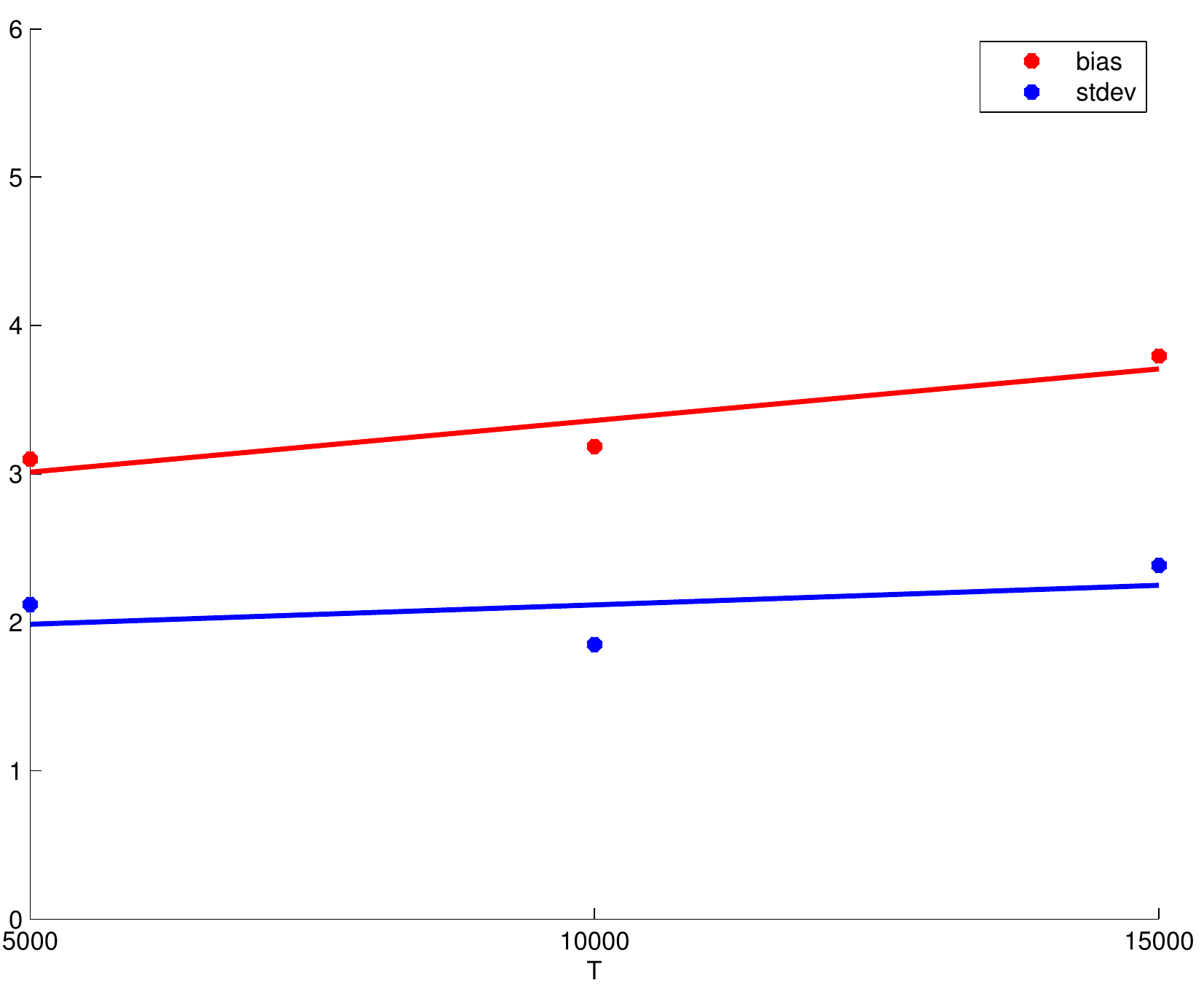}\label{fig:l63_biasvar_kappa_n}}
  \subfigure[$\hat{\sigma}^{200,10}_{1,n}$]{\includegraphics[width=0.24\textwidth,height=5.5cm]{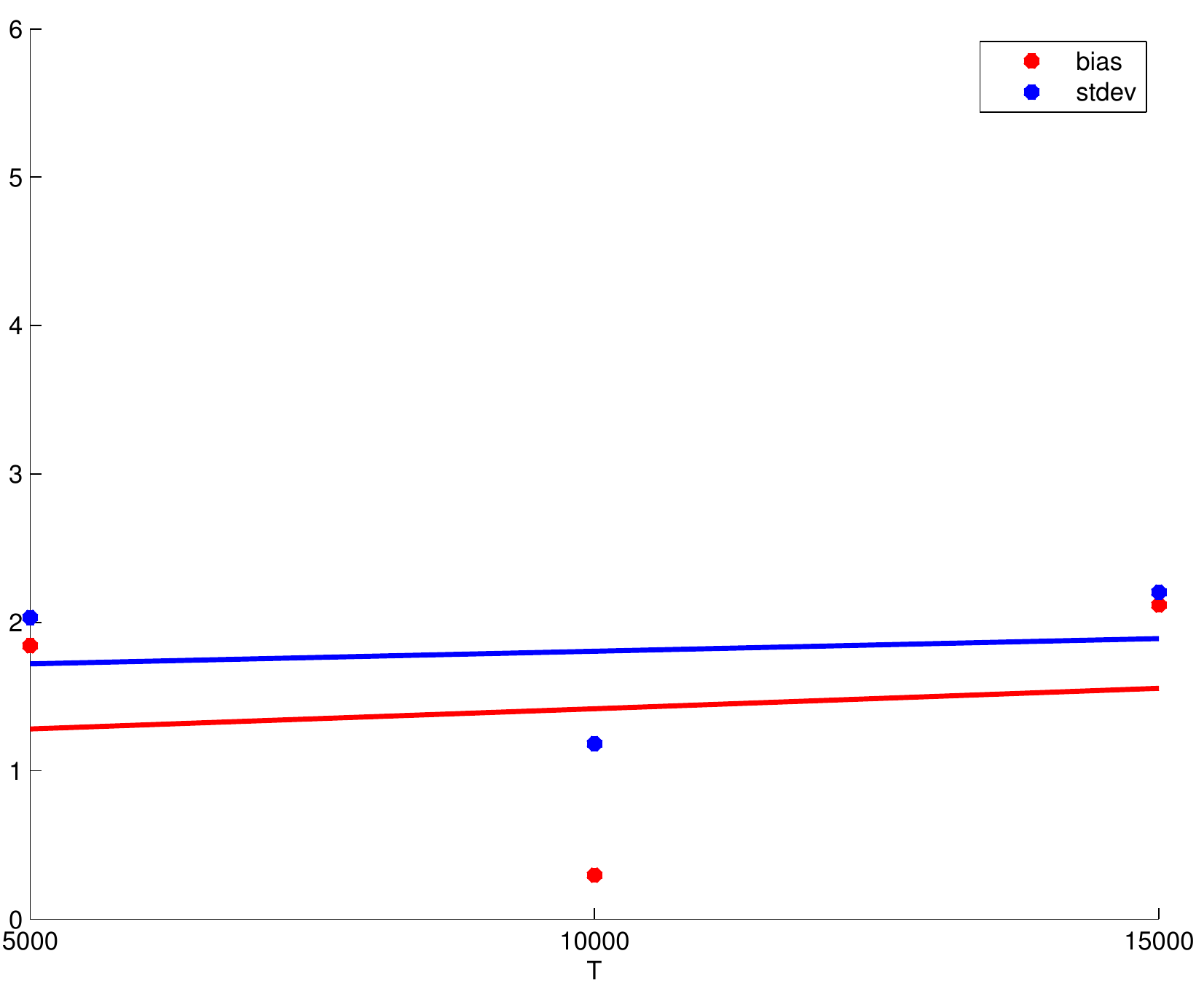}\label{fig:l63_biasvar_sigma_n}}
  \subfigure[$\hat{\sigma}^{200,10}_{63_{1,n}}$]{\includegraphics[width=0.24\textwidth,height=5.5cm]{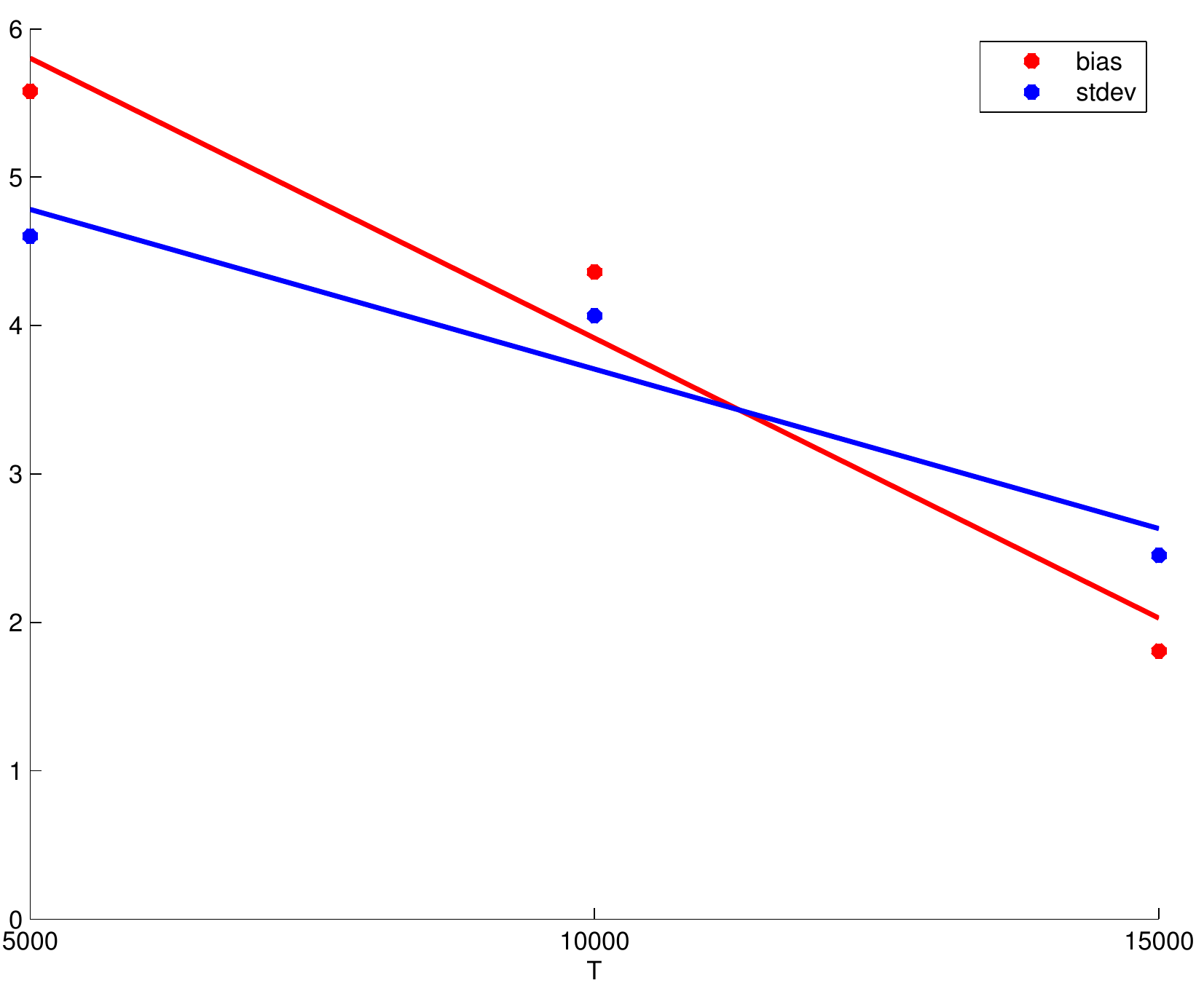}\label{fig:l63_biasvar_sigma_63_n}}
  \subfigure[$\hat{\rho}^{200,10}_{1,n}$]{\includegraphics[width=0.24\textwidth,height=5.5cm]{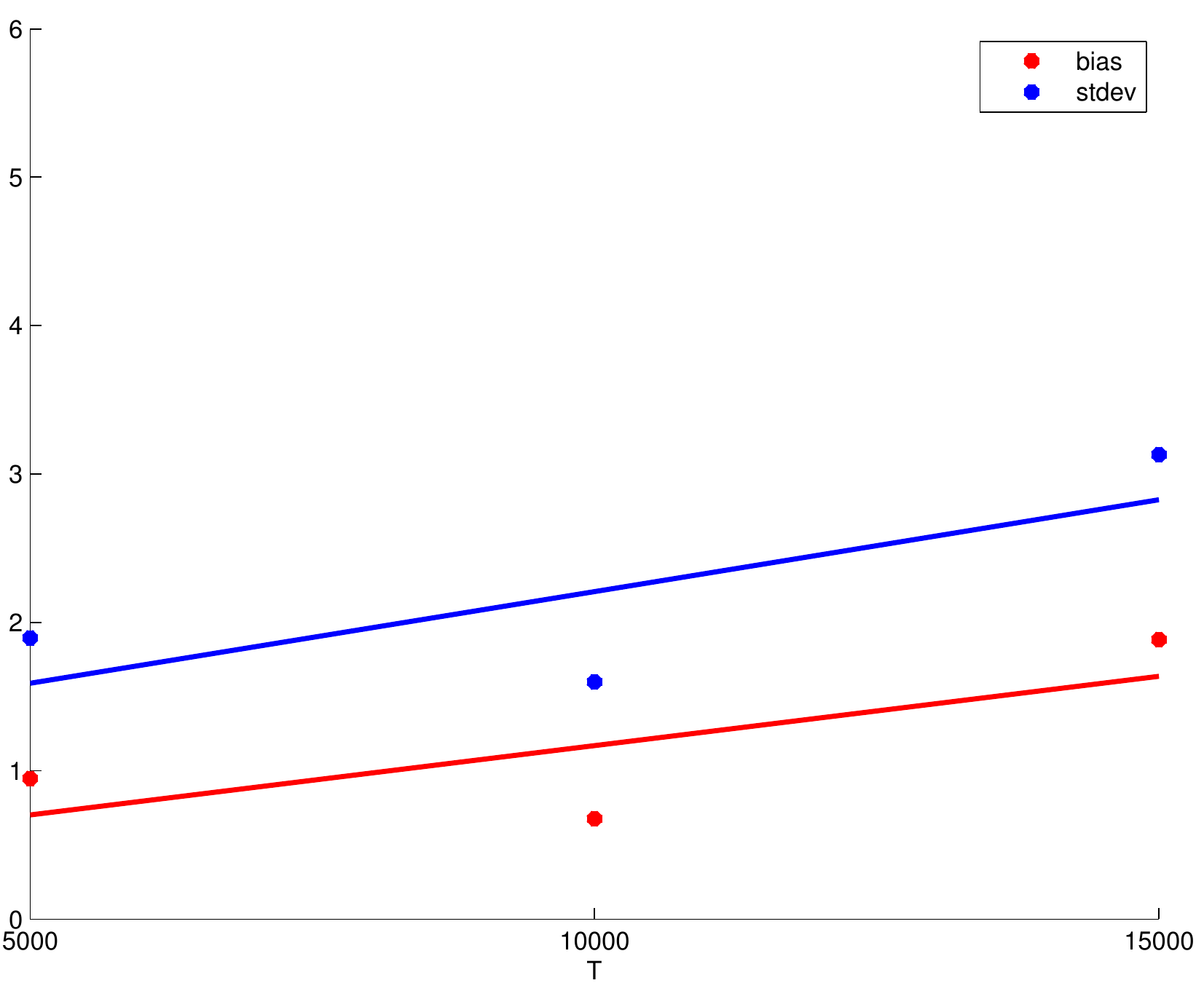}\label{fig:l63_biasvar_rho_n}}
  \caption{ $\widehat{\theta}^{200,10}_{1,n}$ when using ABC-SMC to estimate $\theta=(\kappa,\sigma,\sigma_{63},\rho)$ of the Lorenz '63 HMM
,  for datasets of length $n\in\{5000, 10000, 15000\}$. Figures \ref{fig:l63_boxplot_kappa_n}-\ref{fig:l63_boxplot_rho_n} show the $\widehat{\theta}^{200,10}_{1,n}$ in boxplots and their true values in dotted green lines. Figures \ref{fig:l63_biasvar_kappa_n}-\ref{fig:l63_biasvar_rho_n} show the MC bias and MC standard deviation of the $\widehat{\theta}^{200,10}_{1,n}$, in red and blue, with lines of least squared-error. }
  \label{fig:l63_n}
\end{figure}

Finally,  we investigate the influence of $\epsilon\in\{1, 2, 3, 4, 5, 6, 7, 8, 9, 10, 50\}$. For each $\epsilon$, we again ran fifty independent runs of ABC-SMC with $N=200$ and $M=10$, for the dataset $n=5000$. The boxplot of the parameter estimates are plotted, in Figures \ref{fig:l63_boxplot_kappa_e}-\ref{fig:l63_boxplot_rho_e}, against dotted green lines which indicate the true $\theta$. 
Figures \ref{fig:l63_biasvar_kappa_e}-\ref{fig:l63_biasvar_rho_e} show the absolute value of MC biases in red, and the MC standard deviations in blue. Fitted to the MC biases is a non-linear least squares curve proportional to $\epsilon + \frac{1}{\epsilon}$. The result we presented in Section \ref{sec:result} states that as $\epsilon$ increases, the bias will increase on $\mathcal{O}(\epsilon)$, hence the term proportional to $\epsilon$ of the fitted curve. However, the ABC-SMC algorithm becomes less stable for $\epsilon$ too small (in the sense that, for example, the variance of the weights will become larger as $\epsilon$ grows), incurring more varied estimates and affected biases; thus the term proportional to $\frac{1}{\epsilon}$. Fitted to the MC standard deviations is a non-linear least squares curves proportional to $\frac{1}{\epsilon}$. For this example, the MC standard deviation decreases at this rate as $\epsilon$ increases. 

\begin{figure}[h]
  \centering
  \subfigure[$\hat{\kappa}^{200,10}_{\epsilon,5000}$]{\includegraphics[width=0.24\textwidth,height=5.5cm]{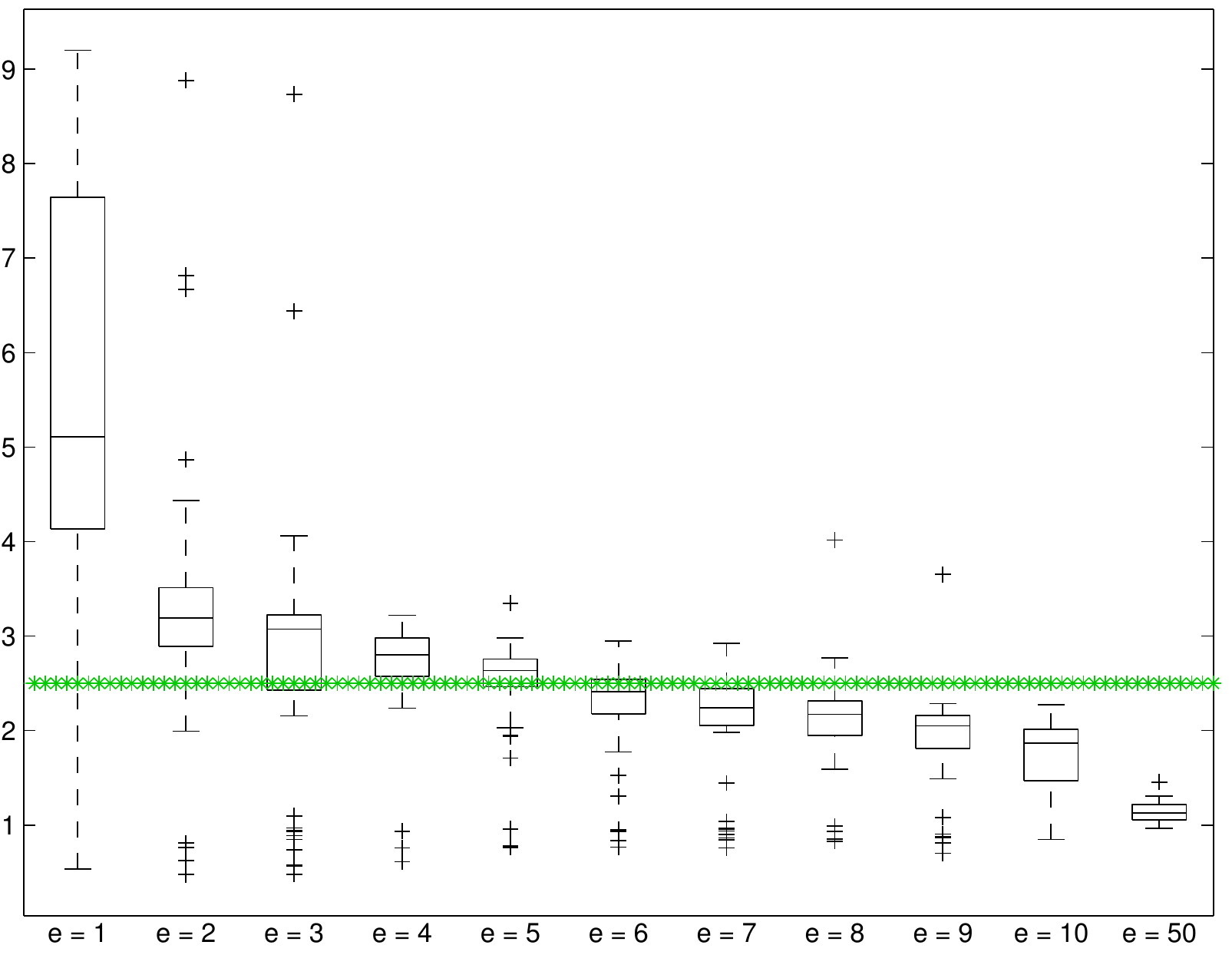}\label{fig:l63_boxplot_kappa_e}}
  \subfigure[$\hat{\sigma}^{200,10}_{\epsilon,5000}$]{\includegraphics[width=0.24\textwidth,height=5.5cm]{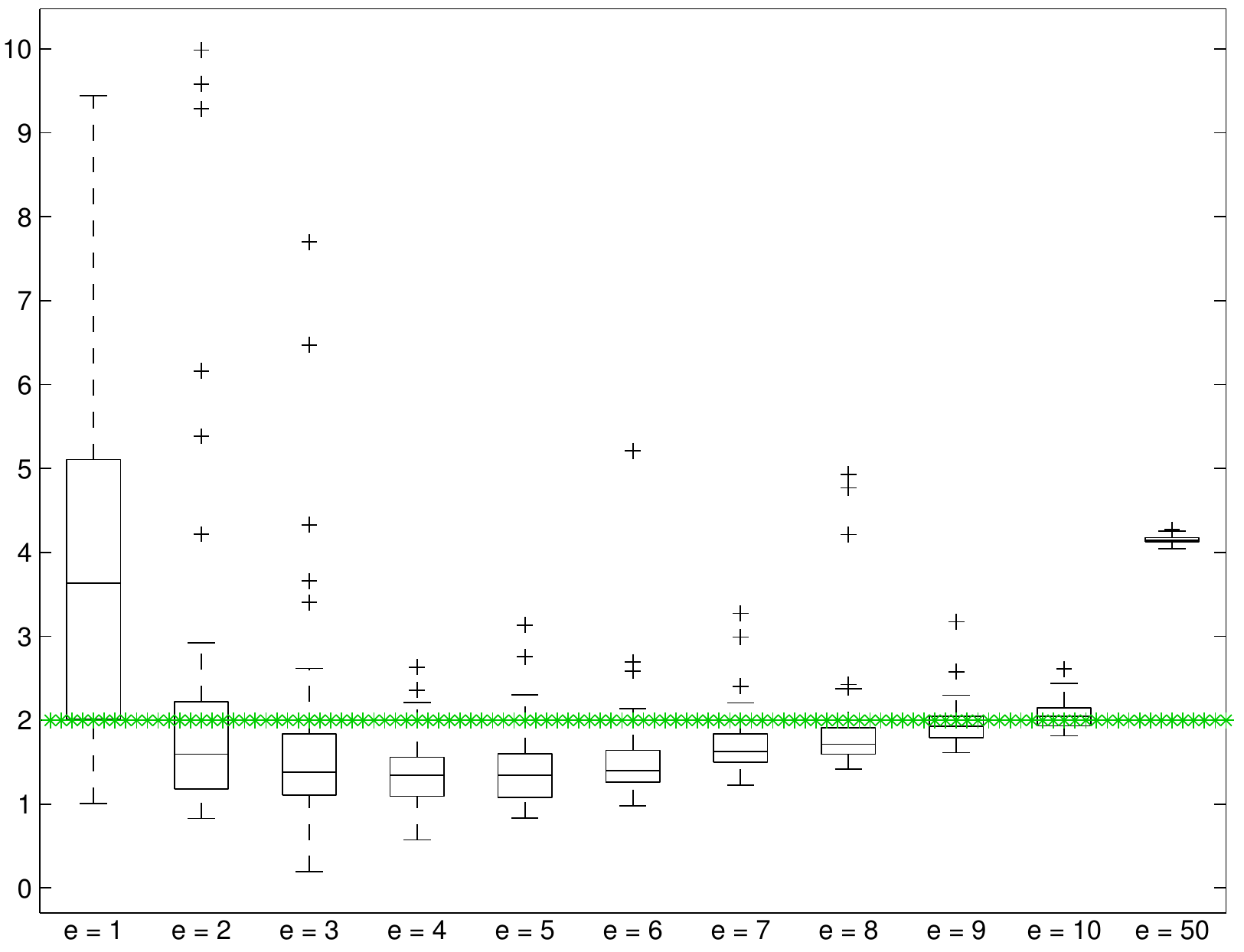}\label{fig:l63_boxplot_sigma_e}}
  \subfigure[$\hat{\sigma}^{200,10}_{63_{\epsilon,5000}}$]{\includegraphics[width=0.24\textwidth,height=5.5cm]{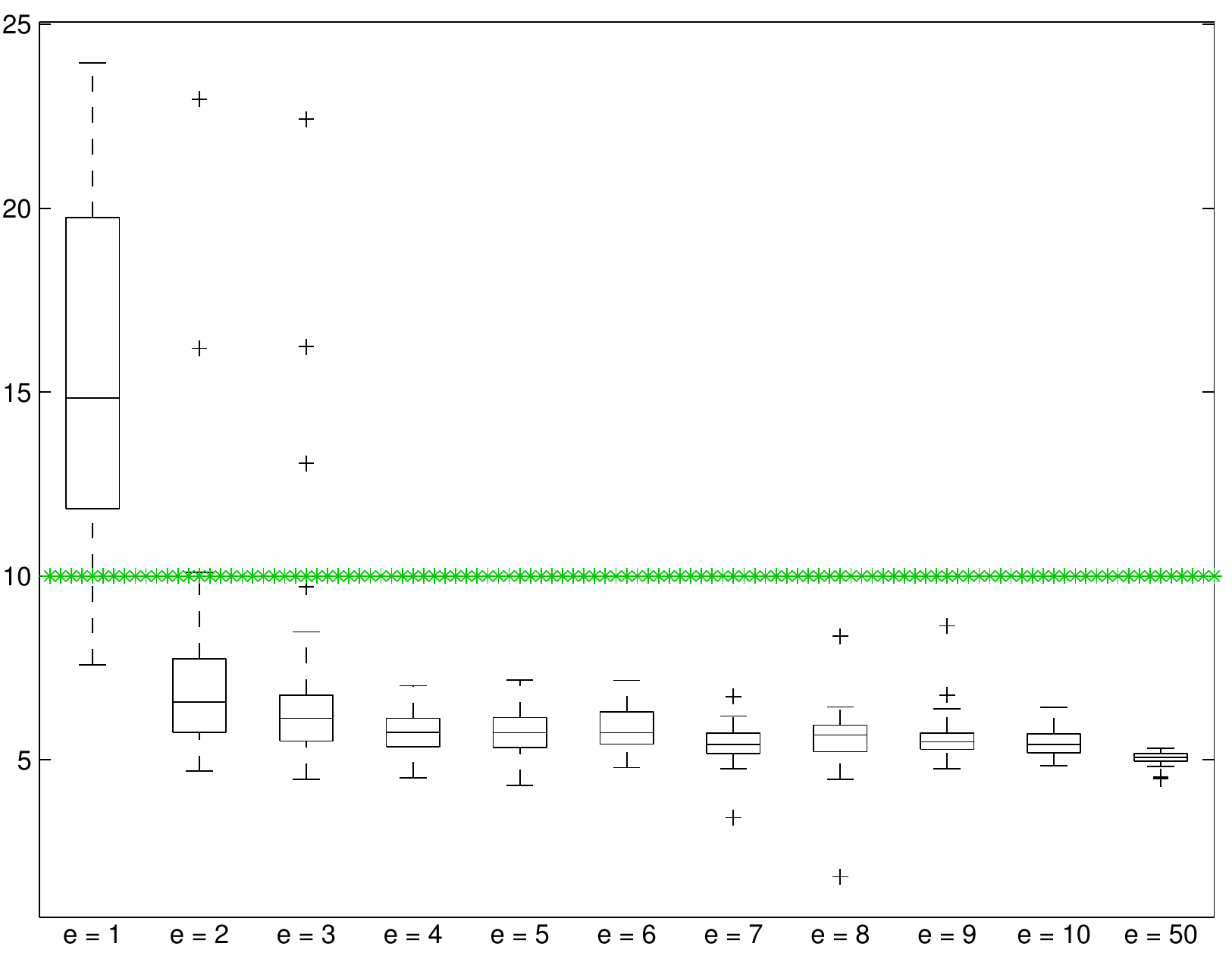}\label{fig:l63_boxplot_sigma_63_e}}
  \subfigure[$\hat{\rho}^{200,10}_{\epsilon,5000}$]{\includegraphics[width=0.24\textwidth,height=5.5cm]{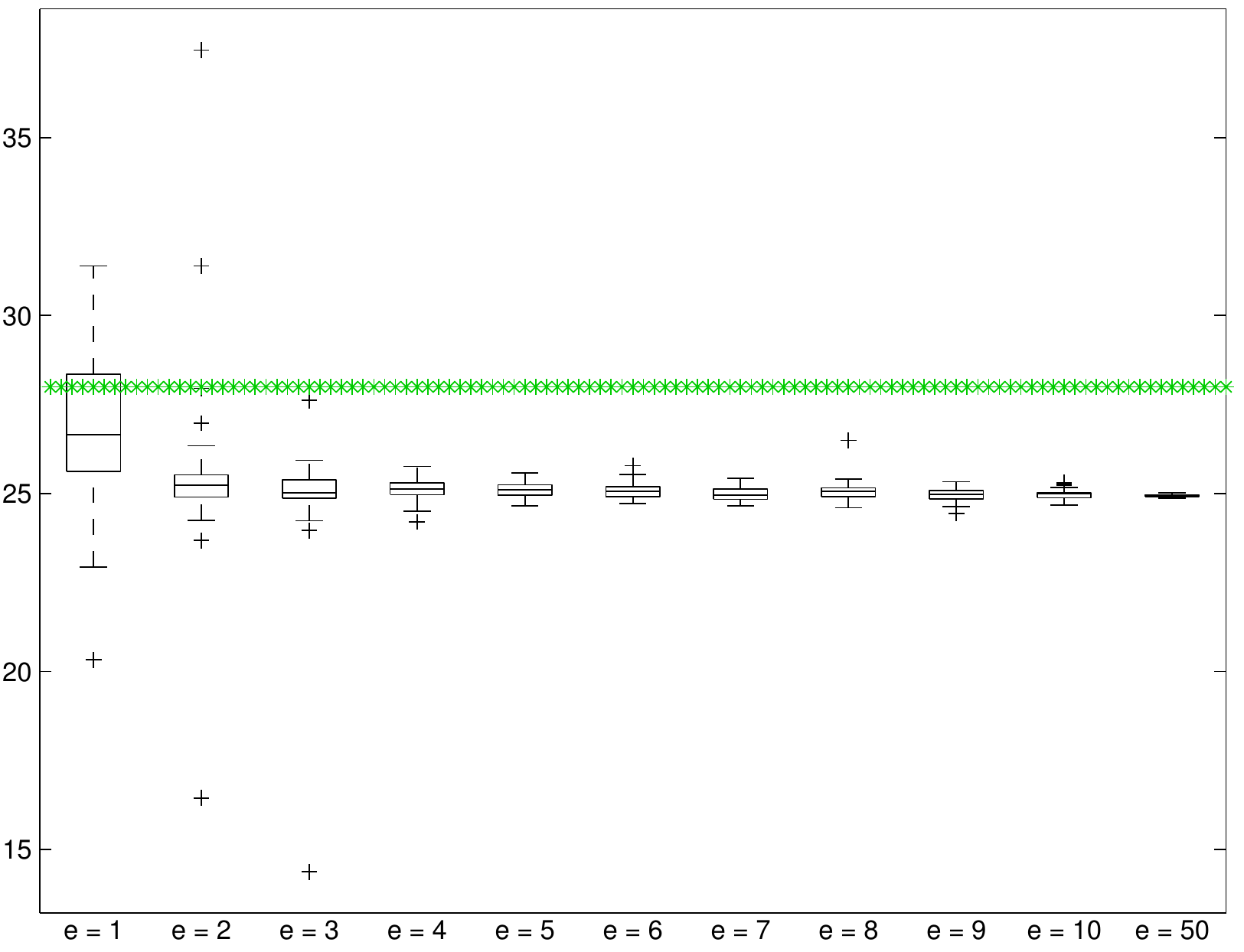}\label{fig:l63_boxplot_rho_e}}\\
  \subfigure[$\hat{\kappa}^{200,10}_{\epsilon,5000}$]{\includegraphics[width=0.24\textwidth,height=5.5cm]{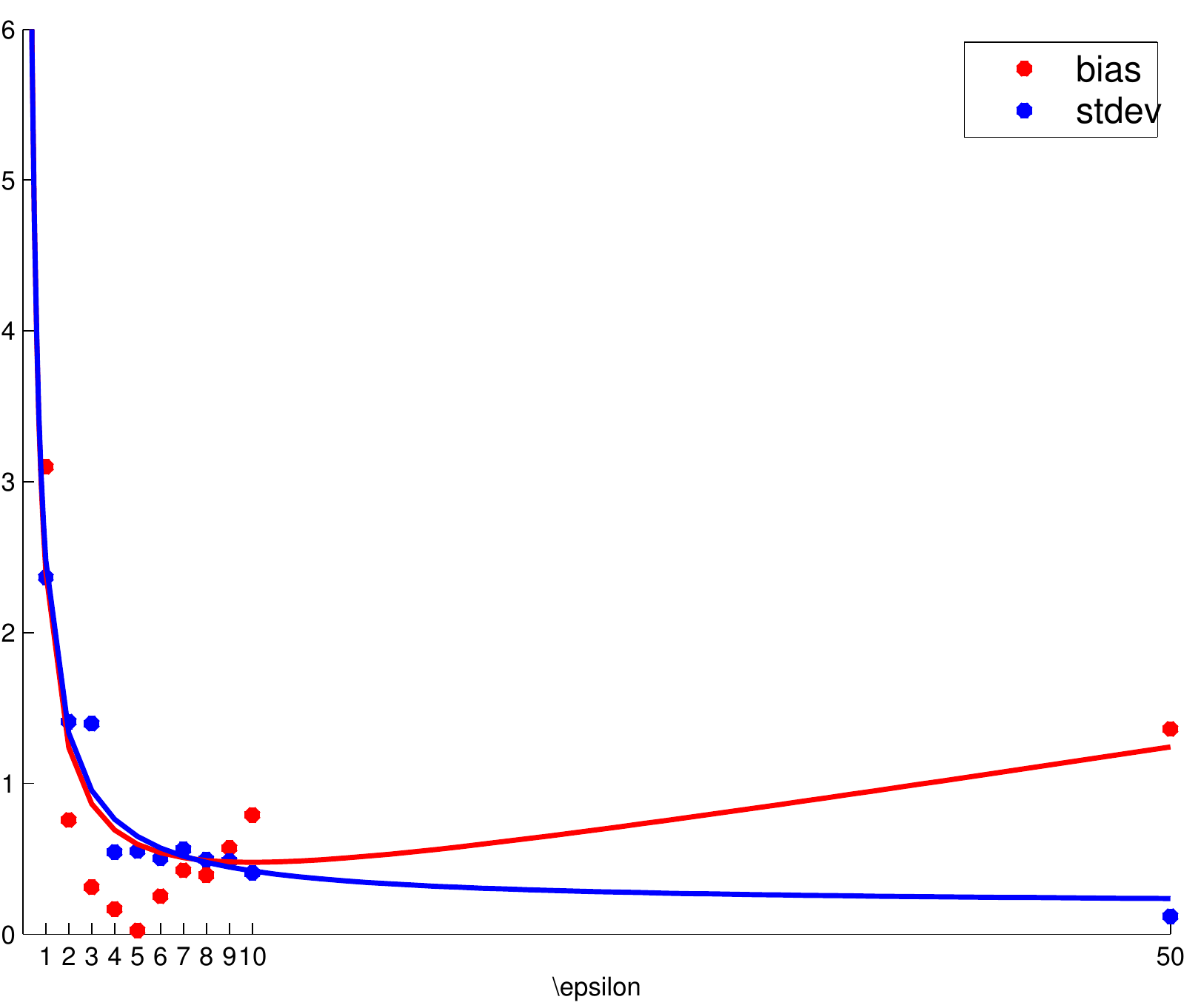}\label{fig:l63_biasvar_kappa_e}}
  \subfigure[$\hat{\sigma}^{200,10}_{\epsilon,5000}$]{\includegraphics[width=0.24\textwidth,height=5.5cm]{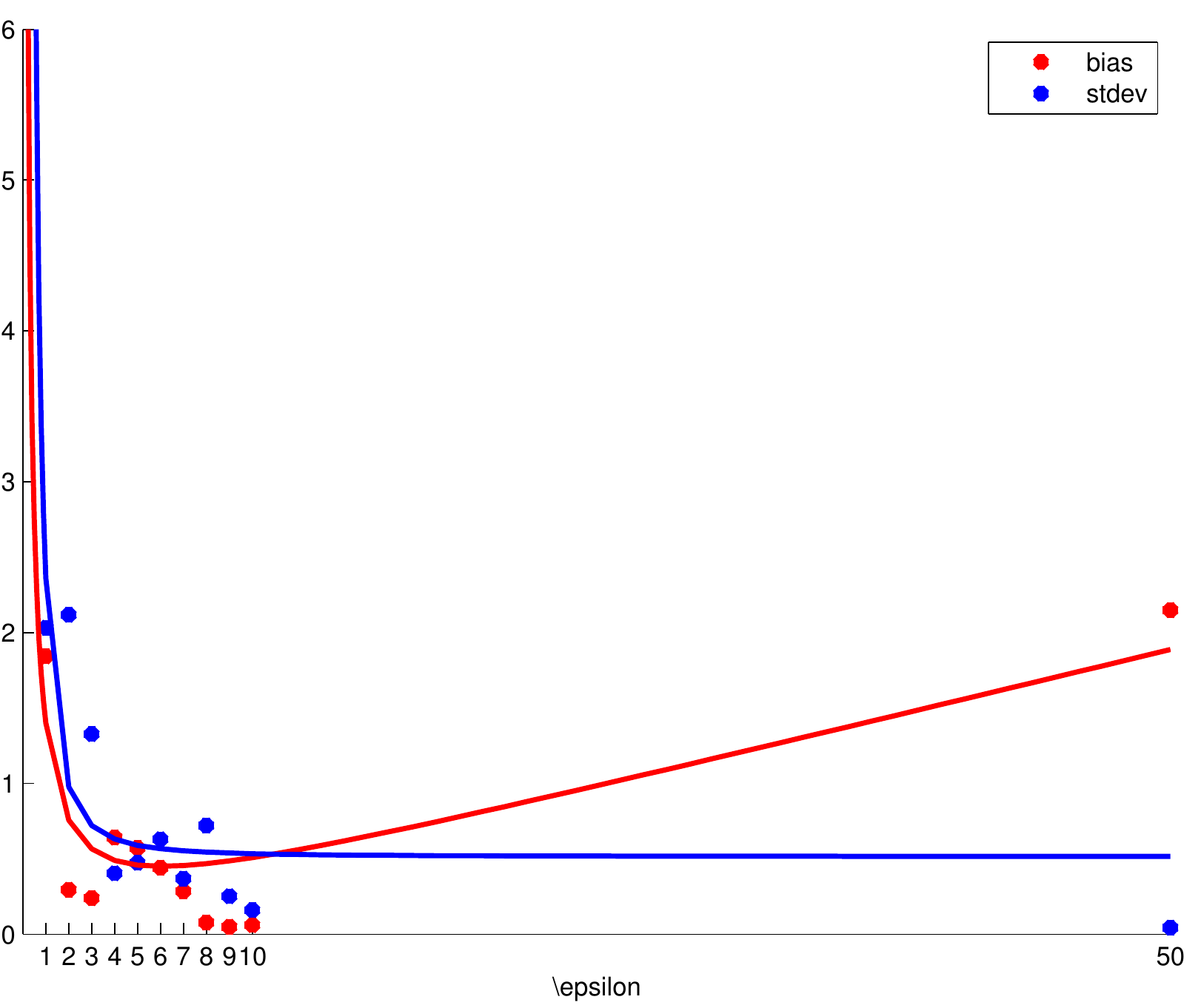}\label{fig:l63_biasvar_sigma_e}}
  \subfigure[$\hat{\sigma}^{200,10}_{63_{\epsilon,5000}}$]{\includegraphics[width=0.24\textwidth,height=5.5cm]{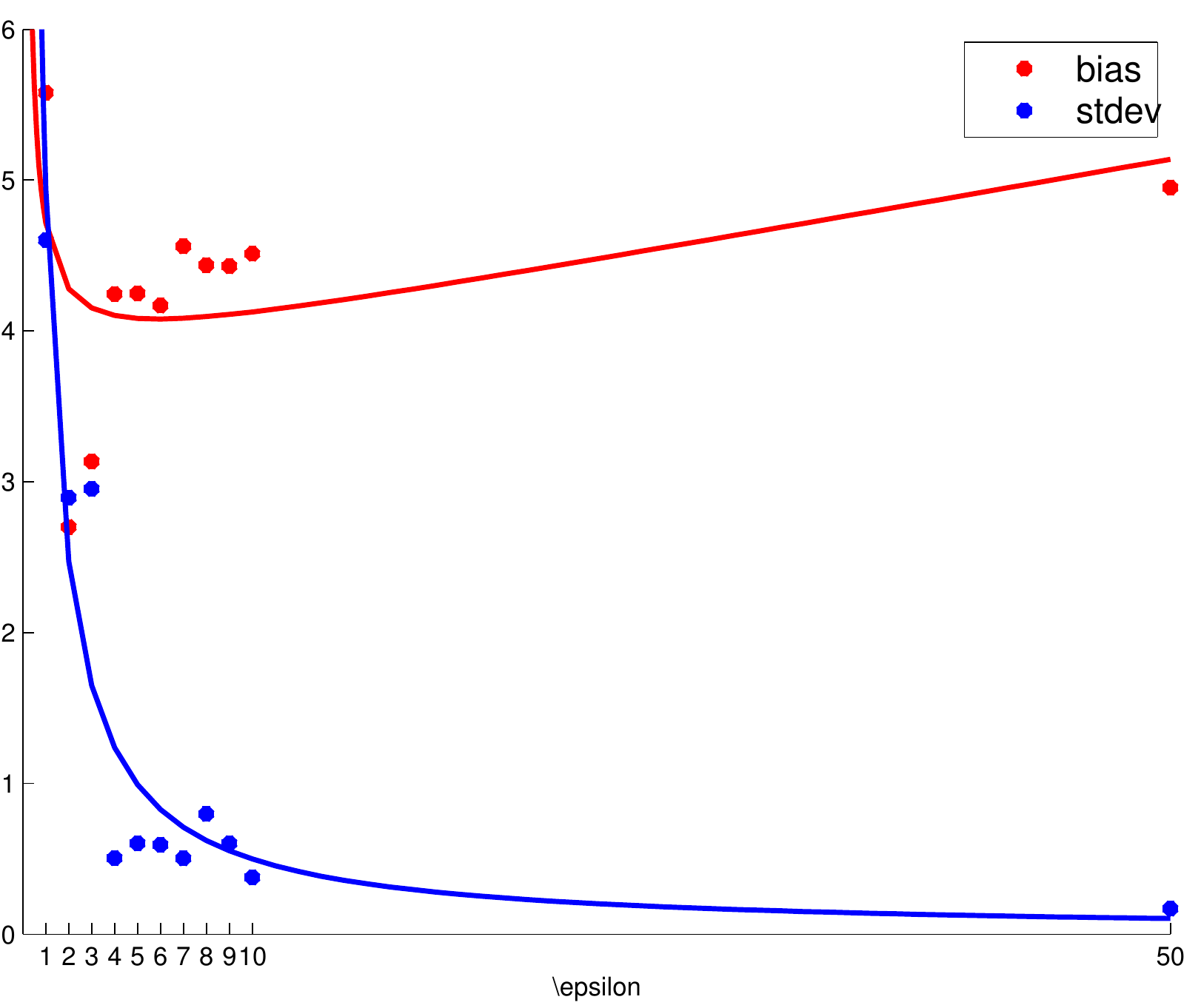}\label{fig:l63_biasvar_sigma_63_e}}
  \subfigure[$\hat{\rho}^{200,10}_{\epsilon,5000}$]{\includegraphics[width=0.24\textwidth,height=5.5cm]{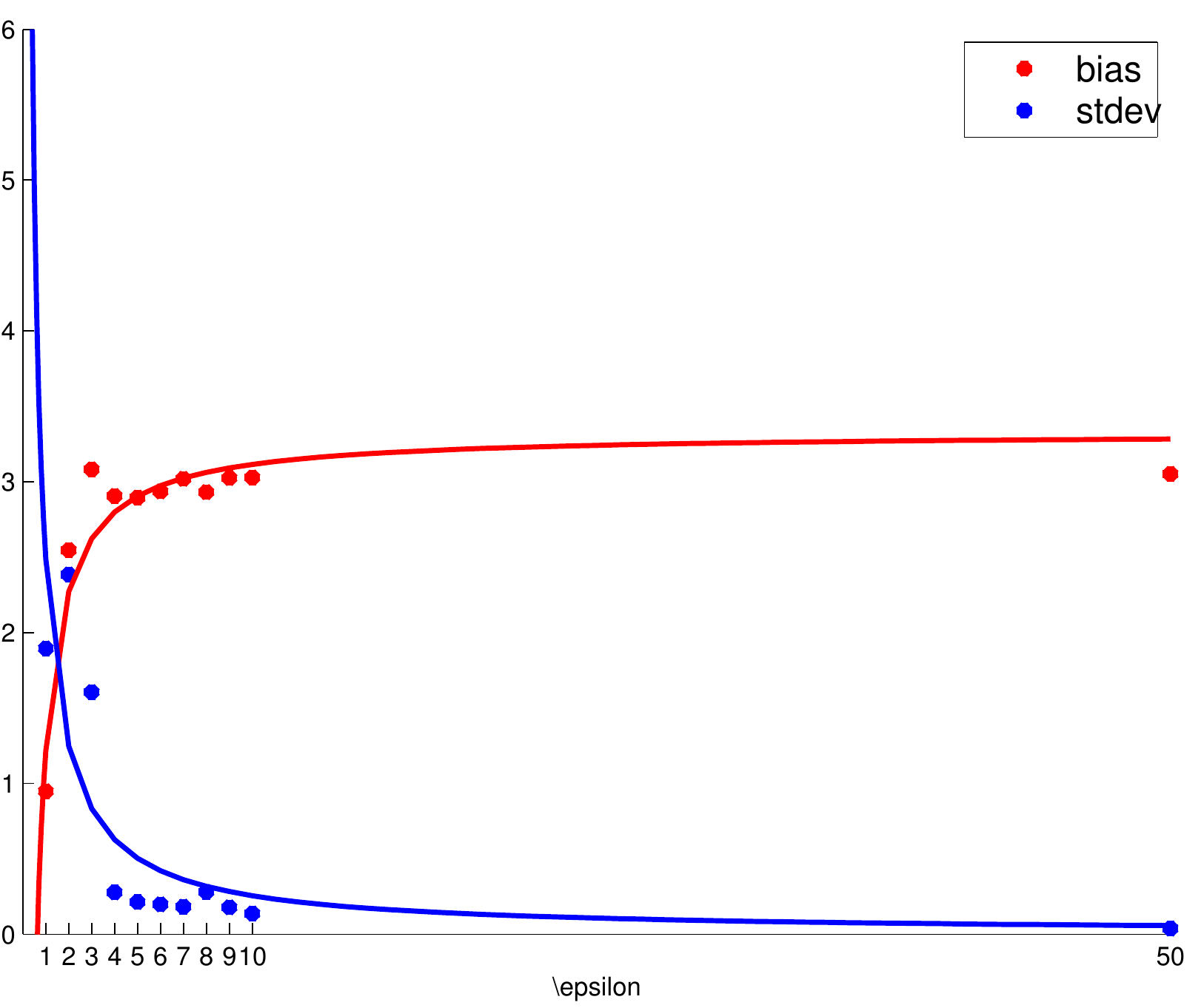}\label{fig:l63_biasvar_rho_e}}
  \caption{ $\widehat{\theta}^{200,10}_{\epsilon,5000}$ when estimating $\theta=(\kappa,\sigma,\sigma_{63},\rho)$ of the Lorenz '63 HMM
, using ABC-SMC with values of $\epsilon\in\{1, 2, 3, \dots, 10, 50\}$. Figures \ref{fig:l63_boxplot_kappa_e}-\ref{fig:l63_boxplot_rho_e} show the MC biases and their curves of non-linear least squared-error proportional to $\epsilon+\frac{1}{\epsilon}$ in red, and the MC standard deviations with their curves of non-linear least squared-error proportional to $\frac{1}{\epsilon}$ in blue. }
  \label{fig:l63_e}
\end{figure}

\section{Summary}\label{sec:summ}

In this article we have presented a technique to perform online parameter estimation using ABC-SMC and SPSA for HMMs. This
is useful for models where the state-dimension is high and the parameter and observations are of moderate dimension. In addition,
it is required when the conditional density of the observations given the hidden state is intractable.

Some future work is as follows.
The representation in Remark \ref{rem:collapsed_abc} can be potentially useful for alternative online parameter estimation techniques,
other than using SPSA. In \cite{dds} we are investigating the use of the online EM algorithm \cite{yidrlim} and any potential benefit that it may have over the ideas in
this paper. We have remarked that one drawback of the SMC algorithm implemented is its inability to deal with small $\epsilon$. Two
potential ways to proceed are as follows. One is to introduce a further approximation by the expectation-propagation algorithm
(as in \cite{chopin}) and potentially removing SMC altogether. The other is to consider more advanced SMC approaches such as \cite{delm:06}
and how this might help one reduce $\epsilon$; this is an area of ongoing research. We are also considering ABC approximations in the scenario
of deterministic dynamics for the hidden state; these models have wide application in applied mathematics as filtering initial conditions of partial differential equations.

\subsubsection*{Acknowledgements}
The second author was funded by an MOE grant and acknowledges useful conversations with David Nott.
We also acknowledge useful conversations with Sumeetpal Singh.

\appendix

\section{Notations}

We introduce a round of notations. As our analysis will rely
upon that in \cite{vlad} our notations will follow that article. 
It is remarked that under our assumptions, one can establish the same assumptions as in \cite{vlad}. Moreover, the time-inhomogenous upper-bounds in that paper
can be made time-homogenous (albeit less tight) under our assumptions. In addition, our proof strategy follows ideas in \cite{andrieu}.

$\mathcal{B}_b(\mathsf{X})$ is the class of bounded and real-valued measurable functions on $\mathsf{X}$.
Throughout, for $\varphi\in\mathcal{B}_b(\mathsf{X})$, $\|\varphi\|_{\infty}:=\sup_{x\in\mathsf{X}}|\varphi(x)|$.
For $\varphi\in\mathcal{B}_b(\mathsf{X})$ and any operator $Q:\mathsf{X}\rightarrow\mathcal{M}(\mathsf{X})$, $Q(\varphi)(x) := \int_{\mathsf{X}}\varphi(y)Q(x,dy)$. In addition for $\mu_{\theta}\in\mathcal{M}(\mathsf{X})$,
$\mu_{\theta} Q(\varphi) := \int_{\mathsf{X}}\mu_{\theta}(dx) Q(\varphi)(x)$.

We introduce the non-negative operator:
$$
R_{\theta,n}(x,dx') := g_{\theta}(y_n|x') f_{\theta}(x'|x) dx'
$$ 
with the ABC equivalent $R_{\theta,\epsilon,n}(x,dx') := g_{\theta,\epsilon}(y_n|x') f_{\theta}(x'|x) dx'$, $
g_{\theta,\epsilon}(y|x) = \int_{A_{\epsilon,y}}g(u|x)dy/\int_{A_{\epsilon,y}}dy$. To keep consistency with \cite{vlad}
and to allow the reader to follow the proofs, we note that the filter at time $n\geq 0$, $F_{\theta}^n(\mu_{\theta})$ (resp.~ABC filter, at time $n$, $F_{\theta,\epsilon}^n(\mu_{\theta})$) is exactly, with initial distribution $\mu_{\theta}\in\mathcal{P}(\mathsf{X})$ and test function
$\varphi\in\mathcal{B}_b(\mathsf{X})$ 
$$
F_{\theta}^n(\mu_{\theta})(\varphi) = \frac{\mu_{\theta} R_{1,n,\theta}(\varphi)}{\mu_{\theta} R_{1,n,\theta}(1)}
$$
resp.
$$
F_{\theta,\epsilon}^n(\mu_{\theta})(\varphi) = \frac{\mu_{\theta} R_{1,n,\theta,\epsilon}(\varphi)}{\mu_{\theta} R_{1,n,\theta,\epsilon}(1)}
$$
where 
$F_{\theta}^0(\mu_{\theta})=F^0_{\theta,\epsilon}(\mu_{\theta})=\mu_{\theta}$, 
$R_{1,n,\theta}(\varphi)(x_0) = \int \prod_{k=1}^n R_{k,\theta}(x_{k-1},dx_k) \varphi(x_n)$. 
In addition, we write the filter derivatives as $\widetilde{F}_{\theta}^n(\mu_{\theta},\widetilde{\mu_{\theta}})(\varphi)$,
$\widetilde{F}_{\theta,\epsilon}^n(\mu_{\theta},\widetilde{\mu_{\theta}})(\varphi)$
where the second argument is the gradient of the initial measure.

The following operators will be used below, for $n\geq 1$:
\begin{eqnarray}
\widetilde{G}^n(\mu_{\theta},\widetilde{\mu_{\theta}})(\varphi) & := & (\mu_{\theta} R_{1,n,\theta}(1))^{-1}[\widetilde{\mu_{\theta}} R_{1,n,\theta}(\varphi)-\widetilde{\mu_{\theta}} R_{1,n,\theta}(1)
F^n_{\theta}(\mu_{\theta})(\varphi)]\label{eq:gtdefn}\\
\widetilde{H}^n(\mu_{\theta})(\varphi) & := & 
F_{\theta}^{n-1}(\mu_{\theta}) R_{n,\theta}(1)^{-1}[F_{\theta}^{n-1}(\mu_{\theta})\widetilde{R}_{n,\theta}(\varphi) - 
F_{\theta}^{n-1}(\mu_{\theta})\widetilde{R}_{n,\theta}(1) F_{\theta}^n(\mu_{\theta})(\varphi)
]\label{eq:htdefn}
\end{eqnarray}
with the convention $\widetilde{G}^0(\mu_{\theta},\widetilde{\mu_{\theta}})(\varphi)=\widetilde{\mu_{\theta}}$. In addition, we set
$$
\widetilde{G}^{(n)}(\mu_{\theta},\widetilde{\mu_{\theta}})(\varphi)  :=  (\mu_{\theta} R_{n,\theta}(1))^{-1}[\widetilde{\mu_{\theta}} R_{n,\theta}(\varphi)-\widetilde{\mu_{\theta}} R_{n,\theta}(1)
F^{(n)}_{\theta}(\mu_{\theta})(\varphi)].
$$
where $F_{\theta}^{(n)}(\mu_{\theta}) = \mu_{\theta} R_{n,\theta}/\mu_{\theta} R_{n,\theta}(1)$.
Finally, an important notational convention is as follows. Throughout we use $C$ to denote a constant whose value may change from line-to-line in the calculations.
This constant will typically not depend upon important parameters such as $\epsilon$ and $n$ and any important dependencies will be highlighted.

\section{Bias of the Log-Likelihood}\label{app:log_like}

\begin{proof}[Proof of Proposition \ref{prop:ll_bias}]
We begin with the equality
\begin{equation}
\log(p_{\theta}(y_{1:n})) - \log(p_{\theta,\epsilon}(y_{1:n})) = 
\sum_{k=1}^n \bigg(\log(p_{\theta}(y_k|y_{1:k-1})) - \log(p_{\theta,\epsilon}(y_k|y_{1:k-1}))\bigg)
\label{eq:master_eq}
\end{equation}
with, for $1\leq k \leq n$
\begin{eqnarray*}
p_{\theta}(y_k|y_{1:k-1}) & = & \int_{\mathsf{X}^{2}} g_{\theta}(y_k|x_k) f_{\theta}(x_k|x_{k-1}) F_{\theta}^{k-1}(\mu_{\theta})(dx_{k-1})dx_{k}\\
p_{\theta,\epsilon}(y_k|y_{1:k-1}) & = & \int_{\mathsf{X}^{2}} g_{\theta}^{\epsilon}(y_k|x_k) f_{\theta}(x_k|x_{k-1}) F_{\theta,\epsilon}^{k-1}(\mu_{\theta})(dx_{k-1}) dx_{k}.
\end{eqnarray*}
 We will consider each summand in \eqref{eq:master_eq}. The case $k\geq 2$ is only considered; the scenario $k=1$ will follow a similar and simpler argument.

Using the inequality
$|\log(x)-\log(y)|\leq |x-y|/(x\wedge y)$ for every $x,y>0$ we have
$$
|\log(p_{\theta}(y_k|y_{1:k-1})) - \log(p_{\theta,\epsilon}(y_k|y_{1:k-1}))| \leq \frac{|p_{\theta}(y_k|y_{1:k-1})-p_{\theta,\epsilon}(y_k|y_{1:k-1})|}{p_{\theta}(y_k|y_{1:k-1})\wedge p_{\theta,\epsilon}(y_k|y_{1:k-1})}.
$$
Note that
$$
p_{\theta}(y_k|y_{1:k-1})\wedge p_{\theta}(y_k|y_{1:k-1}) = 
$$
\begin{equation}
\int_{\mathsf{X}^2} g_{\theta}(y_k|x_k) f_{\theta}(x_k|x_{k-1}) F_{\theta}^{k-1}(\mu_{\theta})(dx_{k-1})dx_k 
\wedge \int_{\mathsf{X}^{2}} g_{\theta}^{\epsilon}(y_k|x_k) f_{\theta}(x_k|x_{k-1}) F_{\theta,\epsilon}^{k-1}(\mu_{\theta})(dx_{k-1}) dx_{k}
\geq
C >0
\label{eq:p_lower_bound}
\end{equation}
where we have applied (A\ref{hyp:like_bound}) and $C$ does not depend upon $\epsilon$. Thus we consider
$$
|p_{\theta,\epsilon}(y_k|y_{1:k-1})-p_{\theta}(y_k|y_{1:k-1})| = 
$$
$$
|\int_{\mathsf{X}^2} g_{\theta}(y_k|x_k) f_{\theta}(x_k|x_{k-1}) F_{\theta}^{k-1}(\mu_{\theta})(dx_{k-1})dx_k - 
\int_{\mathsf{X}^2} g_{\theta,\epsilon}(y_k|x_k) f_{\theta}(x_k|x_{k-1}) F_{\theta,\epsilon}^{k-1}(\mu_{\theta})(dx_{k-1})dx_k|.
$$
The R.H.S.~can be upper-bounded by the sum of
$$
|\int_{\mathsf{X}^2} [g_{\theta}(y_k|x_k) - g_{\theta,\epsilon}(y_k|x_k)]f_{\theta}(x_k|x_{k-1}) F_{\theta}^{k-1}(\mu_{\theta})(dx_{k-1})dx_k |
$$
and
$$
|\int_{\mathsf{X}^2} g_{\theta,\epsilon}(y_k|x_k) f_{\theta}(x_k|x_{k-1}) [F_{\theta,\epsilon}^{k-1}(\mu_{\theta})(dx_{k-1})-F_{\theta,\epsilon}^{k-1}(\mu_{\theta})(dx_{k-1}])dx_k|.
$$
The first expression can be dealt with by using (A\ref{hyp:like_cont}), which implies
\begin{equation}
\sup_{x\in\mathsf{X}}|g_{\theta,\epsilon}(y_k|x)-g_{\theta,\epsilon}(y_k|x)|\leq C\epsilon
\label{eq:g_continuity_epsilon}.
\end{equation}
 The second expression can be controlled by \cite[Theorem 2]{jasra}: 
\begin{equation}
\sup_{k\geq 1} \|F_{\theta}^{k-1}(\mu_{\theta})-F_{\theta,\epsilon}^{k-1}(\mu_{\theta})\| \leq C\epsilon
\label{eq:abc_filter_stability}
\end{equation}
to yield that
\begin{equation}
|p_{\theta,\epsilon}(y_k|y_{1:k-1})-p_{\theta}(y_k|y_{1:k-1})| \leq C\epsilon\label{eq:p_control}.
\end{equation}
One can thus conclude.
\end{proof}

\section{Bias of the Gradient of the Log-Likelihood}\label{app:log_like_grad}

\begin{proof}[Proof of Theorem \ref{prop:ll_bias}]
We have that
$$
\nabla\bigg(\log p_{\theta}(y_{1:n}) - \log p_{\theta,\epsilon}(y_{1:n})\bigg) = 
\nabla\bigg\{\sum_{k=1}^n \bigg(\log[p_{\theta}(y_k|y_{1:k-1}) - \log[p_{\theta,\epsilon}(y_k|y_{1:k-1}) \bigg)\bigg\}.
$$
It then follows that 
$$
\nabla\bigg(\log p_{\theta}(y_{1:n}) - \log p_{\theta,\epsilon}(y_{1:n})\bigg) = 
$$
\begin{equation}
\sum_{k=1}^n \bigg( \frac{[\nabla p_{\theta}(y_k|y_{1:k-1}) - \nabla p_{\theta,\epsilon}(y_k|y_{1:k-1}) ]}{p_{\theta}(y_k|y_{1:k-1})} + 
\frac{\nabla p_{\theta,\epsilon}(y_k|y_{1:k-1})}{p_{\theta}(y_k|y_{1:k-1})p_{\theta,\epsilon}(y_k|y_{1:k-1})}[p_{\theta,\epsilon}(y_k|y_{1:k-1})-p_{\theta}(y_k|y_{1:k-1})]
\bigg).
\label{eq:ll_main_decomp}
\end{equation}
We will deal with the two terms on the R.H.S.~of \eqref{eq:ll_main_decomp} in turn. The scenario $k\geq 2$ is only considered; the case $k=1$ follows a similar and simpler argument.

First starting with summand
$$
\frac{[\nabla p_{\theta}(y_k|y_{1:k-1}) - \nabla p_{\theta,\epsilon}(y_k|y_{1:k-1}) ]}{p_{\theta}(y_k|y_{1:k-1})}.
$$
Noting \eqref{eq:p_lower_bound}, we need only upper-bound the $\mathbb{L}_1$ norm of the following expression
\begin{equation}
\int_{\mathsf{X}^2} \nabla\{g_{\theta}(y_k|x_k)\} f_{\theta}(x_k|x_{k-1}) F_{\theta}^{k-1}(\mu_{\theta})(dx_{k-1})dx_k - 
\int_{\mathsf{X}^2} \nabla\{g_{\theta,\epsilon}(y_k|x_k)\} f_{\theta}(x_k|x_{k-1}) F_{\theta,\epsilon}^{k-1}(\mu_{\theta})(dx_{k-1})dx_k
\label{eq:ll_1_1}
\end{equation}
\begin{equation}
+ \int_{\mathsf{X}^2} g_{\theta}(y_k|x_k) \nabla\{f_{\theta}(x_k|x_{k-1})\} F_{\theta}^{k-1}(\mu_{\theta})(dx_{k-1})dx_k - 
\int_{\mathsf{X}^2} g_{\theta,\epsilon}(y_k|x_k) \nabla\{f_{\theta}(x_k|x_{k-1})\} F_{\theta,\epsilon}^{k-1}(\mu_{\theta})(dx_{k-1})dx_k
\label{eq:ll_1_2}
\end{equation}
\begin{equation}
+ \int_{\mathsf{X}^2} g_{\theta}(y_k|x_k) f_{\theta}(x_k|x_{k-1}) \widetilde{F}_{\theta}^{k-1}(\mu_{\theta},\widetilde{\mu_{\theta}})(dx_{k-1})dx_k - 
\int_{\mathsf{X}^2} g_{\theta,\epsilon}(y_k|x_k) f_{\theta}(x_k|x_{k-1}) \widetilde{F}_{\theta,\epsilon}^{k-1}(\mu_{\theta},\widetilde{\mu_{\theta}})(dx_{k-1})dx_k
\label{eq:ll_1_3}.
\end{equation}
We start with \eqref{eq:ll_1_1}. Using (A\ref{hyp:like_grad_cont}) we can establish that for each $k\geq 1$ 
\begin{equation}
\sup_{x\in\mathsf{X}}|\nabla\{g_{\theta}(y_k|x_k)\}-\nabla\{g_{\theta,\epsilon}(y_k|x_k)\}| \leq C \epsilon
\label{eq:grad_g_cont}
\end{equation}
where $C$ does not depend upon $k,\epsilon$. Hence 
$$
|\int_{\mathsf{X}^2} [\nabla\{g_{\theta}(y_k|x_k)\}-\nabla\{g_{\theta,\epsilon}(y_k|x_k)\}]f_{\theta}(x_k|x_{k-1}) F_{\theta}^{k-1}(\mu_{\theta})(dx_{k-1})dx_k|
\leq C \epsilon.
$$ 
Then we note that by \cite[Theorem 2]{jasra} (see \eqref{eq:abc_filter_stability}) and (A\ref{hyp:like_grad_bound})
$$
|\int_{\mathsf{X}^2} \nabla\{g_{\theta,\epsilon}(y_k|x_k)\} f_{\theta}(x_k|x_{k-1}) [F_{\theta}^{k-1}(\mu_{\theta})(dx_{k-1}) - F_{\theta,\epsilon}^{k-1}(\mu_{\theta})(dx_{k-1})]dx_k| \leq C\epsilon
$$
Thus we have shown that
$$
|\int_{\mathsf{X}^2} \nabla\{g_{\theta}(y_k|x_k)\} f_{\theta}(x_k|x_{k-1}) F_{\theta}^{k-1}(\mu_{\theta})(dx_{k-1})dx_k - 
\int_{\mathsf{X}^2} \nabla\{g_{\theta,\epsilon}(y_k|x_k)\} f_{\theta}(x_k|x_{k-1}) F_{\theta,\epsilon}^{k-1}(\mu_{\theta})(dx_{k-1})dx_k|
\leq C\epsilon.
$$
Now, moving onto \eqref{eq:ll_1_2}, by \eqref{eq:g_continuity_epsilon}
we have
$$
|\int_{\mathsf{X}^2} [g_{\theta}(y_k|x_k)-g_{\theta,\epsilon}(y_k|x_k)]\nabla\{f_{\theta}(x_k|x_{k-1})\} F_{\theta}^{k-1}(\mu_{\theta})(dx_{k-1})dx_k|
\leq C \epsilon.
$$
and can again use \cite[Theorem 2]{jasra} (i.e.~\eqref{eq:abc_filter_stability}) to deduce that
$$
|\int_{\mathsf{X}^2} g_{\theta,\epsilon}(y_k|x_k) \nabla\{f_{\theta}(x_k|x_{k-1})\} [F_{\theta}^{k-1}(\mu_{\theta})(dx_{k-1})-F_{\theta,\epsilon}^{k-1}(\mu_{\theta})(dx_{k-1})]dx_k| \leq C\epsilon
$$
and thus that 
$$
|
\int_{\mathsf{X}^2} g_{\theta}(y_k|x_k) \nabla\{f_{\theta}(x_k|x_{k-1})\} F_{\theta}^{k-1}(\mu_{\theta})(dx_{k-1})dx_k - 
\int_{\mathsf{X}^2} g_{\theta,\epsilon}(y_k|x_k) \nabla\{f_{\theta}(x_k|x_{k-1})\} F_{\theta,\epsilon}^{k-1}(\mu_{\theta})(dx_{k-1})dx_k
| \leq C\epsilon
$$
which upper-bounds the expression in \eqref{eq:ll_1_2}. We now move onto \eqref{eq:ll_1_3}, which upper-bounded by
$$
|\int_{\mathsf{X}^2} [g_{\theta}(y_k|x_k)-g_{\theta,\epsilon}(y_k|x_k)]f_{\theta}(x_k|x_{k-1}) \widetilde{F}_{\theta}^{k-1}(\mu_{\theta},\widetilde{\mu_{\theta}})(dx_{k-1})dx_k| + 
$$
$$
|\int_{\mathsf{X}^2} g_{\theta,\epsilon}(y_k|x_k) f_{\theta}(x_k|x_{k-1}) [\widetilde{F}_{\theta}^{k-1}(\mu_{\theta},\widetilde{\mu_{\theta}})(dx_{k-1})-\widetilde{F}_{\theta,\epsilon}^{k-1}(\mu_{\theta},\widetilde{\mu_{\theta}})(dx_{k-1})]dx_k|.
$$
For the first expression, we can write:
$$
(\sup_{x\in\mathsf{X}}|g_{\theta}(y_k|x) - g_{\theta,\epsilon}(y_k|x)|) |\int_{\mathsf{X}}\bigg(\int_{\mathsf{X}} 
\frac{[g_{\theta}(y_k|x_k) - g_{\theta,\epsilon}(y_k|x_k) ]}{(\sup_{x\in\mathsf{X}}|g_{\theta}(y_k|x) - g_{\theta,\epsilon}(y_k|x)|)}f_{\theta}(x_k|x_{k-1})dx_k\bigg) \widetilde{F}_{\theta}^{k-1}(\mu_{\theta},\widetilde{\mu_{\theta}})(dx_{k-1})|.
$$
Then we can apply \eqref{eq:g_continuity_epsilon} and, noting that 
$$
\bigg(\int_{\mathsf{X}} 
\frac{[g_{\theta}(y_k|x_k) - g_{\theta,\epsilon}(y_k|x_k) ]}{(\sup_{x\in\mathsf{X}}|g_{\theta}(y_k|x) - g_{\theta,\epsilon}(y_k|x)|)}f_{\theta}(x_k|x_{k-1})dx_k\bigg) \leq 1
$$
one can also use Lemma \ref{lem:filter_deriv_upper} to deduce that
$$
|\int_{\mathsf{X}^2} [g_{\theta}(y_k|x_k)-g_{\theta,\epsilon}(y_k|x_k)]f_{\theta}(x_k|x_{k-1}) \widetilde{F}_{\theta}^{k-1}(\mu_{\theta},\widetilde{\mu_{\theta}})(dx_{k-1})dx_k| \leq C(1+\|\widetilde{\mu_{\theta}}\|)\epsilon.
$$
Then, one can easily apply Theorem \ref{theo:filt_deriv_bias} to show that
$$
|\int_{\mathsf{X}^2} g_{\theta,\epsilon}(y_k|x_k) f_{\theta}(x_k|x_{k-1}) [\widetilde{F}_{\theta}^{k-1}(\mu_{\theta},\widetilde{\mu_{\theta}})(dx_{k-1})-\widetilde{F}_{\theta,\epsilon}^{k-1}(\mu_{\theta},\widetilde{\mu_{\theta}})(dx_{k-1})]dx_k|
\leq C(2+\|\widetilde{\mu_{\theta}}\|)\epsilon.
$$
Thus we have upper-bounded the $\mathbb{L}_1-$norm of the sum of the expressions \eqref{eq:ll_1_1}-\eqref{eq:ll_1_3} and  we have established that
\begin{equation}
\frac{[\nabla p_{\theta}(y_k|y_{1:k-1}) - \nabla p_{\theta,\epsilon}(y_k|y_{1:k-1}) ]}{p_{\theta}(y_k|y_{1:k-1})}
\leq C(2+\|\widetilde{\mu_{\theta}}\|)\epsilon.
\label{eq:first_ll_bound}
\end{equation}

Moving onto the second summand on the R.H.S.~of \eqref{eq:ll_main_decomp}, 
$$
\frac{\nabla p_{\theta,\epsilon}(y_k|y_{1:k-1})}{p_{\theta}(y_k|y_{1:k-1})p_{\theta,\epsilon}(y_k|y_{1:k-1})}[p_{\theta,\epsilon}(y_k|y_{1:k-1})-p_{\theta}(y_k|y_{1:k-1}).
$$
By \eqref{eq:p_control}, we need only consider upper-bounding, in $\mathbb{L}_1$, $\nabla p_{\theta,\epsilon}(y_k|y_{1:k-1})$. This can be decomposed into the sum of three expressions:
$$
\int_{\mathsf{X}^2} \nabla\{g_{\theta,\epsilon}(y_k|x_k)\} f_{\theta}(x_k|x_{k-1}) F_{\theta,\epsilon}^{k-1}(\mu_{\theta})(dx_{k-1})dx_k
$$
$$
\int_{\mathsf{X}^2} g_{\theta,\epsilon}(y_k|x_k) \nabla\{f_{\theta}(x_k|x_{k-1})\} F_{\theta,\epsilon}^{k-1}(\mu_{\theta})(dx_{k-1})dx_k
$$
and
$$
\int_{\mathsf{X}^2} g_{\theta,\epsilon}(y_k|x_k) f_{\theta}(x_k|x_{k-1}) \widetilde{F}_{\theta,\epsilon}^{k-1}(\mu_{\theta},\widetilde{\mu_{\theta}})(dx_{k-1})dx_k.
$$
As $\nabla\{g_{\theta,\epsilon}(y_k|x_k)\}$ and $g_{\theta,\epsilon}(y_k|x_k) \nabla\{f_{\theta}(x_k|x_{k-1})\}$ are upper-bounded as well as $\mathsf{X}$ being compact the first two expressions are upper-bounded in $\mathbb{L}_1$. In addition as $\int_{\mathsf{X}} g_{\theta,\epsilon}(y_k|x_k) f_{\theta}(x_k|x_{k-1})dx_k$ is upper-bounded, we can apply Lemma \ref{lem:filter_deriv_upper} to see that the third expression is upper-bounded in $\mathbb{L}_1$. Hence, we have shown that
\begin{equation}
\bigg|\frac{\nabla p_{\theta,\epsilon}(y_k|y_{1:k-1})}{p_{\theta}(y_k|y_{1:k-1})p_{\theta,\epsilon}(y_k|y_{1:k-1})}[p_{\theta,\epsilon}(y_k|y_{1:k-1})-p_{\theta}(y_k|y_{1:k-1})]\bigg|
\leq C(1+\|\widetilde{\mu_{\theta}}\|)\epsilon\label{eq:second_ll_bound}.
\end{equation}
Combining the results \eqref{eq:first_ll_bound}-\eqref{eq:second_ll_bound} and noting \eqref{eq:ll_main_decomp} we can conclude.
\end{proof}

\section{Bias of the Gradient of the Filter}\label{app:filt_deriv}

\begin{theorem}\label{theo:filt_deriv_bias}
Assume (A1-5). Then there exist a $C<+\infty$ such that for any $n\geq 1$, $\mu_{\theta}\in\mathcal{P}(\mathsf{X})$, $\widetilde{\mu_{\theta}}\in\mathcal{M}(\mathsf{X})$, $\epsilon >0$,
$\theta\in\Theta$:
$$
\|\widetilde{F}_{\theta}^n(\mu_{\theta},\widetilde{\mu_{\theta}})-\widetilde{F}_{\theta,\epsilon}^n(\mu_{\theta},\widetilde{\mu_{\theta}})\| \leq C\epsilon(2+\|\widetilde{\mu_{\theta}}\|).
$$
\end{theorem}

\begin{proof}
We have the following telescoping sum decomposition (e.g.~\cite{delmoral}) for the differences in the filters, with $\varphi\in\mathcal{B}_b(\mathsf{X})$:
$$
F_{\theta}^n(\mu_{\theta})(\varphi)-F_{\theta,\epsilon}^n(\mu_{\theta})(\varphi) = 
\sum_{p=1}^n\bigg[ F_{\theta}^{n-p+1,n}(F_{\theta,\epsilon}^{n-p}(\mu_{\theta}))(\varphi) - 
F_{\theta}^{n-p+2,n}(F_{\theta,\epsilon}^{n-p+1}(\mu_{\theta}))(\varphi)
\bigg]
$$
where we are using the notation $F_{\theta}^{q,n}(\mu_{\theta})(\varphi) = \frac{\mu_{\theta} R_{q,n,\theta}(\varphi)}{\mu_{\theta} R_{q,n,\theta}(1)}$, for $1\leq q\leq n$. 
Hence, taking gradients and swapping the order of summation and differentiation
we have and omitting the second arguments of $\widetilde{F}$ on the R.H.S.~(to reduce the notational burden)
\begin{eqnarray}
\widetilde{F}_{\theta}^n(\mu_{\theta},\widetilde{\mu_{\theta}})(\varphi)-\widetilde{F}_{\theta,\epsilon}^n(\mu_{\theta},\widetilde{\mu_{\theta}})(\varphi) & = & 
\sum_{p=1}^n\bigg[ \widetilde{F}_{\theta}^{n-p+2,n}(F_{\theta}^{(n-p+1)}[F_{\theta,\epsilon}^{n-p}(\mu_{\theta})],\widetilde{F}_{\theta}^{(n-p+1)}[F_{\theta,\epsilon}^{n-p}(\mu_{\theta})])(\varphi) - 
\nonumber
\\
& & 
\widetilde{F}_{\theta}^{n-p+2,n}(F_{\theta,\epsilon}^{(n-p+1)}[F_{\theta,\epsilon}^{(n-p)}(\mu_{\theta})],
\widetilde{F}_{\theta,\epsilon}^{(n-p+1)}[F_{\theta,\epsilon}^{(n-p)}(\mu_{\theta})])(\varphi)
\bigg]
\label{eq:filter_deriv_main_decomp}.
\end{eqnarray}
To continue with the proof we will adopt \cite[Lemma 6.4]{vlad}:
$$
\widetilde{F}_{\theta}^n(\mu_{\theta},\widetilde{\mu_{\theta}})(\varphi) = \widetilde{G}_{\theta}^n(\mu_{\theta},\widetilde{\mu_{\theta}})
+ \sum_{q=1}^n \widetilde{G}_{\theta}^{q+1,n}(F_{\theta}^q(\mu_{\theta}),\widetilde{H}^q(\mu_{\theta}))(\varphi)
$$
with  $\widetilde{G}_{\theta}^n$ and $\widetilde{H}^q(\mu_{\theta})$ defined in \eqref{eq:gtdefn}-\eqref{eq:htdefn}
and $\widetilde{G}_{\theta}^{q+1,n}$  similar extension to the notation as for the filter $F_{\theta}^{q+1,n}$ and the convention
$\widetilde{G}_{\theta}^{n+1,n}(\mu_{\theta},\widetilde{\mu_{\theta}}) = \widetilde{\mu_{\theta}}$.
Returning to \eqref{eq:filter_deriv_main_decomp} and again omitting the second arguments of $\widetilde{F}$ on the R.H.S.:
$$
\widetilde{F}_{\theta}^n(\mu_{\theta},\widetilde{\mu_{\theta}})(\varphi)-\widetilde{F}_{\theta,\epsilon}^n(\mu_{\theta},\widetilde{\mu_{\theta}})(\varphi) = 
$$
$$
 \sum_{p=1}^n
\bigg[\widetilde{G}_{\theta}^{n-p+2,n}\{F_{\theta}^{(n-p+1)}(F_{\theta,\epsilon}^{n-p}(\mu_{\theta})),\widetilde{F}_{\theta}^{(n-p+1)}(F_{\theta,\epsilon}^{n-p}(\mu_{\theta}))\}(\varphi) - 
\widetilde{G}_{\theta}^{n-p+2,n}\{F_{\theta,\epsilon}^{(n-p+1)}[F_{\theta,\epsilon}^{n-p}(\mu_{\theta})],\widetilde{F}_{\theta,\epsilon}^{(n-p+1)}[F_{\theta,\epsilon}^{n-p}(\mu_{\theta})]\}(\varphi)~+ 
$$
$$
\sum_{q=n-p+2}^{n} \bigg\{\widetilde{G}_{\theta}^{q+1,n}\{F_{\theta}^{n-p+2,q}[ F_{\theta}^{(n-p+1)}(F_{\theta,\epsilon}^{n-p}(\mu_{\theta})) ]
,\widetilde{H}_{\theta}^{n-p+2,q}[F_{\theta}^{(n-p+1)}(F_{\theta,\epsilon}^{n-p}(\mu_{\theta})) ]\}(\varphi)~- 
$$
\begin{equation}
\widetilde{G}_{\theta}^{q+1,n}\{F_{\theta}^{n-p+2,q}[F_{\theta,\epsilon}^{(n-p+1)}(F_{\theta,\epsilon}^{n-p}(\mu_{\theta}))],
\widetilde{H}_{\theta}^{n-p+2,q}[F_{\theta,\epsilon}^{(n-p+1)}(F_{\theta,\epsilon}^{n-p}(\mu_{\theta}))]\}(\varphi)
\bigg\}
\bigg]\label{eq:main_decomp_filter_+1}.
\end{equation}

We start first with the summand on the R.H.S.~of the second line of \eqref{eq:main_decomp_filter_+1}, which we compactly denote as:
$$
\widetilde{G}_{\theta}^{p-1}\{F_{\theta}[F_{\theta,\epsilon}^{n-p}(\mu_{\theta})],\widetilde{F}_{\theta}[F_{\theta,\epsilon}^{n-p}(\mu_{\theta})]\}(\varphi) - 
\widetilde{G}_{\theta}^{p-1}\{F_{\theta,\epsilon}[F_{\theta,\epsilon}^{n-p}(\mu_{\theta})],\widetilde{F}_{\theta,\epsilon}[F_{\theta,\epsilon}^{n-p}(\mu_{\theta})]\}(\varphi).
$$
This can be decomposed further into the sum of 
\begin{equation}
\widetilde{G}_{\theta}^{p-1}\{F_{\theta}[F_{\theta,\epsilon}^{n-p}(\mu_{\theta})],\widetilde{F}_{\theta}[F_{\theta,\epsilon}^{n-p}(\mu_{\theta})]\}(\varphi)
- 
\widetilde{G}_{\theta}^{p-1}\{F_{\theta,\epsilon}[F_{\theta,\epsilon}^{n-p}(\mu_{\theta})],\widetilde{F}_{\theta}[F_{\theta,\epsilon}^{n-p}(\mu_{\theta})]\}(\varphi)
\label{eq:filt_grad_new1}
\end{equation}
and
\begin{equation}
\widetilde{G}_{\theta}^{p-1}\{F_{\theta,\epsilon}[F_{\theta,\epsilon}^{n-p}(\mu_{\theta})],\widetilde{F}_{\theta}[F_{\theta,\epsilon}^{n-p}(\mu_{\theta})]\}(\varphi) - 
\widetilde{G}_{\theta}^{p-1}\{F_{\theta,\epsilon}[F_{\theta,\epsilon}^{n-p}(\mu_{\theta})],\widetilde{F}_{\theta,\epsilon}[F_{\theta,\epsilon}^{n-p}(\mu_{\theta})]\}(\varphi)
\label{eq:filt_grad_new2}.
\end{equation}
Beginning with \eqref{eq:filt_grad_new1}, by \cite[Lemma 6.7]{vlad}, equation (43) we have
$$
|\widetilde{G}_{\theta}^{p-1}\{F_{\theta}[F_{\theta,\epsilon}^{n-p}(\mu_{\theta})],\widetilde{F}_{\theta}[F_{\theta,\epsilon}^{n-p}(\mu_{\theta})]\}(\varphi)
- 
\widetilde{G}_{\theta}^{p-1}\{F_{\theta,\epsilon}[F_{\theta,\epsilon}^{n-p}(\mu_{\theta})],\widetilde{F}_{\theta}[F_{\theta,\epsilon}^{n-p}(\mu_{\theta})]\}(\varphi)|
$$
$$
\leq C\|\varphi\|_{\infty}\rho^{p-1}\|F_{\theta}[F_{\theta,\epsilon}^{n-p}(\mu_{\theta})]-F_{\theta,\epsilon}[F_{\theta,\epsilon}^{n-p}(\mu_{\theta})]\|
\|\widetilde{F}_{\theta}[F_{\theta,\epsilon}^{n-p}(\mu_{\theta})]\|
$$
where $\rho\in(0,1)$ and $C$ do not depend upon $\mu_{\theta},\epsilon$ or $n,p$.
Applying Lemma \ref{lem:abc_perturbation_filter} we have
$$
|\widetilde{G}_{\theta}^{p-1}\{F_{\theta}[F_{\theta,\epsilon}^{n-p}(\mu_{\theta})],\widetilde{F}_{\theta}[F_{\theta,\epsilon}^{n-p}(\mu_{\theta})]\}(\varphi)
- 
\widetilde{G}_{\theta}^{p-1}\{F_{\theta,\epsilon}[F_{\theta,\epsilon}^{n-p}(\mu_{\theta})],\widetilde{F}_{\theta}[F_{\theta,\epsilon}^{n-p}(\mu_{\theta})]\}(\varphi)|
$$
$$
\leq C\|\varphi\|_{\infty}\rho^{p-1}\epsilon \|\widetilde{F}_{\theta}[F_{\theta,\epsilon}^{n-p}(\mu_{\theta})]\|
$$
where $C$ does not depend upon $\mu_{\theta}$, $\epsilon$ or $n,p$.
Then by Remark \ref{rem:filt_deriv_contr} and Lemma \ref{lem:filter_deriv_upper} $\|\widetilde{F}_{\theta}[F_{\theta,\epsilon}^{n-p}(\mu_{\theta})]\|\leq C(2+\|\widetilde{\mu_{\theta}}\|)$ 
and thus the upper-bound on the $\mathbb{L}_1-$norm of \eqref{eq:filt_grad_new1}:
\begin{equation}
|\widetilde{G}_{\theta}^{p-1}\{F_{\theta}[F_{\theta,\epsilon}^{n-p}(\mu_{\theta})],\widetilde{F}_{\theta}[F_{\theta,\epsilon}^{n-p}(\mu_{\theta})]\}(\varphi)
- 
\widetilde{G}_{\theta}^{p-1}\{F_{\theta,\epsilon}[F_{\theta,\epsilon}^{n-p}(\mu_{\theta})],\widetilde{F}_{\theta}[F_{\theta,\epsilon}^{n-p}(\mu_{\theta})]\}(\varphi)|
\leq C\|\varphi\|_{\infty}\epsilon\rho^{p-1}(2+\|\widetilde{\mu_{\theta}}\|) \label{eq:filt_grad_new3}.
\end{equation}
Now, moving onto  \eqref{eq:filt_grad_new2}, by \cite[Lemma 6.7]{vlad}, equation (42):
$$
|\widetilde{G}_{\theta}^{p-1}\{F_{\theta,\epsilon}[F_{\theta,\epsilon}^{n-p}(\mu_{\theta})],\widetilde{F}_{\theta}[F_{\theta,\epsilon}^{n-p}(\mu_{\theta})]\}(\varphi) - 
\widetilde{G}_{\theta}^{p-1}\{F_{\theta,\epsilon}[F_{\theta,\epsilon}^{n-p}(\mu_{\theta})],\widetilde{F}_{\theta,\epsilon}[F_{\theta,\epsilon}^{n-p}(\mu_{\theta})]\}(\varphi)|
$$
$$
 \leq
C\rho^{p-1}\|\varphi\|_{\infty}\|\widetilde{F}_{\theta}[F_{\theta,\epsilon}^{n-p}(\mu_{\theta})] - \widetilde{F}_{\theta,\epsilon}[F_{\theta,\epsilon}^{n-p}(\mu_{\theta})]\|.
$$
Applying Lemma \ref{lem:abc_filter_deriv_perturb} 
$$
|\widetilde{G}_{\theta}^{p-1}\{F_{\theta,\epsilon}[F_{\theta,\epsilon}^{n-p}(\mu_{\theta})],\widetilde{F}_{\theta}[F_{\theta,\epsilon}^{n-p}(\mu_{\theta})]\}(\varphi) - 
\widetilde{G}_{\theta}^{p-1}\{F_{\theta,\epsilon}[F_{\theta,\epsilon}^{n-p}(\mu_{\theta})],\widetilde{F}_{\theta,\epsilon}[F_{\theta,\epsilon}^{n-p}(\mu_{\theta})]\}(\varphi)|
$$
$$
\leq
C\|\varphi\|_{\infty}\epsilon\rho^{p-1}(1+\|\widetilde{F}_{\theta,\epsilon}^{n-p}(\mu_{\theta})\|).
$$
Then by Lemma \ref{lem:filter_deriv_upper}, we deduce that
\begin{equation}
|\widetilde{G}_{\theta}^{p-1}\{F_{\theta,\epsilon}[F_{\theta,\epsilon}^{n-p}(\mu_{\theta})],\widetilde{F}_{\theta}[F_{\theta,\epsilon}^{n-p}(\mu_{\theta})]\}(\varphi) - 
\widetilde{G}_{\theta}^{p-1}\{F_{\theta,\epsilon}[F_{\theta,\epsilon}^{n-p}(\mu_{\theta})],\widetilde{F}_{\theta,\epsilon}[F_{\theta,\epsilon}^{n-p}(\mu_{\theta})]\}(\varphi)|
 \leq 
C\|\varphi\|_{\infty}\epsilon \rho^{p-1} (2 + \|\widetilde{\mu_{\theta}}\|)\label{eq:filt_grad_new4}.
\end{equation}
Combining \eqref{eq:filt_grad_new3} and \eqref{eq:filt_grad_new4}
\begin{equation}
|\widetilde{G}_{\theta}^{p-1}\{F_{\theta}[F_{\theta,\epsilon}^{n-p}(\mu_{\theta})],\widetilde{F}_{\theta}[F_{\theta,\epsilon}^{n-p}(\mu_{\theta})]\}(\varphi) - 
\widetilde{G}_{\theta}^{p-1}\{F_{\theta,\epsilon}[F_{\theta,\epsilon}^{n-p}(\mu_{\theta})],\widetilde{F}_{\theta,\epsilon}[F_{\theta,\epsilon}^{n-p}(\mu_{\theta})]\}(\varphi)|
\leq C\|\varphi\|_{\infty}\epsilon \rho^{p-1}(2 + \|\widetilde{\mu_{\theta}}\|) \label{eq:filt_grad_new5}.
\end{equation}

We now consider the summands over $q$ in the second and third lines of \eqref{eq:main_decomp_filter_+1}. Again, adopting the compact notation above we can
decompose the summands over $q$ into the sum of
\begin{equation}
\widetilde{G}_{\theta}^{n-q}\{F_{\theta}^{s}[F_{\theta}(F_{\theta,\epsilon}^{n-p}(\mu_{\theta}))],\widetilde{H}_{\theta}^{s}[F_{\theta}(F_{\theta,\epsilon}^{n-p}(\mu_{\theta}))]\}(\varphi) - 
\widetilde{G}_{\theta}^{n-q}\{F_{\theta}^s[ F_{\theta,\epsilon}(F_{\theta,\epsilon}^{n-p}(\mu_{\theta}))],\widetilde{H}_{\theta}^{s}[F_{\theta}(F_{\theta,\epsilon}^{n-p}(\mu_{\theta}))]\}(\varphi)
\label{eq:g1_filt_deriv}
\end{equation}
and
\begin{equation}
\widetilde{G}_{\theta}^{n-q}\{F_{\theta}^s[ F_{\theta,\epsilon}(F_{\theta,\epsilon}^{n-p}(\mu_{\theta}))],\widetilde{H}_{\theta}^{s}[F_{\theta}(F_{\theta,\epsilon}^{n-p}(\mu_{\theta}))]\}(\varphi)
 - 
\widetilde{G}_{\theta}^{n-q}\{F_{\theta}^s[ F_{\theta,\epsilon}(F_{\theta,\epsilon}^{n-p}(\mu_{\theta}))],\widetilde{H}_{\theta}^s[F_{\theta,\epsilon}(F_{\theta,\epsilon}^{n-p}(\mu_{\theta}))]\}(\varphi)
\label{eq:g2_filt_deriv}
\end{equation}
where $s=q-n+p-1$.
We start with \eqref{eq:g1_filt_deriv}; by \cite[Lemma 6.7]{vlad} equation (43), we have
$$
|\widetilde{G}_{\theta}^{n-q}\{F_{\theta}^{s}[F_{\theta}(F_{\theta,\epsilon}^{n-p}(\mu_{\theta}))],\widetilde{H}_{\theta}^{s}[F_{\theta}(F_{\theta,\epsilon}^{n-p}(\mu_{\theta}))]\}(\varphi) - 
\widetilde{G}_{\theta}^{n-q}\{F_{\theta}^s[ F_{\theta,\epsilon}(F_{\theta,\epsilon}^{n-p}(\mu_{\theta}))],\widetilde{H}_{\theta}^{s}[F_{\theta}(F_{\theta,\epsilon}^{n-p}(\mu_{\theta}))]\}(\varphi)|
$$
$$
\leq 
C\|\varphi\|_{\infty}\rho^{n-q}\|F_{\theta}^{s}[F_{\theta}(F_{\theta,\epsilon}^{n-p}(\mu_{\theta}))]-F_{\theta}^s[ F_{\theta,\epsilon}(F_{\theta,\epsilon}^{n-p}(\mu_{\theta}))]\|
\|\widetilde{H}_{\theta}^{s}[F_{\theta}(F_{\theta,\epsilon}^{n-p}(\mu_{\theta}))\|.
$$
Then we will use the stability of the filter (e.g.~\cite[Theorem 3.1]{vlad})
$$
\|F_{\theta}^{s}[F_{\theta}(F_{\theta,\epsilon}^{n-p}(\mu_{\theta}))]-F_{\theta}^s[ F_{\theta,\epsilon}(F_{\theta,\epsilon}^{n-p}(\mu_{\theta}))]\|
\leq C\rho^s
\|F_{\theta}(F_{\theta,\epsilon}^{n-p}(\mu_{\theta}))-F_{\theta,\epsilon}(F_{\theta,\epsilon}^{n-p}(\mu_{\theta}))\|.
$$
By Lemma \ref{lem:abc_perturbation_filter} $\|F_{\theta}(F_{\theta,\epsilon}^{n-p}(\mu_{\theta}))-F_{\theta,\epsilon}(F_{\theta,\epsilon}^{n-p}(\mu_{\theta}))\|\leq C \epsilon$ and thus 
$$
|\widetilde{G}_{\theta}^{n-q}\{F_{\theta}^{s}[F_{\theta}(F_{\theta,\epsilon}^{n-p}(\mu_{\theta}))],\widetilde{H}_{\theta}^{s}[F_{\theta}(F_{\theta,\epsilon}^{n-p}(\mu_{\theta}))]\}(\varphi) - 
\widetilde{G}_{\theta}^{n-q}\{F_{\theta}^s[ F_{\theta,\epsilon}(F_{\theta,\epsilon}^{n-p}(\mu_{\theta}))],\widetilde{H}_{\theta}^{s}[F_{\theta}(F_{\theta,\epsilon}^{n-p}(\mu_{\theta}))]\}(\varphi)|
$$
$$
\leq C\|\varphi\|_{\infty}\epsilon \rho^{p-1} \|\widetilde{H}_{\theta}^{s}[F_{\theta}(F_{\theta,\epsilon}^{n-p}(\mu_{\theta}))]\|.
$$
By \cite[Lemma 6.8]{vlad} we have $\|\widetilde{H}_{\theta}^{s}[F_{\theta}(F_{\theta,\epsilon}^{n-p}(\mu_{\theta}))]\|\leq C$, where $C$ does not depend upon 
$F_{\theta}(F_{\theta,\epsilon}^{n-p}(\mu_{\theta}))$ or $\epsilon$ and hence
$$
|\widetilde{G}_{\theta}^{n-q}\{F_{\theta}^{s}[F_{\theta}(F_{\theta,\epsilon}^{n-p}(\mu_{\theta}))],\widetilde{H}_{\theta}^{s}[F_{\theta}(F_{\theta,\epsilon}^{n-p}(\mu_{\theta}))]\}(\varphi) - 
\widetilde{G}_{\theta}^{n-q}\{F_{\theta}^s[ F_{\theta,\epsilon}(F_{\theta,\epsilon}^{n-p}(\mu_{\theta}))],\widetilde{H}_{\theta}^{s}[F_{\theta}(F_{\theta,\epsilon}^{n-p}(\mu_{\theta}))]\}(\varphi)|
\leq C\|\varphi\|_{\infty}\epsilon \rho^{p-1}.
$$
Now, turning to \eqref{eq:g2_filt_deriv} and applying \cite[Lemma 6.7]{vlad} (42) we have
$$
|\widetilde{G}_{\theta}^{n-q}\{F_{\theta}^s[ F_{\theta,\epsilon}(F_{\theta,\epsilon}^{n-p}(\mu_{\theta}))],\widetilde{H}_{\theta}^{s}[F_{\theta}(F_{\theta,\epsilon}^{n-p}(\mu_{\theta}))]\}(\varphi)
 - 
\widetilde{G}_{\theta}^{n-q}\{F_{\theta}^s[ F_{\theta,\epsilon}(F_{\theta,\epsilon}^{n-p}(\mu_{\theta}))],\widetilde{H}_{\theta}^s[F_{\theta,\epsilon}(F_{\theta,\epsilon}^{n-p}(\mu_{\theta}))]\}(\varphi)|
$$
\begin{equation}
\leq
C\|\varphi\|_{\infty}\rho^{n-q}\|\widetilde{H}_{\theta}^{s}[F_{\theta}(F_{\theta,\epsilon}^{n-p}(\mu_{\theta}))]-\widetilde{H}_{\theta}^s[F_{\theta,\epsilon}(F_{\theta,\epsilon}^{n-p}(\mu_{\theta}))]\|\label{eq:filt_grad_new_new}.
\end{equation}
Then by \cite[Lemma 6.8]{vlad} we have 
$$
\|\widetilde{H}_{\theta}^{s}[F_{\theta}(F_{\theta,\epsilon}^{n-p}(\mu_{\theta}))]-\widetilde{H}_{\theta}^s[F_{\theta,\epsilon}(F_{\theta,\epsilon}^{n-p}(\mu_{\theta}))]\| \leq C\rho^{s}
\|F_{\theta}(F_{\theta,\epsilon}^{n-p})(\mu_{\theta}) - F_{\theta,\epsilon}(F_{\theta,\epsilon}^{n-p}(\mu_{\theta}))\|
$$
and then on applying Lemma \ref{lem:abc_perturbation_filter} we thus have that
$$
\|\widetilde{H}_{\theta}^q(F_{\theta}(F_{\theta,\epsilon}^{n-p})(\mu_{\theta})) - \widetilde{H}_{\theta}^q(F_{\theta,\epsilon}^{n-p+1})(\mu_{\theta})\| \leq C
\epsilon \rho^{s}.
$$
Returning to \eqref{eq:filt_grad_new_new}, it follows by the above calculations that:
$$
|\widetilde{G}_{\theta}^{n-q}\{F_{\theta}^{s}[F_{\theta}(F_{\theta,\epsilon}^{n-p}(\mu_{\theta}))],\widetilde{H}_{\theta}^{s}[F_{\theta}(F_{\theta,\epsilon}^{n-p}(\mu_{\theta}))]\}(\varphi) - 
\widetilde{G}_{\theta}^{n-q}\{F_{\theta}^s[ F_{\theta,\epsilon}(F_{\theta,\epsilon}^{n-p}(\mu_{\theta}))],\widetilde{H}_{\theta}^{s}[F_{\theta}(F_{\theta,\epsilon}^{n-p}(\mu_{\theta}))]\}(\varphi)|
\leq C\|\varphi\|_{\infty}\epsilon \rho^{p-1}.
$$
Thus we have proved that
\begin{equation}
|\widetilde{G}_{\theta}^{n-q}\{F_{\theta}^{s}[F_{\theta}(F_{\theta,\epsilon}^{n-p}(\mu_{\theta}))],\widetilde{H}_{\theta}^{s}[F_{\theta}(F_{\theta,\epsilon}^{n-p}(\mu_{\theta}))]\}(\varphi)-
\widetilde{G}_{\theta}^{n-q}\{F_{\theta}^s[ F_{\theta,\epsilon}(F_{\theta,\epsilon}^{n-p}(\mu_{\theta}))],\widetilde{H}_{\theta}^s[F_{\theta,\epsilon}(F_{\theta,\epsilon}^{n-p}(\mu_{\theta}))]\}(\varphi)|| \leq C\|\varphi\|_{\infty}\epsilon \rho^{p-1}\label{eq:filt_grad_new_new1}.
\end{equation}

Then, returning to \eqref{eq:main_decomp_filter_+1} and noting \eqref{eq:filt_grad_new5}, \eqref{eq:filt_grad_new_new1} we have the upper-bound
$$
\|\widetilde{F}_{\theta}^n(\mu_{\theta},\widetilde{\mu_{\theta}})-\widetilde{F}_{\theta,\epsilon}^n(\mu_{\theta},\widetilde{\mu_{\theta}})\| \leq C\epsilon(2+\|\widetilde{\mu_{\theta}}\|)\sum_{p=1}^n [\rho^{p-1} + \sum_{q=n-p}^{n}\rho^{p-1}]\leq C\epsilon(2+\|\widetilde{\mu_{\theta}}\|).
$$
\end{proof}

\subsection{Technical Results for ABC Bias of the Filter-Derivative}

\begin{lem}\label{lem:abc_filter_deriv_perturb}
Assume (A1-5). Then there exist a $C<+\infty$ such that for any $n\geq 1$, $\mu_{\theta}\in\mathcal{P}(\mathsf{X})$, $\widetilde{\mu_{\theta}}\in\mathcal{M}(\mathsf{X})$, $\epsilon>0$
$\theta\in\Theta$:
$$
\|\widetilde{F}_{\theta}^{(n)}(\mu_{\theta},\widetilde{\mu_{\theta}}) -\widetilde{F}_{\theta,\epsilon}^{(n)}(\mu_{\theta},\widetilde{\mu_{\theta}})\|\leq C\epsilon(1+\|\widetilde{\mu_{\theta}}\|).
$$
\end{lem}

\begin{proof}
By \cite[Lemma 6.7]{vlad} we have the decomposition, for $\varphi\in\mathcal{B}_b(\mathsf{X})$:
$$
\widetilde{F}_{\theta}^{(n)}(\mu_{\theta},\widetilde{\mu_{\theta}})(\varphi) = \widetilde{G}_{\theta}^{(n)}(\mu_{\theta},\widetilde{\mu_{\theta}})(\varphi) - \widetilde{H}_{\theta}^{(n)}(\mu_{\theta})(\varphi)
$$
where 
$$
\widetilde{H}^{(n)}(\mu_{\theta})(\varphi) := 
\mu_{\theta} R_{n,\theta}(1)^{-1}[\mu_{\theta}\widetilde{R}_{n,\theta}(\varphi) - 
\mu_{\theta}\widetilde{R}_{n,\theta}(1) \mu_{\theta}(\varphi).
$$
Thus to control the difference, we can consider the two differences $\widetilde{G}_{\theta}^{(n)}(\mu_{\theta},\widetilde{\mu_{\theta}})(\varphi) - \widetilde{G}_{\theta,\epsilon}^{(n)}(\mu_{\theta},\widetilde{\mu_{\theta}})(\varphi)$ 
and $\widetilde{H}_{\theta}^{(n)}(\mu_{\theta})(\varphi)-\widetilde{H}_{\theta,\epsilon}^{(n)}(\mu_{\theta})(\varphi)$.

\textbf{Control of} $\widetilde{G}_{\theta}^{(n)}(\mu_{\theta},\widetilde{\mu_{\theta}})(\varphi) - \widetilde{G}_{\theta,\epsilon}^{(n)}(\mu_{\theta},\widetilde{\mu_{\theta}})(\varphi)$. We will use the Hahn-Jordan decomposition: 
$\widetilde{\mu_{\theta}}=\widetilde{\mu_{\theta}}^+-\widetilde{\mu_{\theta}}^-$. It is assumed that both $\widetilde{\mu_{\theta}}^+(1),\widetilde{\mu_{\theta}}^-(1)>0$. The scenario with either 
$\widetilde{\mu_{\theta}}^+(1)=0$ or  $\widetilde{\mu_{\theta}}^+(1)=0$ is straightforward and omitted for brevity.
We can write:
$$
\widetilde{G}_{\theta}^{(n)}(\mu_{\theta},\widetilde{\mu_{\theta}})(\varphi) = \frac{\widetilde{\mu_{\theta}}^+ R_{n,\theta}(1)}{\mu_{\theta} R_{n,\theta}(1)}[ F_{\theta}^{(n)}(\widetilde{\bar{\mu_{\theta}}}^+)(\varphi) - F_{\theta}^{(n)}(\mu_{\theta})(\varphi)]  + 
\frac{\widetilde{\mu_{\theta}}^- R_{n,\theta}(1)}{\mu_{\theta} R_{n,\theta}(1)}[ F_{\theta}^{(n)}(\widetilde{\bar{\mu_{\theta}}}^-)(\varphi) - F_{\theta}^{(n)}(\mu_{\theta})(\varphi)] 
$$
where $\widetilde{\bar{\mu_{\theta}}}^+(\cdot)=\widetilde{\mu_{\theta}}^+(\cdot)/\widetilde{\mu_{\theta}}^+(1)$ and $\widetilde{\bar{\mu_{\theta}}}^-(\cdot)=\widetilde{\mu_{\theta}}^-(\cdot)/\widetilde{\mu_{\theta}}^-(1)$. Thus we have
\begin{eqnarray}
\widetilde{G}_{\theta}^{(n)}(\mu_{\theta},\widetilde{\mu_{\theta}})(\varphi) - 
\widetilde{G}_{\theta,\epsilon}^{(n)}(\mu_{\theta},\widetilde{\mu_{\theta}})(\varphi) & = &
\bigg[\frac{\widetilde{\mu_{\theta}}^+R_{n,\theta}(1)}{\mu_{\theta} R_{n,\theta}(1)} - 
\frac{\widetilde{\mu_{\theta}}^+R_{n,\theta,\epsilon}(1)}{\mu_{\theta} R_{n,\theta,\epsilon}(1)}\bigg]
[F_{\theta}^{(n)}(\widetilde{\bar{\mu_{\theta}}}^+)(\varphi) - F_{\theta}^{(n)}(\mu_{\theta})(\varphi)]\nonumber\\
& & + \frac{\widetilde{\mu_{\theta}}^+R_{n,\theta,\epsilon}(1)}{\mu_{\theta} R_{n,\theta,\epsilon}(1)}
[F_{\theta}^{(n)}(\widetilde{\bar{\mu_{\theta}}}^+)(\varphi) - F_{\theta}^{(n)}(\mu_{\theta})(\varphi) - F_{\theta,\epsilon}^{(n)}(\widetilde{\bar{\mu_{\theta}}}^+)(\varphi)
+ F_{\theta,\epsilon}^{(n)}(\mu_{\theta})(\varphi)] \nonumber\\
& &+ 
\bigg[\frac{\widetilde{\mu_{\theta}}^-R_{n,\theta}(1)}{\mu_{\theta} R_{n,\theta}(1)} - 
\frac{\widetilde{\mu_{\theta}}^-R_{n,\theta,\epsilon}(1)}{\mu_{\theta} R_{n,\theta,\epsilon}(1)}\bigg]
[F_{\theta}^{(n)}(\widetilde{\bar{\mu_{\theta}}}^-)(\varphi) - F_{\theta}^{(n)}(\mu_{\theta})(\varphi)]\nonumber\\
& & + \frac{\widetilde{\mu_{\theta}}^-R_{n,\theta,\epsilon}(1)}{\mu_{\theta} R_{n,\theta,\epsilon}(1)}
[F_{\theta}^{(n)}(\widetilde{\bar{\mu_{\theta}}}^-)(\varphi) - F_{\theta}^{(n)}(\mu_{\theta})(\varphi) - F_{\theta,\epsilon}^{(n)}(\widetilde{\bar{\mu_{\theta}}}^-)(\varphi)
+ F_{\theta,\epsilon}^{(n)}(\mu_{\theta})(\varphi)]\label{eq:f-mu-epsilon-main}.
\end{eqnarray}
By symmetry, we need only consider the terms including $\widetilde{\mu_{\theta}}^+$; one can treat those with $\widetilde{\mu_{\theta}}^-$ by using similar arguments. First dealing with term on the first line of the  R.H.S.~of
\eqref{eq:f-mu-epsilon-main}. We have that
$$
\bigg[\frac{\widetilde{\mu_{\theta}}^+R_{n,\theta}(1)}{\mu_{\theta} R_{n,\theta}(1)} - 
\frac{\widetilde{\mu_{\theta}}^+R_{n,\theta,\epsilon}(1)}{\mu_{\theta} R_{n,\theta,\epsilon}(1)}\bigg]
[F_{\theta}^{(n)}(\widetilde{\bar{\mu_{\theta}}}^+)(\varphi) - F_{\theta}^{(n)}(\mu_{\theta})(\varphi)]
 = 
$$
$$
\bigg[\frac{\widetilde{\mu_{\theta}}^+R_{n,\theta}(1)-\widetilde{\mu_{\theta}}^+R_{n,\theta,\epsilon}(1)}{\mu_{\theta} R_{n,\theta}(1)} + 
\widetilde{\mu_{\theta}}^+R_{n,\theta,\epsilon}(1)\frac{\mu_{\theta} R_{n,\theta,\epsilon}(1) - \mu_{\theta} R_{n,\theta}(1)}{\mu_{\theta} R_{n,\theta,\epsilon}(1)\mu_{\theta} R_{n,\theta}(1)}
\bigg]
[F_{\theta}^{(n)}(\widetilde{\bar{\mu_{\theta}}}^+)(\varphi) - F_{\theta}^{(n)}(\mu_{\theta})(\varphi)]
$$
Now by (A\ref{hyp:like_cont}), for any $n$
\begin{equation}
\sup_{x\in\mathsf{X}}|R_{n,\theta}(1)(x) - 
R_{n,\theta,\epsilon}(1)(x)| \leq
C\epsilon
\label{eq:r_ineq}
\end{equation}
thus
$$
\bigg[\frac{\widetilde{\mu_{\theta}}^+R_{n,\theta}(1)-\widetilde{\mu_{\theta}}^+R_{n,\theta,\epsilon}(1)}{\mu_{\theta} R_{n,\theta}(1)} + 
\widetilde{\mu_{\theta}}^+R_{n,\theta,\epsilon}(1)\frac{\mu_{\theta} R_{n,\theta,\epsilon}(1) - \mu_{\theta} R_{n,\theta}(1)}{\mu_{\theta} R_{n,\theta,\epsilon}(1)\mu_{\theta} R_{n,\theta}(1)}
\bigg] \leq \frac{C\epsilon\widetilde{\mu_{\theta}}^+(1)}{\mu_{\theta} R_{n,\theta}(1)} + C\epsilon\frac{\widetilde{\mu_{\theta}}^+R_{n,\theta,\epsilon}(1)}{\mu_{\theta} R_{n,\theta,\epsilon}(1)\mu_{\theta} R_{n,\theta}(1)}.
$$
Now one can show that there exist a $C<+\infty$ such that for any $x,y\in\mathsf{X}$
\begin{equation}
R_{n,\theta}(1)(x) \geq CR_{n,\theta}(1)(y)\quad\quad R_{n,\theta,\epsilon}(1)(x) \geq C R_{n,\theta,\epsilon}(1)(y).
\label{eq:r_minor}
\end{equation}
Then it follows that
$$
\frac{C\epsilon\widetilde{\mu_{\theta}}^+(1)}{\mu_{\theta} R_{n,\theta}(1)} + C\epsilon\frac{\widetilde{\mu_{\theta}}^+R_{n,\theta,\epsilon}(1)}{\mu_{\theta} R_{n,\theta,\epsilon}(1)\mu_{\theta} R_{n,\theta}(1)}
\leq 
C\epsilon\widetilde{\mu_{\theta}}^+(1).
$$
Hence we have shown that
$$
\bigg[\frac{\widetilde{\mu_{\theta}}^+R_{n,\theta}(1)}{\mu_{\theta} R_{n,\theta}(1)} - 
\frac{\widetilde{\bar{\mu_{\theta}}}^+R_{n,\theta,\epsilon}(1)}{\mu_{\theta} R_{n,\theta,\epsilon}(1)}\bigg]
[F_{\theta}^{(n)}(\widetilde{\mu_{\theta}}^+)(\varphi) - F_{\theta}^{(n)}(\mu_{\theta})(\varphi)]
\leq
C\|\varphi\|_{\infty}\epsilon\widetilde{\mu_{\theta}}^+(1).
$$
Second, the second line of the R.H.S.~of
\eqref{eq:f-mu-epsilon-main}. By Lemma \ref{lem:abc_perturbation_filter}, for any $\mu_{\theta}\in\mathcal{P}(\mathsf{X})$, $\|F_{\theta}^{(n)}(\mu_{\theta})-F_{\theta,\epsilon}^{(n)}(\mu_{\theta})\|\leq C\epsilon$, with
$C$ independent of $\mu_{\theta}$, and in addition using 
\eqref{eq:r_minor} we have
$$
\frac{\widetilde{\mu_{\theta}}^+R_{n,\theta,\epsilon}(1)}{\mu_{\theta} R_{n,\theta,\epsilon}(1)}
[F_{\theta}^{(n)}(\widetilde{\bar{\mu_{\theta}}}^+)(\varphi) - F_{\theta}^{(n)}(\mu_{\theta})(\varphi) - F_{\theta,\epsilon}^{(n)}(\widetilde{\bar{\mu_{\theta}}}^+)(\varphi)
+ F_{\theta,\epsilon}^{(n)}(\mu_{\theta})(\varphi)] \leq C\|\varphi\|_{\infty}\epsilon \widetilde{\mu_{\theta}}^+(1).
$$
Thus we have shown:
\begin{equation}
\|\widetilde{G}_{\theta}^{(n)}(\mu_{\theta},\widetilde{\mu_{\theta}})(\varphi) - \widetilde{G}_{\theta,\epsilon}^{(n)}(\mu_{\theta},\widetilde{\mu_{\theta}})(\varphi)\|\leq
C\epsilon[\widetilde{\mu_{\theta}}^+(1) + \widetilde{\mu_{\theta}}^-(1)] = C\epsilon\|\widetilde{\mu_{\theta}}\|.\label{eq:g_cont_ineq}
\end{equation}

\textbf{Control of} $\widetilde{H}_{\theta}^{(n)}(\mu_{\theta})(\varphi)-\widetilde{H}_{\theta,\epsilon}^{(n)}(\mu_{\theta})(\varphi)$. We have
\begin{equation}
\widetilde{H}_{\theta}^{(n)}(\mu_{\theta})(\varphi)-\widetilde{H}_{\theta,\epsilon}^{(n)}(\mu_{\theta})(\varphi) = \bigg[\frac{\mu_{\theta} \widetilde{R}_{n,\theta}(\varphi)}{\mu_{\theta} R_{n,\theta}(1)} - 
\frac{\mu_{\theta} \widetilde{R}_{n,\theta,\epsilon}(\varphi)}{\mu_{\theta} R_{n,\theta,\epsilon}(1)}
\bigg] + \bigg[
\frac{\mu_{\theta} \widetilde{R}_{n,\theta,\epsilon}(1) F_{\theta,\epsilon}^{(n)}(\mu_{\theta})(\varphi)}{\mu_{\theta} R_{n,\theta,\epsilon}(1)} - 
\frac{\mu_{\theta} \widetilde{R}_{n,\theta}(1) F_{\theta}^{(n)}(\mu_{\theta})(\varphi)}{\mu_{\theta} R_{n,\theta}(1)}
\bigg]\label{eq:h_main_decomp}.
\end{equation}
We start with the first bracket on the R.H.S.~of \eqref{eq:h_main_decomp}. We first note that
\begin{equation}
\widetilde{R}_{n,\theta}(\varphi)(x)-\widetilde{R}_{n,\theta,\epsilon}(\varphi)(x) = 
\int f_{\theta}(x'|x)\varphi(x')[\nabla g_{\theta}(y_n|x')-\nabla g_{\theta,\epsilon}(y_n|x')] dx' \leq C \|\varphi\|_{\infty}\epsilon
\label{eq:rtilde_ineq}
\end{equation}
where we have applied \eqref{eq:grad_g_cont}. Then we have
$$
\frac{\mu_{\theta} \widetilde{R}_{n,\theta}(\varphi)}{\mu_{\theta} R_{n,\theta}(1)} - 
\frac{\mu_{\theta} \widetilde{R}_{n,\theta,\epsilon}(\varphi)}{\mu_{\theta} R_{n,\theta,\epsilon}(1)} = 
\frac{\mu_{\theta} \widetilde{R}_{n,\theta}(\varphi)-\mu_{\theta}\widetilde{R}_{n,\theta,\epsilon}(\varphi)}{\mu_{\theta} R_{n,\theta}(1)}
+ \mu_{\theta} \widetilde{R}_{n,\theta,\epsilon}(\varphi)
\frac{\mu_{\theta} R_{n,\theta,\epsilon}(1)-\mu_{\theta} R_{n,\theta}(1)}{\mu_{\theta} R_{n,\theta,\epsilon}(1)\mu_{\theta} R_{n,\theta}(1)}.
$$
By using \eqref{eq:rtilde_ineq} on the first term on the R.H.S.~of the above equation and by using \eqref{eq:r_ineq} in the numerator
for the second, along with \eqref{eq:r_minor} in the denominator, we have
$$
\bigg|\frac{\mu_{\theta} \widetilde{R}_{n,\theta}(\varphi)}{\mu_{\theta} R_{n,\theta}(1)} - 
\frac{\mu_{\theta} \widetilde{R}_{n,\theta,\epsilon}(\varphi)}{\mu_{\theta} R_{n,\theta,\epsilon}(1)}\bigg| \leq
C\epsilon [\|\varphi\|_{\infty} + |\mu_{\theta} \widetilde{R}_{n,\theta,\epsilon}(\varphi)|].
$$
Then as
\begin{equation}
\widetilde{R}_{n,\theta,\epsilon}(\varphi)(x) = \int \varphi(x')[\nabla g_{\theta,\epsilon}(y_n|x') f_{\theta}(x'|x)
-
g_{\theta,\epsilon}(y_n|x) \nabla f_{\theta}(x'|x)]dx' \leq C \|\varphi\|_{\infty}\int_{\mathsf{X}}dx' \leq C \|\varphi\|_{\infty}
\label{eq:bound_r_tilde}
\end{equation}
where the compactness of $\mathsf{X}$ and (A\ref{hyp:like_grad_bound}) have been used, we have the upper-bound
\begin{equation}
\bigg|\frac{\mu_{\theta} \widetilde{R}_{n,\theta}(\varphi)}{\mu_{\theta} R_{n,\theta}(1)} - 
\frac{\mu_{\theta} \widetilde{R}_{n,\theta,\epsilon}(\varphi)}{\mu_{\theta} R_{n,\theta,\epsilon}(1)}\bigg|
\leq C\|\varphi\|_{\infty}\epsilon.
\label{eq:h_first_part_ineq}
\end{equation}
Moving onto the second bracket on the R.H.S.~of \eqref{eq:h_main_decomp}, this is equal to
$$
\bigg[\frac{\mu_{\theta} \widetilde{R}_{n,\theta,\epsilon}(1) }{\mu_{\theta} R_{n,\theta,\epsilon}(1)}
-
\frac{\mu_{\theta} \widetilde{R}_{n,\theta}(1) }{\mu_{\theta} R_{n,\theta}(1)}
\bigg]
F_{\theta,\epsilon}^{(n)}(\mu_{\theta})(\varphi) +
\frac{\mu_{\theta} \widetilde{R}_{n,\theta}(1)}{\mu_{\theta} R_{n,\theta}(1)}
[F_{\theta,\epsilon}^{(n)}(\mu_{\theta})(\varphi)-F_{\theta}^{(n)}(\mu_{\theta})(\varphi)]
$$
By using the inequality \eqref{eq:h_first_part_ineq}, we have
$$
\bigg[\frac{\mu_{\theta} \widetilde{R}_{n,\theta,\epsilon}(1) }{\mu_{\theta} R_{n,\theta,\epsilon}(1)}
-
\frac{\mu_{\theta} \widetilde{R}_{n,\theta}(1) }{\mu_{\theta} R_{n,\theta}(1)}
\bigg]
F_{\theta,\epsilon}^{(n)}(\mu_{\theta})(\varphi)
\leq C\epsilon |F_{\theta,\epsilon}^{(n)}(\mu_{\theta})(\varphi)| \leq C\|\varphi\|_{\infty}\epsilon.
$$
Using Lemma \ref{lem:abc_perturbation_filter} and in addition using 
\eqref{eq:r_minor} in the denominator and \eqref{eq:bound_r_tilde} in the numerator we have
$$
\frac{\mu_{\theta} \widetilde{R}_{n,\theta}(1)}{\mu_{\theta} R_{n,\theta}(1)}
[F_{\theta,\epsilon}^{(n)}(\mu_{\theta})(\varphi)-F_{\theta}^{(n)}(\mu_{\theta})(\varphi)]
\leq C\|\varphi\|_{\infty}\epsilon
$$
where $C$ does not depend upon $\mu_{\theta}$ and $\epsilon$. Thus we have established that
\begin{equation}
\frac{\mu_{\theta} \widetilde{R}_{n,\theta,\epsilon}(1) F_{\theta,\epsilon}^{(n)}(\mu_{\theta})(\varphi)}{\mu_{\theta} R_{n,\theta,\epsilon}(1)} - 
\frac{\mu_{\theta} \widetilde{R}_{n,\theta}(1) F_{\theta}^{(n)}(\mu_{\theta})(\varphi)}{\mu_{\theta} R_{n,\theta}(1)}
\leq C\|\varphi\|_{\infty}\epsilon
\label{eq:h_second_part_ineq}.
\end{equation}
One can put together the results of \eqref{eq:h_first_part_ineq} and
\eqref{eq:h_second_part_ineq} and establish that
\begin{equation}
|\widetilde{H}_{\theta}^{(n)}(\mu_{\theta})(\varphi)-\widetilde{H}_{\theta,\epsilon}^{(n)}(\mu_{\theta})(\varphi)|
\leq C\|\varphi\|_{\infty}\epsilon
\label{eq:h_control}.
\end{equation}
On combining the results \eqref{eq:g_cont_ineq} and \eqref{eq:h_control} and noting \eqref{eq:h_main_decomp} we conclude the proof.
\end{proof}

\begin{lem}\label{lem:abc_perturbation_filter}
Assume (A1-3). Then there exist a $C<+\infty$ such that for any $n\geq 1$, $\mu_{\theta}\in\mathcal{P}(\mathsf{X})$, $\epsilon>0$,
$\theta\in\Theta$:
$$
\|F_{\theta}^{(n)}(\mu_{\theta})  - F_{\theta,\epsilon}^{(n)}(\mu_{\theta})\|\leq C\epsilon.
$$
\end{lem}

\begin{proof}
For $\varphi\in\mathcal{B}_b(\varphi)$
$$
F_{\theta}^{(n)}(\mu_{\theta})(\varphi)  - F_{\theta,\epsilon}^{(n)}(\mu_{\theta})(\varphi) = 
\frac{\mu_{\theta} R_{n,\theta}(\varphi) - \mu_{\theta} R_{n,\theta,\epsilon}(\varphi)}{\mu_{\theta} R_{n,\theta}(1)} 
+ \mu_{\theta} R_{n,\theta,\epsilon}(\varphi)\bigg[
\frac{\mu_{\theta} R_{n,\theta,\epsilon}(1)-\mu_{\theta} R_{n,\theta}(1)}{\mu_{\theta} R_{n,\theta,\epsilon}(1)
\mu_{\theta} R_{n,\theta}(1)}
\bigg].
$$
Then by applying \eqref{eq:r_ineq} on both terms on the R.H.S.~we have the upper-bound
$$
\frac{C\|\varphi\|_{\infty}\epsilon}{\mu_{\theta} R_{n,\theta}(1)}.
$$
One can conclude by using the inequality \eqref{eq:r_minor} for $R_{n,\theta}(1)(\cdot)$.
\end{proof}

\begin{lem}\label{lem:filter_deriv_upper}
Assume (A1-5). Then there exist a $C<+\infty$ such that for any $n\geq 1$, $\mu_{\theta}\in\mathcal{P}(\mathsf{X})$, $\widetilde{\mu_{\theta}}\in\mathcal{M}(\mathsf{X})$, $\epsilon>0$,
$\theta\in\Theta$:
$$
\|\widetilde{F}_{\theta}^{n}(\mu_{\theta},\widetilde{\mu_{\theta}})\|\vee
\|\widetilde{F}_{\theta,\epsilon}^{n}(\mu_{\theta},\widetilde{\mu_{\theta}})\|
\leq C(1 + \|\widetilde{\mu_{\theta}}\|).
$$
\end{lem}

\begin{proof}
We will consider only $F_{\theta}^{n}(\mu_{\theta},\widetilde{\mu_{\theta}})$ as the ABC filter derivative will follow similar calculations, for any $\epsilon>0$ (with upper-bounds that are independent of $\epsilon$).
By \cite[Lemma 6.4]{vlad} we have for $\varphi\in\mathcal{B}_b(\mathsf{X})$
$$
\widetilde{F}^n_{\theta}(\mu_{\theta},\widetilde{\mu_{\theta}})(\varphi) = \widetilde{G}^n_{\theta}(\mu_{\theta},\widetilde{\mu_{\theta}})(\varphi) + 
\sum_{p=1}^n \widetilde{G}^{n-p}_{\theta}(F_{\theta}^p(\mu_{\theta}),\widetilde{H}_{\theta}^p(\mu_{\theta}))(\varphi).
$$
By \cite[Lemma 6.6]{vlad} we have the upper-bound
$$
\|\widetilde{F}^n_{\theta}(\mu_{\theta},\widetilde{\mu_{\theta}})\| \leq C \Big(\rho^n\|\widetilde{\mu_{\theta}}\|
+\sum_{p=1}^n \rho^{n-p}\|\widetilde{H}_{\theta}^p(\mu_{\theta})\|
\Big)
$$
with $\rho\in(0,1)$.
Then by \cite[Lemma 6.8]{vlad}, it follows that
$$
\|\widetilde{F}^n_{\theta}(\mu_{\theta},\widetilde{\mu_{\theta}})\| \leq C \Big(\rho^n\|\widetilde{\mu_{\theta}}\|
+\sum_{p=1}^n \rho^{n-p}
\Big)
$$
from which one concludes.
\end{proof}

\begin{rem}\label{rem:filt_deriv_contr}
Using the proof above, one can also show that there exist a $C<+\infty$ such that for
any $n\geq 1$, $\mu_{\theta}\in\mathcal{P}(\mathsf{X})$, $\widetilde{\mu_{\theta}}\in\mathcal{M}(\mathsf{X})$, $\epsilon>0$,
$\theta\in\Theta$
$$
\|\widetilde{F}_{\theta}^{(n)}(\mu_{\theta},\widetilde{\mu_{\theta}})\|\vee
\|\widetilde{F}_{\theta,\epsilon}^{(n)}(\mu_{\theta},\widetilde{\mu_{\theta}})\|
\leq C(1 + \|\widetilde{\mu_{\theta}}\|).
$$
\end{rem}

\end{document}